\documentclass{article}

\def\noheaderplainsetup{

\topmargin=0pt \headheight=0pt \headsep=0pt  \oddsidemargin=0pt \evensidemargin=0pt  \textheight=9.1truein \textwidth=6.5truein}   

\noheaderplainsetup

\usepackage{amsfonts}

\begin{document}

%    LOGICS:

\newcommand{\clthree}{\mbox{\bf CL12}}
\newcommand{\cltw}{\mbox{\bf CL12}}
\newcommand{\clfour}{\mbox{\bf CL4}}
\newcommand{\arfour}{\mbox{\bf CLA4}} 
\newcommand{\pa}{\mbox{\bf PA}} 
\newcommand{\bound}{\mathfrak{b}} 

%     MISC.:

\newcommand{\intimpl}{\mbox{\hspace{2pt}$\circ$\hspace{-0.14cm} \raisebox{-0.043cm}{\Large --}\hspace{2pt}}}
\newcommand{\zero}{\mbox{\small {\bf 0}}}
\newcommand{\izero}{\mbox{\scriptsize {\bf 0}}}
\newcommand{\one}{\mbox{\small {\bf 1}}}
\newcommand{\ione}{\mbox{\scriptsize {\bf 1}}}
\newcommand{\plus}{\mbox{\hspace{1pt}\raisebox{0.05cm}{\tiny\boldmath $+$}\hspace{1pt}}}
\newcommand{\iplus}{\mbox{\raisebox{0.03cm}{\tiny $+$}}}
\newcommand{\minus}{\mbox{\hspace{1pt}\raisebox{0.05cm}{\tiny\boldmath $-$}\hspace{1pt}}}
\newcommand{\iminus}{\mbox{\raisebox{0.03cm}{\tiny $-$}}}
\newcommand{\mult}{\mbox{\hspace{1pt}\raisebox{0.05cm}{\tiny\boldmath $\times$}\hspace{1pt}}}
\newcommand{\imult}{\mbox{\raisebox{0.03cm}{\tiny $\times$}}}
\newcommand{\equals}{\mbox{\hspace{1pt}\raisebox{0.05cm}{\tiny\boldmath $=$}\hspace{1pt}}}
\newcommand{\notequals}{\mbox{\hspace{1pt}\raisebox{0.05cm}{\tiny\boldmath $\not=$}\hspace{1pt}}}
\newcommand{\successor}{\mbox{\hspace{1pt}\boldmath $'$}}

\newcommand{\mless}{\mbox{\hspace{1pt}\raisebox{0.05cm}{\tiny\boldmath $<$}\hspace{1pt}}}
\newcommand{\mgreater}{\mbox{\hspace{1pt}\raisebox{0.05cm}{\tiny\boldmath $>$}\hspace{1pt}}}
\newcommand{\mleq}{\mbox{\hspace{1pt}\raisebox{0.05cm}{\tiny\boldmath $\leq$}\hspace{1pt}}}
\newcommand{\mgeq}{\mbox{\hspace{1pt}\raisebox{0.05cm}{\tiny\boldmath $\geq$}\hspace{1pt}}}

\newcommand{\elz}[1]{\mbox{$\parallel\hspace{-3pt} #1 \hspace{-3pt}\parallel$}} 
\newcommand{\elzi}[1]{\mbox{\scriptsize $\parallel\hspace{-3pt} #1 \hspace{-3pt}\parallel$}}
\newcommand{\emptyrun}{\langle\rangle} 
\newcommand{\oo}{\bot}            
\newcommand{\pp}{\top}            
\newcommand{\xx}{\wp}               
\newcommand{\legal}[2]{\mbox{\bf Lr}^{#1}_{#2}} 
\newcommand{\win}[2]{\mbox{\bf Wn}^{#1}_{#2}} 
\newcommand{\seq}[1]{\langle #1 \rangle} 
\newcommand{\code}[1]{\ulcorner #1 \urcorner}          

%     OPERATORS:

\newcommand{\pst}{\mbox{\raisebox{-0.01cm}{\scriptsize $\wedge$}\hspace{-4pt}\raisebox{0.16cm}{\tiny $\mid$}\hspace{2pt}}}
\newcommand{\pcost}{\mbox{\raisebox{0.12cm}{\scriptsize $\vee$}\hspace{-4pt}\raisebox{0.02cm}{\tiny $\mid$}\hspace{2pt}}}

\newcommand{\gneg}{\mbox{\small $\neg$}}                  %game negation
\newcommand{\mli}{\hspace{2pt}\mbox{\small $\rightarrow$}\hspace{2pt}}                      %strong reduction
\newcommand{\cla}{\mbox{$\forall$}}      %blind universal quantifier
\newcommand{\cle}{\mbox{$\exists$}}        %blind existential quantifier
\newcommand{\mld}{\hspace{2pt}\mbox{\small $\vee$}\hspace{2pt}}     %multiplicative disjunction
\newcommand{\mlc}{\hspace{2pt}\mbox{\small $\wedge$}\hspace{2pt}}   %multiplicative conjunction
\newcommand{\mlci}{\hspace{2pt}\mbox{\footnotesize $\wedge$}\hspace{2pt}}   %multiplicative conjunction
\newcommand{\ade}{\mbox{\large $\sqcup$}}      %additive existential quantifier
\newcommand{\ada}{\mbox{\large $\sqcap$}}      %additive universal quantifier
\newcommand{\add}{\hspace{2pt}\mbox{\small $\sqcup$}\hspace{2pt}}                     %additive disjunction
\newcommand{\adc}{\hspace{2pt}\mbox{\small $\sqcap$}\hspace{2pt}} 
\newcommand{\adci}{\hspace{2pt}\mbox{\footnotesize $\sqcap$}\hspace{2pt}}              %index additive conjunction
\newcommand{\clai}{\forall}     %index blind universal quantifier
\newcommand{\clei}{\exists}        %index blind existential quantifier
\newcommand{\tlg}{\bot}               %classical \bot; trivially lost elementary game
\newcommand{\twg}{\top}               %classical \top; trivially won elementary game
\newcommand{\fintimpl}{\mbox{\hspace{2pt}$\bullet$\hspace{-0.14cm} \raisebox{-0.058cm}{\Large --}\hspace{-6pt}\raisebox{0.008cm}{\scriptsize $\wr$}\hspace{-1pt}\raisebox{0.008cm}{\scriptsize $\wr$}\hspace{4pt}}}
\newcommand{\col}[1]{\mbox{$#1$:}}

%   NUMERATED ITEMS and ENVIRONMENTS

\newtheorem{theoremm}{Theorem}[section]
\newtheorem{factt}[theoremm]{Fact}
\newtheorem{definitionn}[theoremm]{Definition}
\newtheorem{lemmaa}[theoremm]{Lemma}
\newtheorem{propositionn}[theoremm]{Proposition}
\newtheorem{conventionn}[theoremm]{Convention}
\newtheorem{examplee}[theoremm]{Example}
\newtheorem{exercisee}[theoremm]{Exercise}
\newtheorem{thesiss}[theoremm]{Thesis}
\newenvironment{definition}{\begin{definitionn} \em}{ \end{definitionn}}
\newenvironment{theorem}{\begin{theoremm}}{\end{theoremm}}
\newenvironment{lemma}{\begin{lemmaa}}{\end{lemmaa}}
\newenvironment{fact}{\begin{factt}}{\end{factt}}
\newenvironment{proposition}{\begin{propositionn} }{\end{propositionn}}
\newenvironment{convention}{\begin{conventionn} \em}{\end{conventionn}}
\newenvironment{example}{\begin{examplee} \em}{\end{examplee}}
\newenvironment{thesis}{\begin{thesiss} \em}{\end{thesiss}}
\newenvironment{exercise}{\begin{exercisee} \em}{\end{exercisee}}
\newenvironment{proof}{ {\bf Proof.} }{\  \rule{2.5mm}{2.5mm} \vspace{.2in} }
\newenvironment{idea}{ {\bf Proof idea.} }{\  \rule{1.5mm}{1.5mm} \vspace{.15in} }
\newenvironment{subproof}{ {\em Proof.} }{\  \rule{2mm}{2mm} \vspace{.1in} }

\title{Introduction to clarithmetic I}
\author{Giorgi Japaridze}

\date{}
\maketitle

\begin{abstract} ``{\em Clarithmetic}'' is a  generic name for formal number theories similar to  Peano arithmetic, but based  on    {\em computability logic}   instead of the more traditional classical or intuitionistic logics. Formulas of clarithmetical theories represent interactive computational problems, and their ``truth'' is understood as existence of an algorithmic solution. Imposing various complexity constraints on such solutions yields various versions of clarithmetic. The present paper introduces a system   of clarithmetic for polynomial time computability, which is shown to be sound and complete. Sound     in the sense that every theorem $T$ of the system represents an interactive  number-theoretic computational problem with a polynomial time solution and, furthermore, such a solution can be efficiently  extracted from a proof of $T$. And complete in the sense that every interactive number-theoretic problem with a polynomial time solution is represented by some theorem $T$ of the system. The paper is written in a semitutorial style and targets readers with no prior familiarity with computability logic.  
\end{abstract}


\


%\tableofcontents

\section{Introduction}\label{intr}
%\marginpar{intr}

{\em Computability logic} (CoL), introduced in \cite{Jap03,Japic,Japfin}, is a semantical, mathematical and philosophical platform, and a long-term program, for redeveloping logic as a formal theory of computability, as opposed to the formal theory of truth which logic has more traditionally been. 
 Under the approach of CoL, formulas represent computational problems, 
and their ``truth'' is seen as algorithmic solvability. In turn, computational problems --- understood in their  most general, {\em interactive} sense --- are defined as games played by a machine against its environment, with ``algorithmic solvability'' meaning existence of a machine that wins the game against any possible behavior of the environment. And an open-ended collection of the most basic and natural operations 
on computational problems forms the logical vocabulary of the theory.  With this semantics, CoL provides a systematic answer to the fundamental question ``{\em what can be computed?}\hspace{1pt}'', just as classical logic is a systematic tool for telling what is true. Furthermore, as it turns out, in positive cases ``{\em what} can be computed'' always allows itself to be replaced by ``{\em how} can be computed'', which makes CoL of potential interest in not only theoretical computer science, but many  applied areas as well, including interactive knowledge base systems, resource oriented systems for planning and action, or declarative programming languages. 

While potential applications have been repeatedly pointed out in earlier papers on CoL, so far all technical efforts had been mainly focused on finding axiomatizations for various fragments of this semantically conceived and inordinately expressive logic. Considerable advances have already been made in this direction (\cite{Japtocl1}-\cite{Cirq}, \cite{Japtcs}-\cite{Japlbcs},  \cite{Ver}), and more results in the same style are probably still to come. It should be however remembered that the main value of CoL, or anything else claiming to be a ``Logic'' with a capital ``L'', will  eventually be determined by whether and how it relates to the outside, extra-logical world. In this respect, unlike many  other systems officially classified as ``logics'', the merits of classical logic  are obvious, most eloquently demonstrated by the fact that applied formal theories, a model example of which is {\em Peano arithmetic} $\pa$,\label{iPA} can be and have been successfully based on it. Unlike pure logics with their meaningless symbols, such theories are direct tools for studying and navigating the real world with its non-man-made, meaningful objects, such as natural numbers in the case of arithmetic. To make this point more clear to a computer scientist, one could compare a pure  logic  with a programming language, and applied theories based on it with application programs written in that language. A programming language created for its own sake, mathematically or esthetically appealing  but otherwise unusable as a general-purpose, comprehensive basis for application  programs, would hardly be of much interest.  

So, in parallel with studying possible axiomatizations and various metaproperties of pure CoL, it would certainly be worthwhile to devote some efforts to justifying its right on existence through revealing its power and appeal as a basis for applied systems. First and so far the only concrete steps in this direction have been made very recently in \cite{Japtowards}, where a CoL-based system {\bf CLA1} of (Peano) arithmetic was constructed.\footnote{The paper \cite{Xu} (in Chinese) is apparently another exception, focused on applications of CoL in AI.} Unlike its classical-logic-based counterpart {\bf PA}, {\bf CLA1} is not merely about what arithmetical facts are {\em true}, but about what arithmetical problems can be actually {\em computed} or effectively {\em solved}. More precisely, every formula of the language of {\bf CLA1} expresses a number-theoretic computational {\em problem} (rather than just a true/false {\em fact}), every theorem expresses a problem that has an algorithmic solution,  and every proof encodes such a solution. Does not this sound  exactly like what the constructivists have been calling for? 

Unlike the mathematical or philosophical constructivism, however, and even unlike the early-day theory of computation, modern computer science has long
understood that, what really matters, is not just {\em computability}, but rather {\em efficient computability}. So, the next natural step on the road of revealing the importance of CoL for computer science would be showing that it can be used for  studying efficient computability just as successfully as for studying computability-in-principle. Anyone familiar with the earlier work on CoL could have found reasons for optimistic expectations here. Namely, every provable formula of any of the known sound  axiomatizations of CoL happens to be a scheme of  not only ``always computable'' problems, but ``always efficiently computable'' problems just as well, whatever efficiency exactly means in the context of interactive computation that CoL operates in. That is, at the level of pure logic, computability and efficient computability yield the same classes of valid principles. The study of logic abounds with   phenomena in this style. One example would be the well known fact about classical logic, according to which validity with respect to all possible models is equivalent to validity with respect to just models with countable domains.       

At the level of reasonably expressive applied theories, however, one should certainly expect significant differences depending on whether the underlying concept of interest is  efficient computability or computability-in-principle. For instance, the earlier-mentioned system {\bf CLA1} proves formulas expressing computable but  often intractable arithmetical problems. A purpose of the present paper is to construct a CoL-based system for arithmetic which, unlike {\bf CLA1}, proves only efficiently --- specifically, polynomial time --- computable problems. The new applied formal theory $\arfour$\label{iPTA}  presented in Section \ref{ss11} achieves this purpose. It is also a good starting point for exploring the wider  class of systems under the generic name ``{\em clarithmetic}'' --- arithmetical theories based on CoL, with $\arfour$ being a model example of complexity-oriented versions of clarithmetic, a series of other variations of which, such as systems for polynomial space computability, primitive recursive computability, {\bf PA}-provably recursive computability and so on, are still to come in the near future (see \cite{cla5,cla8}).     Among the main purposes of the present piece of writing is to introduce the promising world of clarithmetic to a relatively wide audience. This explains the semitutorial style in which the paper is written.  It targets readers with no prior familiarity with CoL.

Just like {\bf CLA1}, our present system $\arfour$ is not only a cognitive, but also a problem-solving tool: in order to find a solution for a given problem, it would be sufficient to write the problem in the language of the system, and find a proof of it. An algorithmic solution for the problem then would automatically come together with such a proof. However, unlike the solutions extracted from {\bf CLA1}-proofs, which might be intractable, the solutions extracted from $\arfour$-proofs would always be efficient. 

Furthermore, $\arfour$ turns out to be not only sound, but also complete in a certain reasonable sense that we call {\em extensional completeness}.\label{iextcom} According to the latter, every number-theoretic computational problem that has a polynomial time solution is represented by some theorem of  $\arfour$. Taking into account that there are many ways to represent the same problem, extensional completeness is weaker than what can be called {\em intensional completeness},\label{iintcom} according to which any formula representing an (efficiently) computable problem is provable. In these terms, G\"{o}del's celebrated theorem,\label{igincom} here  with ``truth''=``computability'', is about intensional rather than extensional incompleteness. In fact, extensional completeness is not at all interesting in the context of classical-logic-based theories such as $\pa$. In such theories, unlike CoL-based theories, it is trivially achieved, as the provable formula $\twg$ represents every true sentence.  

Syntactically, our $\arfour$ is an extension of {\bf PA}, and the semantics of the former is a conservative generalization of the semantics of the latter. Namely, the formulas of $\pa$, which form only a proper subclass of the formulas of $\arfour$, are seen as special, ``moveless'' sorts of problems/games, automatically solved/won when true and failed/lost when false. This makes the classical concept of truth just a special case of computability in our sense --- it is nothing but computability restricted to (the problems represented by) the traditional sorts of formulas. And this means that G\"{o}del's incompleteness theorems automatically extend from $\pa$ to $\arfour$, so that, unlike extensional completeness,  intensional completeness in $\arfour$  or any other sufficiently expressive sound CoL-based applied theory is impossible to achieve in principle. As for {\bf CLA1}, it turns out to be incomplete in both senses. Section \ref{sincom} shows that any recursively axiomatizable, sufficiently expressive, sound  system would be (not only intensionally but also) extensionally incomplete, as long as the semantics of the system is based on unrestricted (as opposed to, say, efficient) computability.

Among the main moral merits of the present investigation and its contributions to the overall CoL project is an illustration of the fact that,  in constructing CoL-based applied theories, successfully switching from computability to efficient computability is possible and even more than just possible. As noted, efficient computability, in fact, turns out to be much better behaved than computability-in-principle: the former allows us to achieve completeness in a sense in which 
the latter yields inherent incompleteness.    

An advanced reader will easily understand that the present paper, while focused on the system $\arfour$ of (cl)arithmetic, in fact is not only about arithmetic, but also just as much about CoL-based applied theories or knowledge base systems in general, with $\arfour$ only serving as a model example of such systems.  
Generally, the nonlogical axioms or the knowledge base of a CoL-based applied  system would be any collection of (formulas expressing) problems whose algorithmic or efficient solutions are known. Sometimes, together with nonlogical axioms, we may also have nonlogical rules of inference, preserving the property of computability or efficient computability.   Then, the soundness of the corresponding underlying axiomatization of CoL (in our present case, it is system  \clthree\ studied in \cite{Japtowards, Japlbcs}) --- which usually comes in the strong form called {\em uniform-constructive soundness} --- guarantees that every theorem $T$ of the theory also has an effective or efficient solution and that, furthermore, such
a solution can be effectively or efficiently extracted from a proof of $T$. 
It is this fact that, as mentioned,  makes CoL-based systems problem-solving tools.

More specifically, efficiency-oriented  systems in the above style and $\arfour$ in particular can be seen as  programming languages, where ``programming'' simply means theorem-proving. The soundness of the underlying system guarantees that any proof that can be written will be translatable into a program that runs efficiently and indeed is a solution of the problem expressed by the target formula of the proof. Note that the problem of verifying whether a program meets its specification, which is generally undecidable, is fully neutralized here: the ``specification'' is nothing but the target formula of the proof, and the  proof itself, while  encoding an efficient program, also automatically serves as a verification of the correctness of that program. Furthermore,  every step/formula of the proof can be viewed as its own (best possible)  ``comment''. In a more ambitious and, at this point, somewhat fantastic perspective, after developing reasonable theorem-provers, CoL-based  efficiency-oriented systems can be seen as declarative programming languages in an extreme sense, where human ``programming'' just means writing a formula expressing the problem whose efficient solution is sought for systematic usage in the future. That is, a program simply coincides with its specification. The compiler's job would be finding a proof (the hard part) and translating it into a machine-language code (the easy part). The process of compiling could thus take long but, once compiled, the program would run fast ever after. 

Various  complexity-oriented systems  have been studied in the literature (\cite{bbb1,bbb2,Buss,Bussint,bbb3,bbb4,bbb5,bbb6,Sch} and more).    A notable advantage of CoL-based complexity-oriented systems over the other systems with similar aspirations, which typically happen to be inherently weak systems,  is having actually or potentially unlimited strength, with the latter including the   full  expressive and deductive power of classical logic and Peano arithmetic. In view of the above-outlined potential applications, the importance of this feature is obvious: the stronger a system, the better the chances that a proof/program will be found for a declarative, non-preprocessed, ad hoc specification of the goal. Among the other appealing features of clarithmetic is  being semantically meaningful in the full generality of its language, scalable, and easy to understand in its own right. Syntactically it also tends to be remarkably simple. For instance, on top of the standard Peano axioms, our present system $\arfour$ only has two   additional axioms $\ada x\ade y(y\equals x\plus 1)$ and $\ada x\ade y(y\equals 2x)$, one saying that the function $x\plus 1$ is (efficiently) computable, and the other saying the same about the function $2x$. As will be seen later, from these two innocuous-looking axioms and one (also very simple) rule of induction called {\em $\arfour$-Induction}, via CoL, one can obtain  ``practically full''  information about polynomial time computability of number-theoretic problems, in the same sense as $\pa$, despite   G\"{o}del's incompleteness, allows us to obtain ``practically full'' information about  arithmetical truth. To put it in other words, if a formula $F$ is not provable in $\arfour$, it is unlikely that anyone would  find a polynomial time algorithm solving the problem expressed by $F$: either such an algorithm does not exist, or (as will be seen from Theorem \ref{jan30}) showing its correctness requires going  beyond ordinary combinatorial reasoning formalizable in $\pa$. 

The closest ancestor of our present system $\arfour$ is Buss's {\em bounded arithmetic} for polynomial time. The similarity is related to the single yet important fact that the above-mentioned rule of $\arfour$-Induction is nothing but an adaptation of Buss's PIND (``Polynomial Induction'') principle to the new semantical environment in which $\arfour$ operates. In this sense, $\arfour$ can be characterized as a ``CoL-based bounded arithmetic'', as opposed to Buss's original versions of bounded arithmetic that are based on classical (\cite{Buss}) or intuitionistic (\cite{Bussint}) logics. The switch to CoL as the logical basis for such theories creates notable differences. Among the advantages offered by this switch is absolute flexibility (as long as certain minimum-strength requirements are satisfied) in selecting the underlying ``purely arithmetical'' axioms. Choosing the latter to be the kind old axioms of Peano, as done in $\arfour$, allows us to achieve dramatically greater (than in the case of Buss's systems) intensional strength. Furthermore, as shown in Section \ref{sculprit}, replacing Peano axioms with stronger ones can take us arbitrarily close to intensional completeness. This is just as far as one can go in similar pursuits because, as already noted, in view of G\"{o}del's incompleteness phenomenon, no particular  recursively enumerable system can be intensionally complete. In contrast, even ``slightly'' increasing the strength of the underlying (carefully hand-picked and intensionally very weak) set of arithmetical axioms in classical-logic-based or intuitionistic-logic-based bounded arithmetic immediately results in loss of soundness. We will come back to this topic in Section \ref{scesar}.

\section{An informal overview of the main operations on games}\label{ss2}
%\marginpar{ss2}
Introducing and justifying CoL in  full generality is not among the goals of the present paper --- this job has been done in \cite{Jap03,Japic,Japfin}. We will reintroduce only as much of (the otherwise much wider) CoL as technically necessary for understanding the system $\arfour$ based on  it.  

As noted, formulas in CoL represent computational problems. Such problems are understood as games between two players: $\pp$,\label{ipp} called   {\bf Machine},\label{imachine} and $\oo$,\label{ioo} called  {\bf Environment} (these names will not always be capitalized, and may take articles ``a'' or ``the'').\label{ienvironment} Machine is a mechanical device with  fully determined, algorithmic behavior.  On the other hand, there are no restrictions on the behavior of Environment. A given machine is considered to be  {\em solving} a given problem iff it wins the corresponding game no matter 
how the environment acts. 

Standard atomic sentences, such as ``$0\equals 0$'' or ``Peggy is John's mother'', are understood as special sorts of games, called {\bf elementary}.\label{ielem1} There are no moves in elementary games, and they are automatically won or lost. Specifically, the elementary game represented by a true sentence is won (without making any moves) by Machine, and the elementary game represented by a false sentence is won by Environment.  
 
Logical operators are understood as operations on games/problems. One of the important groups of such operations, termed {\bf choice operations},\label{ichoiceop} comprises  $\adc,\add,\ada,\ade$. These are called {\bf choice conjunction},\label{ichoicecon} {\bf choice disjunction},\label{ichoicedis} {\bf choice universal quantifier}\label{ichoiceuq} and {\bf choice existential quantifier},\label{ichoiceeq} respectively. $A_0\adc A_1$ is a game where the first legal move (``choice"), which should be either $0$ or $1$, is by $\oo$. After such a move/choice $i$ is made, the play continues and the winner is determined according to the rules of $A_i$; if a choice is never made, $\oo$ loses. 
 $A_0\add A_1$ is defined in a symmetric way with the roles of $\oo$ and $\pp$ interchanged: here it is $\pp$ who makes an initial choice and who loses if such a choice is not made. With the universe of discourse being $\{0,1,10,11,100,\ldots\}$ (natural numbers identified with their binary representations), the meanings of the quantifiers $\ada$ and $\ade$  can now be explained by 
\[\ada x A(x)= A(0)\adc A(1)\adc A(10)\adc A(11)\adc A(100)\adc \ldots\] and \[\ade x A(x)= A(0)\add A(1)\add A(10)\add A(11)\add A(100)\add \ldots.\] 

So, for example, 
\[\ada x\bigl(\mbox{\em Prime}(x)\add \mbox{\em Composite}(x)\bigr)\]
is a game where the first move is by Environment. Such a move should consist in selecting a particular number $n$ for $x$, intuitively amounting to asking whether $n$ is prime or composite. This move brings the game down to (in the sense that the game continues as) 
\[\mbox{\em Prime}(n)\add \mbox{\em Composite}(n).\]
Now Machine has to move, or else it loses. The move should consist in choosing one of the two disjuncts. Let us say the left disjunct is chosen, which further brings the game down to $\mbox{\em Prime}(n)$. The latter is an elementary game, and here the interaction ends. Machine wins iff it has chosen a true disjunct. The choice of the left disjunct by Machine thus amounts to claiming/answering that $n$ is prime. Overall, as we see,  $\ada x\bigl(\mbox{\em Prime}(x)\add \mbox{\em Composite}(x)\bigr)$ represents the problem of deciding the primality question.\footnote{For simplicity, here we treat ``Composite'' as the complement of ``Prime'', even though, strictly speaking, this is not quite so: the numbers $0$ and $1$ are neither prime nor composite. Writing ``Nonprime'' instead of ``Composite'' would easily correct this minor inaccuracy.} 

Similarly, 
\[\ada x\ada y\ade z(z\equals x\mult  y)\]
is the problem of computing the product of any two numbers. Here the first two moves are by Environment, which selects some particular $m= x$ and $n= y$, thus asking Machine to tell what the product of $m$ and $n$ is. Machine wins if and only if, in response, it selects a (the) number $k$ for $z$ such that $k\equals m\mult  n$.

Another group of game operations dealt with in this paper comprises $\gneg,\mlc,\mld,\mli$. Employing the classical symbols for these operations is no accident, as they are conservative generalizations of the corresponding Boolean operations from elementary games to all games. 

{\bf Negation} $\gneg$\label{igneg} is a role-switch operation: it turns $\pp$'s moves and wins into $\oo$'s moves and wins, and vice versa. Since elementary games have no moves, only the winners are switched there, so that, as noted, $\gneg$ acts just as the ordinary classical negation when applied to such games. For instance, as $\pp$ is the winner in $0\plus 1\equals 1$, the winner in $\gneg 0\plus 1\equals 1$ will be $\oo$. That is, $\pp$ wins the negation $\gneg A$ of an elementary game $A$ iff it loses $A$, i.e., if $A$ is false. As for the meaning of negation when applied to nonelementary games, at this point it may be useful to observe that $\gneg$ interacts with choice operations in the kind old 
DeMorgan fashion. For example, it would not be hard to see that \[\gneg \ada x\ada y\ade z(z\equals x\mult  y)\ = \ 
\ade x\ade y\ada z(z\notequals x\mult  y).\]

The operations $\mlc$\label{imlc} and $\mld$\label{imld} are called {\bf parallel conjunction} and {\bf parallel disjunction}, respectively.  Playing $A_0\mlc A_1$ (resp. $A_0\mld A_1$) means playing the two games in parallel where, in order to win, $\pp$ needs to win in both (resp. at least one) of the components $A_i$. It is obvious that, just as in the case of negation, $\mlc$ and $\mld$ act as classical conjunction and disjunction when applied to elementary games. For instance, $0\plus 1\equals 1\mld 0\mult  1\equals 1$ is a game automatically won by Machine. There are no moves in it as there are no moves in either disjunct, and   Machine is an automatic winner because it is so in the left disjunct. To appreciate the difference between the two --- choice and parallel --- groups of connectives, compare \[\ada x\bigl(\mbox{\em Prime}(x)\add \gneg \mbox{\em Prime}(x)\bigr)\] and \[\ada x\bigl(\mbox{\em Prime}(x)\mld \gneg \mbox{\em Prime}(x)\bigr).\] The former is a computationally nontrivial problem, existence of an easy (polynomial time) solution for which had remained an open question until a few years ago. As for the latter, it is trivial, as Machine has nothing to do in it: the first (and only) move is by Environment, consisting in choosing a number $n$ for $x$. Whatever $n$ is chosen, Machine wins, as $\mbox{\em Prime}(n)\mld \gneg \mbox{\em Prime}(n)$ is a true sentence and hence an automatically $\pp$-won elementary game.

The operation $\mli$,\label{imli} called {\bf strict reduction}, is defined by $A\mli B= (\gneg A)\mld B$. Intuitively, this is indeed the problem of {\em reducing} 
$B$ to $A$: solving $A\mli B$ means solving $B$ while having $A$ as an external {\em computational resource}. Resources are symmetric to problems: what is a problem to solve for one player is a resource that the other player can use, and vice versa. Since 
$A$ is negated in  $(\gneg A)\mld B$ and negation means switching the roles, $A$ appears as a resource rather than problem for 
$\pp$ in $A\mli B$. 

Consider $\ada x\ade  y(y\equals x^2)$. Anyone who knows the definition of $x^2$ in terms of $\mult $ (but perhaps does not know the meaning of multiplication, or is unable to compute this  function for whatever reason) would be able to solve the  problem 
\begin{equation}\label{april15}
\ada z\ada u\ade v(v\equals  z\mult  u)\ \mli \ \ada x\ade y(y\equals x^2),
\end{equation}
i.e., the problem 
\[\ade z\ade u\ada v(v\notequals z\mult  u)\ \mld \ \ada x\ade y(y\equals x^2),\]
 as it is about reducing the consequent to the antecedent. 
A solution here goes like this. Wait till Environment specifies a value $n$ for $x$, i.e. asks ``what is the square of $n$?''. Do not try to immediately answer this question, but rather specify the same value $n$ for both $z$ and $u$, thus asking the counterquestion: ``what is $n$ times $n$?''. Environment will have to provide a correct answer $m$  to this counterquestion (i.e., specify $v$ as $m$ where $m= n\mult  n$), or else it loses. Then, specify $y$ as $m$, and rest your case. Note that, in this solution, Machine did not have to compute multiplication, doing which had become Environment's responsibility. Machine only correctly reduced the problem of computing square to the problem of computing product, which made it the winner.

Another group of operations that play an important role in CoL comprises  $\cla$\label{icla} and its dual $\cle$\label{icle} (with $\cle xA(x)= \gneg\cla x\gneg A(x)$), called  {\bf blind universal quantifier} and {\bf blind existential quantifier}, respectively.  $\cla xA(x)$ 
can be thought of as a ``version" of $\ada xA(x)$ where the particular value of $x$ that Environment selects is invisible to Machine, so that it has to play blindly in a way that guarantees success no matter what that value is. 

Compare the problems
\(\ada x\bigl(\mbox{\em Even$(x)$}\add \mbox{\em Odd$(x)$}\bigr)\) and  \(\cla x\bigl(\mbox{\em Even$(x)$}\add \mbox{\em Odd$(x)$}\bigr).\) 
Both of them are about telling whether a given number is even or odd; the difference is only in whether that ``given number" is known to Machine or not. The first problem is an easy-to-win, two-move-deep game of a structure that we have already seen.  The second game, on the other hand, is one-move deep with only  Machine to make a move --- select the ``true"  disjunct, which is hardly possible to do as the value of $x$ remains unspecified.

Just like all other operations for which we use classical symbols, the meanings of $\cla$ and $\cle$ are exactly classical  when applied to elementary games. Having this full collection of classical operations makes CoL a generalization and conservative extension of classical logic. 

Going back to an earlier example, even though  (\ref{april15}) expresses a ``very easily solvable'' problem, that formula is still not logically valid. Note that the success  of the reduction strategy of the consequent to the antecedent that we provided for it relies on the nonlogical fact that $x^2\equals x\mult  x$. That strategy would fail in a general case where the meanings of $x^2$ and $x\mult  x$ may not necessarily be the same. On the other hand, the goal of CoL as a general-purpose problem-solving tool should be to allow us find purely logical solutions, i.e., solutions that do not require any special, domain-specific knowledge and (thus) would be good no matter what the particular predicate or function symbols of the formulas mean. Any knowledge that might be relevant should be explicitly stated and included either in the antecedent of a given formula or in the set of axioms (``implicit antecedents'' for every potential formula) of a CoL-based theory. 
In our present case,  formula (\ref{april15}) easily turns into a logically valid one by adding, to its antecedent,  the definition of square in terms of multiplication:
%\marginpar{april16}
\begin{equation}\label{april16}
\cla w (w^2\equals w\mult  w) \mlc \ada z\ada u\ade v(v\equals  z\mult  u)\ \mli \ \ada x\ade y(y\equals x^2).
\end{equation}
The strategy that we provided earlier for (\ref{april15}) is just as good for (\ref{april16}), with the difference that it is successful for (\ref{april16}) no matter what $x^2$ and $z\mult  u$ mean, whereas, in the case of (\ref{april15}), it was guaranteed to be successful only under the standard arithmetical  interpretations of the square and product functions. Thus, our strategy for (\ref{april16}) is, in fact, a ``purely logical'' solution.

The above examples should not suggest that blind quantifiers are meaningful or useful  only when applied to elementary problems. The following is an example of a  winnable nonelementary $\cla$-game:

%\marginpar{lkj}
\begin{equation}\label{lkj}\cla y\Bigl(\mbox{\em Even$(y)$}\add \mbox{\em Odd$(y)$}\ \mli\ \ada x\bigl(\mbox{\em Even$(x\plus  y)$}\add
\mbox{\em Odd$(x\plus  y)$}\bigr)\Bigr).\vspace{-3pt}\end{equation}
Solving this problem, which means reducing the consequent to the antecedent without knowing the value of $y$, is easy: 
$\pp$ waits till $\oo$ selects  a value $n$ for $x$, and also tells --- by selecting a $\add$-disjunct in the antecedent --- whether $y$ is even or odd. Then, 
if $n$ and $y$ are both even or both odd, $\pp$ chooses the left $\add$-disjunct in the consequent, otherwise it chooses the right $\add$-disjunct. Replacing the $\cla y$ prefix by $\ada y$ would significantly weaken the problem, obligating Environment to specify a value for $y$. Our strategy does not really need to know the exact value of $y$, as it only exploits the information about $y$'s being even or odd, provided by the antecedent of the formula.

Many more --- natural, meaningful and useful --- operations beyond the ones discussed in this section have been introduced and studied within the framework of CoL. Here we have only surveyed those that are relevant to our present investigation.

\section{Constant games}\label{cg}
%\marginpar{cg}

Now we are getting down to formal definitions of the concepts informally explained in the previous section. 

To define games formally, we need certain technical terms and conventions. Let us agree that   a {\bf move}\label{imove} means any finite string over the standard keyboard alphabet. 
A {\bf labeled move} ({\bf labmove})\label{ilabmove} is a move prefixed with $\pp$ or $\oo$, with such a prefix ({\bf label})\label{ilabel} indicating which player has made the move. 
A {\bf run}\label{irun} is a (finite or infinite) sequence of labmoves, and a {\bf position}\label{iposition} is a finite run.

We will be exclusively using the letters $\Gamma,\Delta,\Phi$  for runs, and  $\alpha,\beta$ for moves. The letter $\xx$\label{ixx} will always be a variable for players, and \[\overline{\xx}\label{ixxneg}\]  will mean ``$\xx$'s adversary'' (``the other player'').
Runs will be often delimited by ``$\langle$" and ``$\rangle$", with $\emptyrun$ thus denoting the {\bf empty run}.\label{iempty} The meaning of an expression such as $\seq{\Phi,\xx\alpha,\Gamma}$ must be clear: this is the result of appending to the position $\seq{\Phi}$ 
the labmove $\seq{\xx\alpha}$ and then the run $\seq{\Gamma}$.

The following is a formal definition of what we call constant games, combined with some less formal conventions regarding the usage of certain terminology.

\begin{definition}\label{game}
%\marginpar{game}
 A {\bf constant game}\label{iconstantgame} is a pair $A= (\legal{A}{},\win{A}{})$, where:

1. $\legal{A}{}$\label{ilr} is a set of runs  satisfying the condition that a (finite or infinite) run is in $\legal{A}{}$ iff all of its nonempty finite  initial
segments are in $\legal{A}{}$ (notice that this implies $\emptyrun\in\legal{A}{}$). The elements of $\legal{A}{}$ are
said to be {\bf legal runs}\label{ilegrun} of $A$, and all other runs are said to be {\bf illegal}.\label{iillegrun} We say that $\alpha$ is a {\bf legal move}\label{ilegmove} for $\xx$ in a position $\Phi$ of $A$ iff $\seq{\Phi,\xx\alpha}\in\legal{A}{}$; otherwise 
$\alpha$ is {\bf illegal}.\label{iillegmove} When the last move of the shortest illegal initial segment of $\Gamma$  is $\xx$-labeled, we say that $\Gamma$ is a {\bf $\xx$-illegal}\label{ipillegal} run of $A$. 

2. $\win{A}{}$\label{iwn}  is a function that sends every run $\Gamma$ to one of the players $\pp$ or $\oo$, satisfying the condition that if $\Gamma$ is a $\xx$-illegal run of $A$, then $\win{A}{}\seq{\Gamma}= \overline{\xx}$. When $\win{A}{}\seq{\Gamma}= \xx$, we say that $\Gamma$ is a {\bf $\xx$-won}\label{iwon} (or {\bf won by $\xx$}) run of $A$; otherwise $\Gamma$ is {\bf lost}\label{ilost} by $\xx$. Thus, an illegal run is always lost by the player who has made the first illegal move in it.  
\end{definition}

An important operation not explicitly mentioned in Section \ref{ss2} is what is called {\em prefixation}.\label{iprefixation}
This operation takes two arguments: a constant game $A$ and a position $\Phi$ 
 that must 
be a legal position of $A$ (otherwise the operation is undefined), and returns the game $\seq{\Phi}A$.
Intuitively, $\seq{\Phi}A$ is the game playing which means playing $A$ starting (continuing) from position $\Phi$. 
That is, $\seq{\Phi}A$ is the game to which $A$ {\bf evolves} (will be ``{\bf brought down}") after the moves of $\Phi$ have been made. We have already used this intuition when explaining the meaning of choice operations in Section \ref{ss2}: we said that after $\oo$ makes an initial move $i\in\{0,1\}$,
 the game 
$A_0\adc A_1$ continues as $A_i$. What this meant was nothing but that 
$\seq{\oo i}(A_0\adc A_1)= A_i$.
Similarly, $\seq{\pp i}(A_0\add A_1)= A_i$. Here is a definition of prefixation:

\begin{definition}\label{prfx}
%\marginpar{prfx}
Let $A$ be a constant game and $\Phi$ a legal position of $A$. The game 
$\seq{\Phi}A$\label{ipr} is defined by: 
\begin{itemize}
\item $\legal{\seq{\Phi}A}{}= \{\Gamma\ |\ \seq{\Phi,\Gamma}\in\legal{A}{}\}$;
\item $\win{\seq{\Phi}A}{}\seq{\Gamma}= \win{A}{}\seq{\Phi,\Gamma}$.
\end{itemize}
\end{definition}

A terminological convention important to remember is that we often identify a legal position $\Phi$ of a game $A$ with the game $\seq{\Phi}A$. So, for instance, we may say that the move $1$ by $\oo$ brings the game $B_0\adc B_1$ down to the position $B_1$. Strictly speaking, $B_1$ is not a position but a game, and what {\em is} a position is $\seq{\oo 1}$, which we here identified with the game $B_1=\seq{\oo 1}(B_0\adc B_1)$.

We say that a constant game $A$ is {\bf finite-depth} iff there is an integer $d$ such that no legal run of $A$ contains more than $d$ labmoves. The smallest of such integers $d$ is called the {\bf depth}\label{idepth} of $A$. ``{\bf Elementary}'' means ``of depth $0$''.

This paper will exclusively deal with finite-depth games. This restriction of focus makes many definitions and proofs simpler. Namely, in order to define a finite-depth-preserving game operation $O(A_1,\ldots,A_n)$ applied to such games, it suffices  to specify the following:

\begin{description}
\item[(i)] Who wins $O(A_1,\ldots,A_n)$ if no moves are made, i.e., the value of $\win{O(A_1,\ldots,A_n)}{}\emptyrun$.
\item[(ii)] What are the {\bf initial legal (lab)moves},\label{iilm} i.e., the elements of  $\{\xx\alpha\ |\ \seq{\xx\alpha}\in\legal{O(A_1,\ldots,A_n)}{}\}$, and to 
what game  is the game $O(A_1,\ldots,A_n)$ brought down after such an initial legal labmove $\xx\alpha$ is made. Recall that, by saying that a given labmove $\xx\alpha$ brings a given game $A$ down to $B$, we mean that $\seq{\xx\alpha}A= B$.  
\end{description}
Then, the set of legal runs of $O(A_1,\ldots,A_n)$ will be uniquely defined, and so will be the winner in every legal (and hence finite) run of the game. 

Below we define a number of operations for finite-depth games only. Each of these operations can be easily seen to preserve the finite-depth property. Of course, more general definitions of these operations --- not restricted to finite-depth games --- do exist (see, e.g., \cite{Japfin}), but in this paper we are trying to keep things as simple as possible.

\begin{definition}\label{op} Let $A$, $B$, $A_0,A_1,\ldots$ be finite-depth constant games, and $n$ be a positive integer.\vspace{9pt}
%\marginpar{op}

\noindent 1. $\gneg A$\label{igneg2} is defined by: 
\begin{quote}\begin{description}
\item[(i)] $\win{\gneg A}{}\emptyrun = \xx$ iff $\win{A}{}\emptyrun =\overline{\xx}$. 
\item[(ii)] $\seq{\xx\alpha}\in\legal{\gneg A}{}$ iff $\seq{\overline{\xx}\alpha}\in\legal{A}{}$. Such an initial legal labmove $\xx\alpha$ brings the game down to 
$\gneg \seq{\overline{\xx}\alpha}A$.\vspace{5pt}
\end{description}\end{quote}

\noindent 2. $A_0\adc\ldots\adc  A_n$\label{iadc2} is defined by: 
\begin{quote}\begin{description}
\item[(i)] $\win{A_0\adci\ldots\adci  A_n}{}\emptyrun = \pp$. 
\item[(ii)] $\seq{\xx\alpha}\in\legal{A_0\adci\ldots\adci  A_n}{}$ iff $\xx= \oo$ and $\alpha= i  \in\{0,\ldots,n\}$.\footnote{Here the {\em number} $i$ is identified with the standard bit {\em string} representing it in the binary notation.  The same applies to the other clauses of this definition.}  Such an initial legal labmove $\oo i$ brings the game down to 
$A_i$.\vspace{5pt} 
\end{description}\end{quote}

\noindent 3. $A_0\mlc\ldots\mlc A_n$\label{imlc2} is defined by: 
\begin{quote}\begin{description}
\item[(i)] $\win{A_0\mlci\ldots\mlci  A_n}{}\emptyrun= \pp$ iff, for each $i\in\{0,\ldots,n\}$,  $\win{A_i}{}\emptyrun= \pp$. 
\item[(ii)] $\seq{\xx\alpha}\in\legal{A_0\mlci\ldots\mlci  A_n}{}$ iff $\alpha= i.\beta$, where $i\in\{0,\ldots,n\}$ and $\seq{\xx\beta}\in\legal{A_i}{}$. Such an initial legal labmove $\xx i.\beta$ brings the game down to  
\[ A_0\mlc\ldots\mlc A_{i-1}\mlc \seq{\xx\beta}A_i\mlc A_{i\plus 1}\mlc\ldots\mlc A_n.\vspace{3pt}\] 
\end{description}\end{quote}

\noindent 4. $A_0\add\ldots\add A_n$\label{iadd2} and $A_0\mld\ldots\mld  A_n$
are defined exactly as $A_0\adc\ldots\adc A_n$ and $A_0\mlc\ldots\mlc  A_n$, respectively, only with ``$\pp$" and ``$\oo$" interchanged.\vspace{7pt}

\noindent 5. The infinite $\adc$-conjunction $A_0\adc A_1\adc\ldots$ is defined exactly as $A_0\adc\ldots\adc A_n$, only with ``$i\in\{0,1,\ldots\}$" instead of ``$i\in\{0,\ldots,n\}$". Similarly for the infinite version of $\add$.\vspace{7pt}

\noindent 6. In addition to the earlier-established meanings, the symbols $\twg$\label{itwg2} and $\tlg$ also denote two special --- simplest --- constant games, defined by  $\win{\twg}{}\emptyrun=\pp$, $\win{\tlg}{}\emptyrun= \oo$ and $\legal{\twg}{}= \legal{\tlg}{}= \{\emptyrun\}$.\vspace{7pt} 

\noindent 7. $A\mli B$\label{imli2} is treated as an abbreviation of $(\gneg A)\mld B$.
\end{definition}

\begin{example}
 The game $(0\equals 0\adc 0\equals 1)\mli(10\equals 11\adc 10\equals 10)$, i.e. \(\gneg (0\equals 0\adc 0\equals 1)\mld(10\equals 11\adc 10\equals 10),\)
 has thirteen legal runs, which are: 
\begin{description}
\item[1] $\seq{}$. It is won by $\pp$, because $\pp$ is the winner in the right $\mld$-disjunct (consequent).
\item[2] $\seq{\pp 0.0}$. (The labmove of) this run brings the game down to $\gneg 0\equals 0\mld(10\equals 11\adc 10\equals 10)$, and $\pp$ is the winner for the same reason as in the previous case.
\item[3] $\seq{\pp 0.1}$. It brings the game down to $\gneg 0\equals 1\mld(10\equals 11\adc 10\equals 10)$, and $\pp$ is the winner because it wins in both $\mld$-disjuncts. 
\item[4] $\seq{\oo 1.0}$. It brings the game down to $\gneg(0\equals 0\adc 0\equals 1)\mld 10\equals 11$.  $\pp$ loses as it loses in both $\mld$-disjuncts. 
\item[5] $\seq{\oo 1.1}$. It brings the game down to $\gneg (0\equals 0\adc 0\equals 1)\mld 10\equals 10$.  $\pp$ wins as it wins in the right $\mld$-disjunct. 
\item[6-7] $\seq{\pp 0.0,\oo 1.0}$ and $\seq{\oo 1.0, \pp 0.0}$. Both bring the game down to the false $\gneg 0\equals 0 \mld 10\equals 11$, and both are lost by  $\pp$. 
\item[8-9] $\seq{\pp 0.1,\oo 1.0}$ and $\seq{\oo 1.0, \pp 0.1}$. Both bring the game down to the true $\gneg 0\equals 1 \mld 10\equals 11$, which makes  $\pp$ the winner.
\item[10-11] $\seq{\pp 0.0,\oo 1.1}$ and $\seq{\oo 1.1, \pp 0.0}$. Both bring the game down to the true $\gneg 0\equals 0 \mld 10\equals 10$, so $\pp$ wins.
\item[12-13] $\seq{\pp 0.1,\oo 1.1}$ and $\seq{\oo 1.1, \pp 0.1}$. Both bring the game down to the true $\gneg 0\equals 1 \mld 10\equals 10$, so $\pp$ wins.
\end{description}
\end{example}

\section{Games as generalized predicates}\label{nncg}
%\marginpar{nncg}

Constant games can be seen as generalized propositions: while propositions in classical logic are just elements 
of $\{\twg,\tlg\}$, constant games are functions from runs to $\{\twg,\tlg\}$.
As we know, however, propositions only offer  very limited expressive power, 
and classical logic needs 
to consider the more general concept of predicates, with propositions being nothing but special --- constant --- cases of predicates. The situation in CoL is similar. Our concept of a (simply) game generalizes that of a constant game in the same sense as the classical concept of a predicate generalizes that of a proposition.

We fix an infinite set of expressions called {\bf variables}, for which we will be  using \(x,y,z,s,r,t,u,v,w\) as metavariables.  An expression like $\vec{x}$ will usually stand for a finite sequence of  variables. Similarly for later-defined objects such as constants or terms.

We also fix another infinite set of expressions called {\bf constants}:\label{iconstant} 
\[\{\epsilon,1,10,11,100,101,110,111,1000,\ldots\}.\] These are thus  {\bf binary numerals}\label{ibinnum} --- the strings matching the regular expression $\epsilon\cup 1(0\cup 1)^*$, where $\epsilon$ is the empty string.  We will be typically identifying such strings --- by some rather innocent abuse of concepts --- with the natural numbers represented by them in the standard binary notation, and vice versa. Note that $\epsilon$ represents $0$. For this reason, following the many-century tradition, we shall usually write $0$ instead of $\epsilon$, keeping in mind that, in such contexts, the length $|0|$ of the string $0$ should be seen to be $0$ rather than $1$.  We will be mostly using $a,b,c,d$ as metavariables for constants.

A {\bf universe} (of discourse) is a pair $(U, ^U)$, where $U$ is a nonempty set, and $^U$, called the {\bf naming function} of the universe,  is a function that sends each constant $c$  to an element $c^U$ of $U$. The intuitive meaning of $c^U=\mathfrak{s}$ is that $c$ is a {\bf name} of $\mathfrak{s}$. Both terminologically and notationally, we will typically identify each universe $(U, ^U)$ with its first component and, instead of ``$(U, ^U)$'', write simply ``$U$'', keeping in mind that each such ``universe'' $U$ comes with a fixed  associated function $^U$. A universe $U=(U,^U)$ is said to be {\bf ideal} iff $U$ coincides with the above-fixed set of constants, and $^U$ is the identity function on that set. Note that, in a non-ideal universe, such as the set of all real numbers, some objects may have several names (e.g., $1/3, 2/6, 3/9$ are different names of the same number), some have unique names (e.g. the famous number $\pi$), and some have no names at all.  The same applies to the universe of astronomy, where some stars and planets have unique names, some have several names (Venus = Morning Star = Evening Star), and most have no names at all. On the other hand, the standard universe of arithmetic is ideal: every natural number has a unique name --- the corresponding binary numeral --- with which it can be identified. 

By a {\bf valuation}\label{ivaluation} on a universe $U$, or a $U$-valuation,  we mean 
a mapping $e$ that sends each variable $x$ to an element $e(x)$ of $U$. 
When a universe $U$ is fixed,  irrelevant or clear from the context, we may omit references to it and simply say ``valuation''. In these terms, a classical predicate $p$ can be understood as 
a function that sends each valuation $e$ to a proposition, i.e., to a constant predicate.   Similarly, what we call a game sends valuations to constant games: 

\begin{definition}\label{ngame}
%\marginpar{ngame}
Let $U$ be a universe. A {\bf game on $U$} is a function $A$ from $U$-valuations   to constant games. We write $e[A]$\label{iea} (rather than $A(e)$) to denote the constant game returned by $A$ on valuation $e$. Such a constant game $e[A]$ is said to be an {\bf instance}\label{iinstance} of $A$. 
For readability, we usually write $\legal{A}{e}$\label{ilre} and $\win{A}{e}$ instead of $\legal{e[A]}{}$ and $\win{e[A]}{}$.
\end{definition}

Just as this is the case with propositions versus predicates, constant games in the sense of Definition \ref{game} will
be thought of as special, constant cases of games in the sense of Definition \ref{ngame}. In particular, each constant game $A'$ is the game $A$ such that, for every valuation $e$,
$e[A]= A' $. From now on we will no longer distinguish between such $A$ and $A' $, so that, if $A$ is a constant game,
it is its own instance, with $A= e[A]$ for every $e$.

Where $n$ is a natural number, we say that a game $A$ is {\bf $n$-ary}\label{igarity} iff there are $n$ variables such that, for any two valuations $e_1$ and $e_2$ that agree on all those variables, we have $e_1[A]= e_2[A]$. Generally, a game that is $n$-ary for some $n$, is said to be {\bf finitary}.\label{ifinitary} Our paper is going to exclusively deal with finitary games and, for this reason, we agree that, from now on, when we say ``game'', we usually mean ``finitary game''.  

For a variable $x$ and valuations  $e_1,e_2$, we write $e_1\equiv_x e_2$ to mean that the two valuations have the same universe  and agree on all variables  other than $x$.

We say that a game $A$ {\bf depends} on a variable $x$ iff there are two valuations  $e_1,e_2$ with $e_1\equiv_x e_2$  such that $e_1[A]\not= e_2[A]$. An $n$-ary game thus depends on at most $n$ variables. And constant games are nothing but $0$-ary games, i.e., games that do not depend on any variables.

We say that a (not necessarily constant) game $A$ is {\bf elementary}\label{ielem2} iff so are all of its instances $e[A]$. And we say that $A$ is {\bf finite-depth}\label{fdpth} iff there is a (smallest) integer $d$, called the {\bf depth} of $A$, such that the depth of no instance of $A$ exceeds $d$.

Just as constant games are generalized propositions, games can be treated as generalized predicates. Namely, we will see each predicate $p$ of whatever arity as  the same-arity elementary game such that, for every valuation $e$,
$\win{p}{e}\emptyrun=\pp$ iff $p$ is true at $e$.  
And vice versa: every elementary game $p$ will be seen as the same-arity predicate which is true at a given valuation $e$ iff  $\win{p}{e}\emptyrun=\pp$.   
Thus, for us, ``predicate'' and ``elementary game'' are going to be synonyms. Accordingly,  any standard terminological or notational conventions familiar from the literature for predicates also apply to them seen as elementary games. 

Just as the Boolean operations straightforwardly extend from propositions to all predicates, our operations 
$\gneg,\mlc,\mld,\mli,\adc,\add$ extend from constant games to all games. This is done by simply stipulating that $e[\ldots]$ commutes with all of those operations: $\gneg A$ is 
the game such that, for every valuation $e$, $e[\gneg A]=\gneg e[A]$; $A\adc B$ is the game such that,
for every valuation $e$, $e[A\adc B]= e[A]\adc e[B]$; etc. So does the operation of prefixation: provided that $\Phi$ is a legal position of every instance of $A$,  $\seq{\Phi}A$ is  understood as the unique game such that, for every valuation $e$, $e[\seq{\Phi}A]= \seq{\Phi}e[A]$.

\begin{definition}\label{sov}
%\marginpar{sov}
Let $A$ be a finite-depth game on a universe $U$, $x_1,\ldots,x_n$ be pairwise distinct variables, and $c_1,\ldots,c_n$ be  constants. 
On the same universe, the game  which we call the result of {\bf substituting $x_1,\ldots,x_n$ by $c_1,\ldots,c_n$ in $A$}, denoted $A(x_1/c_1,\ldots,x_n/c_n)$, is defined by stipulating that, for every valuation $e$ on $U$, $e[A(x_1/c_1,\ldots,x_n/c_n)]= e'[A]$, where $e'$ is the valuation on $U$ that sends each $x_i$ to $c_i$ and agrees with $e$ on all other variables. 
\end{definition}

Following the standard readability-improving practice established in the literature for predicates, we will often fix pairwise distinct  variables $x_1,\ldots,x_n$ for a game $A$ and write $A$ as $A(x_1,\ldots,x_n)$. 
Representing $A$ in this form  sets a context in which we can write $A(c_1,\ldots,c_n)$ to mean the same as the more clumsy expression $A(x_1/c_1,\ldots,x_n/c_n)$. 

\begin{definition}\label{bq}
%\marginpar{bq} 
Let  $A(x)$ be a finite-depth game on a given universe. On the same universe,   $\ada xA(x)$  and  $\ade xA(x)$ are defined as the following two games, respectively:    
\[\begin{array}{l}
A(0)\adc A(1)\adc A(10)\adc A(11)\adc A(100)\adc \ldots;\\  
A(0)\add A(1)\add A(10)\add A(11)\add A(100)\add \ldots .   
\end{array}\]
\end{definition}

Thus, every initial legal move of $\ada xA(x)$ or $\ade xA(x)$ is a constant $c\in\{0,1,10,11,100,\ldots\}$,  which in our informal language we may refer to as ``the constant chosen (by the corresponding player) for $x$''.

We will say that a game $A$ is  {\bf unistructural}\label{iunistructural} iff, for any two valuations $e_1$ and $e_2$,    $\legal{A}{e_1}= \legal{A}{e_2}$. Of course, all constant or elementary games are unistructural. It can also be easily seen that all our game operations preserve the unistructural property of games. For the purposes of the present paper, considering only unistructural games would be sufficient. 

We define the remaining operations $\cla$ and $\cle$ only for unistructural games:

\begin{definition}\label{op5} Below  $A(x)$ is an arbitrary finite-depth unistructural game on a universe $U$. On the same universe:\vspace{9pt}
%\marginpar{op5}

\noindent 1. The game $\cla x A(x)$  is defined by stipulating that, for every $U$-valuation $e$, player $\xx$ and move $\alpha$, we have: 
\begin{quote}\begin{description}
\item[(i)] $\win{\clai x A(x)}{e}\emptyrun= \pp$ iff, for every valuation $g$ with $g\equiv_x e$, $\win{A(x)}{g}\emptyrun= \pp$.
\item[(ii)] $\seq{\xx\alpha}\in\legal{\clai x A(x)}{e}$ iff $\seq{\xx\alpha}\in\legal{A(x)}{e}$. Such an initial legal labmove $\xx\alpha$ brings the game $e[\cla x A(x)]$ down to 
$e[\cla x\seq{\xx\alpha}A(x)]$.\vspace{5pt}
\end{description}\end{quote}
\noindent 2. The game $\cle x A(x)$  is defined in exactly the same way, only with $\pp$ and $\oo$ interchanged.  
\end{definition}

\begin{example}\label{may14}
%\marginpar{may14}
Consider the game (\ref{lkj}) on the ideal universe, 
discussed earlier in Section \ref{ss2}. The sequence 
$\seq{\oo 1.11,\ \oo 0.0,\ \pp 1.1}$ 
is a legal run of (\ref{lkj}), the effects of the moves of which are shown below:
\[\begin{array}{ll}
(\ref{lkj}):  & \cla y\Bigl(\mbox{\em Even}(y)\add \mbox{\em Odd}(y) \mli \ada  x\bigl(\mbox{\em Even}(x\plus y)\add \mbox{\em Odd}(x\plus y)\bigr)\Bigr)\\
\seq{\oo 1.11}(\ref{lkj}):  & \cla y\bigl(\mbox{\em Even}(y)\add \mbox{\em Odd}(y) \mli \mbox{\em Even}(11\plus y)\add \mbox{\em Odd}(11\plus y)\bigr)\\
\seq{\oo 1.11, \oo 0.0}(\ref{lkj}): &  \cla y\bigl(\mbox{\em Even}(y) \mli \mbox{\em Even}(11\plus y)\add \mbox{\em Odd}(11\plus y)\bigr)\\
\seq{\oo 1.11, \oo 0.0,\pp 1.1}(\ref{lkj}): & \cla y\bigl(\mbox{\em Even}(y) \mli \mbox{\em Odd}(11\plus y)\bigr)
\end{array}\]
The play hits (ends as) the true proposition $\cla y\bigl(\mbox{\em Even}(y) \mli \mbox{\em Odd}(11\plus y)\bigr)$ and hence is won by $\pp$. 

\end{example}

Before closing this section, we want to make the rather straightforward observation that the DeMorgan dualities hold for all of our sorts of conjunctions, disjunctions and quantifiers, and so does the double negation principle. That is,  we always have:
\[\gneg\gneg A= A;\vspace{-3pt}\]
\[\gneg(A\mlc B)= \gneg A\mld\gneg B; \ \ \ \ \gneg(A\mld B)= \gneg A\mlc\gneg B;\vspace{-3pt}\]
\[\gneg(A\adc B)= \gneg A\add\gneg B;  \ \ \ \ \gneg(A\add B)= \gneg A\adc\gneg B;\vspace{-3pt}\]
\[\gneg \cla xA(x)= \cle x\gneg A(x); \ \ \ \ \gneg \cle xA(x)= \cla x\gneg A(x);\vspace{-3pt}\]
\[\gneg \ada xA(x)= \ade x\gneg A(x); \ \ \ \ \gneg \ade xA(x)= \ada x\gneg A(x).\]

\section{Algorithmic strategies through interactive machines}\label{icp}
%\marginpar{icp}

In traditional game-semantical approaches, including Blass's \cite{Bla72,Bla92} approach which is the closest precursor of ours, player's strategies are understood as {\em functions} --- typically as functions from interaction histories (positions) to moves, or sometimes (\cite{Abr94}) as functions that only look at the latest move of the history. This {\em strategies-as-functions} approach, however, is inapplicable in the context of CoL, whose relaxed semantics, in striving to get rid of ``bureaucratic pollutants'' and only deal with the remaining true essence of games,  does not impose any regulations on which player can or should move in a given situation. Here, in many cases, either player may have (legal) moves, and then it is unclear whether the next move should be the one prescribed by $\pp$'s strategy function or the one prescribed by the strategy function of $\oo$. For a game semantics whose ambition is to provide a comprehensive, natural and direct tool for modeling interaction, the strategies-as-functions approach would be  less than adequate, even if technically possible. This is so for the simple reason that  the strategies that real computers follow are not functions. If the strategy of your personal computer was a function from the history of interaction with you, then its performance would keep noticeably worsening due to the need to read the continuously lengthening --- and, in fact, practically infinite --- interaction history every time before responding. Fully ignoring that history and looking only at your latest keystroke in the spirit of \cite{Abr94} is also not what your computer does, either.  

In CoL, ($\pp$'s effective) strategies are defined in terms of interactive machines, where computation is one continuous process interspersed with --- and influenced by --- multiple ``input'' (environment's moves) and ``output'' (machine's moves) events. Of several, seemingly rather different yet equivalent,  machine models of interactive computation studied in CoL, here we will employ the most basic, {\bf HPM}\label{ihpm} (``Hard-Play Machine'') model.

An HPM is nothing but a Turing machine with the additional capability of making moves. The adversary can also move at any time, with such moves being the only nondeterministic events from the machine's perspective. Along with the ordinary  work tape, the machine has an additional  tape called the run tape. The latter, serving as a dynamic input,  at any time  spells the ``current position'' of the play. Its role is to make  the run fully visible to the machine.  

In these terms,  an  algorithmic solution ($\pp$'s winning strategy) for a given  constant game $A$ is understood as an HPM $\cal M$ such that,  no matter how the environment acts during its interaction with $\cal M$ (what moves it makes and when), the run incrementally spelled on the run tape is a $\pp$-won run of $A$.  
As for $\oo$'s strategies, there is no need to define them: all possible behaviors by $\oo$ are accounted for by the different possible nondeterministic updates  of the run tape of an HPM. 

In the above outline, we described HPMs in a relaxed fashion, without being specific about technical details such as, say, how, exactly, moves are made by the machine, how many moves either player can make at once, what happens if both players attempt to move ``simultaneously'', etc. As it turns out, all reasonable design choices yield the same class of winnable games as long as we consider a certain natural subclass of games called {\bf static}.\label{istatic} Such games are obtained by imposing a certain simple formal condition on games (see, e.g., Section 5 of \cite{Japfin}), which we do not reproduce here as nothing in this paper relies on it. We shall only point out that, intuitively, static games are interactive tasks where the relative speeds of the players are irrelevant, as it never hurts a player to postpone making moves. In other words, static games are games that are contests of intellect rather than contests of speed. And one of the theses that CoL philosophically relies on is that static games present an adequate formal counterpart of our intuitive concept of ``pure'', speed-independent interactive computational problems. Correspondingly, CoL restricts its attention (more specifically, possible interpretations of the atoms of its formal language) to static games. All elementary games turn out to be trivially static, and the class of static games turns out to be closed under all game operations studied in CoL. More specifically, all games expressible in the language of the later-defined logic $\cltw$, or theory $\arfour$, are static, as well as constant, finitary and finite-depth.   
Accordingly, we agree that, in this paper, we shall use the 
 term ``{\bf computational problem}", or simply ``{\bf problem}", as a synonym of ``constant, finitary, finite-depth, static game''.

\section{The HPM model in greater detail}

As noted, computability of static games is rather robust with respect to the technical details of the underlying model of interaction. And the  loose description of HPMs that we gave in the previous section would be sufficient for most purposes, just as mankind had been rather comfortably studying and using algorithms long before the Church-Turing thesis in its precise form came around. Namely, relying on just the intuitive concept of algorithmic strategies (believed in CoL to be adequately captured by the HPM model) would be sufficient if we only needed to show existence of such strategies for various games. As it happens, however, later sections of this paper need  to arithmetize such strategies in order to prove the promised extensional completeness of $\arfour$. The complexity-theoretic concepts defined in the next section also require certain more specific details about HPMs, and in this section we provide such details. It should be pointed out again that most --- if not all --- of such details are ``negotiable'', as different reasonable arrangements would   yield equivalent models. 

Just like an ordinary Turing machine, an HPM has a finite set of {\bf states},\label{istate} one of which has the special status of being the {\bf start state}. There are no accept, reject, or halt states, but there are specially designated states called {\bf move states}.\label{imovestate} It is assumed that the start state is not among the move states. As noted earlier, this is a two-tape machine, with a  read-write {\bf work tape}\label{iworktape}  and read-only {\bf run tape}.\label{iruntape}  Each tape has a beginning but no end, and is divided into infinitely many {\bf cells},\label{icell} arranged in the left-to-right order: cell $\# 0$, cell $\# 1$, cell $\# 2$, etc. At any time, each cell will contain one symbol from a certain fixed finite set of {\bf tape symbols}.\label{itapesymbol} The {\bf blank} symbol, as well as $\pp$ and $\oo$, are among the tape symbols. 
We also assume that these three symbols  are not among the symbols that any (legal or illegal) move can ever contain.  
Either tape has its own {\bf scanning head},\label{ihead} at any given time looking (located) at one of the cells of the tape.  A transition from one {\bf computation step}  (``{\bf clock cycle}'')\label{icc}   to another happens according to the fixed {\bf transition function}\label{itf} of the machine. The latter, depending on the current state, and the symbols seen by the two heads on the corresponding tapes, deterministically prescribes the next state, the tape symbol by which the old symbol should be overwritten in the current cell   (the cell currently scanned by the  head) of the work tape, and, for each head, the direction --- one cell left or one cell right --- in which the head should move. A constraint here is that the blank symbol, $\pp$ or $\oo$ can never be written by the machine on the work tape. An attempt to move left when the head of a  given   tape is looking at the leftmost cell  results in staying put. So does an attempt 
to move right when the head is looking at the blank symbol. 

When the machine starts working, it is in its start state, both scanning heads are looking at the leftmost cells of the corresponding tapes, and (all cells of) both  tapes are blank (i.e., contain the blank symbol). Whenever the machine enters a move state, the string $\alpha$ spelled by (the contents of) its work tape cells, starting from cell $\# 0$ and ending with the cell immediately left to the work-tape scanning head,   will be automatically appended --- at the beginning of the next clock cycle --- to the contents of the run tape in the $\pp$-prefixed form  $\pp\alpha$. And, on every transition, whether the machine is in a move state or not, any finite sequence $\oo\beta_1,\ldots,\oo\beta_m$ of $\oo$-labeled moves may be nondeterministically appended to the contents of the run tape. If the above two events happen on the same clock cycle, then the moves will be appended to the contents of the run tape in the following order: $\pp\alpha\oo\beta_1\ldots\oo\beta_m$ (note the technicality that labmoves are listed on the run tape without blanks or commas between them). 

With each labmove  that emerges on the run tape, we associate its {\bf timestamp},\label{itimestamp} which is the number of the clock cycle immediately preceding the cycle on which the move first emerged on the run tape. Intuitively, the timestamp indicates on which cycle the move was  {\bf made} rather than {\em appeared} on the run tape: a move made during cycle $\# i$ appears on the run tape on cycle $\# i\plus 1$ rather than $\# i$. Also, we agree that the count of clock cycles, just like the count of cells, starts from $0$, meaning that the very first clock cycle is cycle $\#0$ rather than $\#1$. 

A {\bf configuration}  of a given HPM $\cal M$ is a full description of the contents of the two tapes, the locations of the two scanning heads, and the state of the machine at the beginning of some (``current'') clock cycle. 
A {\bf computation branch}\label{icb} of $\cal M$ is an infinite sequence $C_0,C_1,C_2,\ldots$ of configurations of $\cal M$, where $C_0$ is the initial configuration (as explained earlier), and every $C_{i\plus 1}$ is a configuration that could have legally followed (again,  in the sense explained earlier) $C_i$.  For a computation branch $B$, the {\bf run spelled by $B$}\label{irsb} is the run $\Gamma$ incrementally spelled on the run tape in the corresponding scenario of interaction. We say that such a $\Gamma$ is {\bf a run generated by}\label{irgb} the machine. 

We say that a given HPM $\cal M$ {\bf wins} ({\bf computes}, {\bf solves}) a given constant game $A$ --- and write ${\cal M}\models A$\label{imodels} --- iff every run $\Gamma$ generated by $\cal M$  is a $\pp$-won run of $A$. We say that $A$ is {\bf computable}\label{icomputable} iff there is an HPM $\cal M$ such that ${\cal M}\models A$; such an HPM is said to be an (algorithmic) {\bf solution},\label{isol} or {\bf winning strategy}, for $A$.

\section{Interactive complexity}\label{s7}
%\marginpar{s7}

The {\bf size} of a move $\alpha$ means the length of $\alpha$ as a string.  In the context of a given computation branch of a given HPM $\cal M$, by the 
{\bf background} of a clock cycle $c$ we mean the greatest of the sizes of Environment's moves made by (before) time $c$, or $0$ if there are no such moves. 
If $\cal M$ makes a move on cycle $c$, then the background of that move\footnote{As easily understood, here and in similar contexts,  ``move'' means a move not as a {\em string}, but as an {\em event}, namely, the event of $\cal M$ making a move at time $c$.} means the background of $c$. 
Next, whenever $\cal M$ makes a move on cycle $c$, by  the {\bf timecost} of that move  we mean $c\minus d$, where $d$ is the greatest cycle with $d\mless c$ on which a move was made by either player, or is $0$ if there is no such cycle.   

Throughout this paper, an $n$-ary {\bf arithmetical function} means a  function from $n$-tuples of natural numbers to natural numbers. As always, ``unary'' means ``$1$-ary''. 

\begin{definition}\label{deftcs}
%\marginpar{deftcs}
Let  $h$ be a unary arithmetical  function, and $\cal M$ an HPM. 

1. We say that {\bf $\cal M$ runs in time $h$}, or that $\cal M$ is an {\bf $h$ time machine}, or that $h$ is a {\bf bound} for the time complexity of $\cal M$, iff, in every play (computation branch), for  any clock cycle $c$ on which $\cal M$ makes a move, neither the timecost nor the size of that move  exceeds   $h(\ell)$,  where $\ell$ is the background of $c$.  

2. We say that {\bf $\cal M$ runs in space $h$}, or that $\cal M$ is an {\bf $h$ space machine}, or that $h$ is a {\bf bound} for the space complexity of $\cal M$, iff,  in every play (computation branch), for  any clock cycle $c$,    the number of cells ever visited by the work-tape head of $\cal M$ by time $c$ does not exceed $h(\ell)$, where $\ell$ is the background of $c$.       

\end{definition}

Our time complexity concept can be seen to be in the spirit of what is usually called {\em response time}. The latter generally does not and should not depend on the length of the preceding interaction history. On the other hand, it is not and should not be merely a function of the adversary's last move, either. A   similar characterization applies to our concept of space complexity. Both complexity measures are equally  meaningful whether it be in the context of ``short-lasting'' games (such as the ones represented by the formulas  of the later-defined logic $\cltw$) or  the context of games that may have ``very long'' and even infinitely long legal runs.

Let $A$ be a constant game, $h$ a unary arithmetical function, and $\cal M$ an HPM. We say that {\bf $\cal M$ wins} ({\bf computes}, {\bf solves}) {\bf $A$ in time $h$}, or that {\bf $\cal M$ is an $h$ time solution for $A$},  iff $\cal M$ is an $h$ time machine with ${\cal M}\models A$. 
We say that $A$ is {\bf computable} ({\bf winnable}, {\bf solvable}) {\bf in time $h$} iff it has an $h$ time solution. Similarly for ``{\bf space}'' instead of ``time''. 

When we say {\bf polynomial time}, it is to be understood as ``time $h$ for some polynomial function $h$''. Similarly for {\bf polynomial space}.

Many concepts introduced within the framework of CoL are generalizations ---  for  the interactive context --- of ordinary and well-studied concepts of the traditional theory of computation. The above-defined  time and space complexities are among such concepts. Let us focus on polynomial time for the rest of this section, and look at the traditional notion of  {\em polynomial time computability} of a function $f(x)$ for instance. With a moment's thought, it can be seen to be equivalent to  polynomial time computability (in our sense) of the problem $\ada x\ade y\bigl(y=f(x)\bigr)$. Similarly, {\em polynomial time decidability} of a predicate $p(x)$ means the same as polynomial time computability of the problem 
 $\ada x\bigl(\gneg p(x)\add p(x)\bigr)$.  Further, what is traditionally called (mapping) {\em polynomial time reducibility} of a predicate $p(x)$ to a predicate $q(x)$ can be seen to mean nothing but polynomial time computability of the problem $\ada x\ade y \bigl(p(x)\leftrightarrow q(y)\bigr)$, where $E\leftrightarrow F$ is an abbreviation of $(E\mli F)\mlc (F\mli E)$. If we want to say that a particular function $f(x)$ is a polynomial time reduction of $p(x)$ to $q(x)$, then we can write $\ada x\ade y\bigl(y=f(x)\bigr)\mlc \cla x\Bigl(p\bigl(x\bigr)\leftrightarrow q\bigl(f(x)\bigr)\Bigr)$. And so on. 

Our formalism can be used for systematically defining and studying an infinite variety of meaningful complexity-theoretic properties, relations and operations, only some of which (as the above ones) may have established names in the literature. Consider, for instance, the problem
\begin{equation}\label{ffb1} 
\ada x\bigl(\gneg q(x)\add q(x)\bigr)\mli \ada x\bigl(\gneg p(x)\add p(x)\bigr).
\end{equation} 
It expresses (its polynomial time computability means, that is) a sort of polynomial time reducibility of $p(x)$ to $q(x)$. This reducibility can be seen to be strictly weaker than the traditional sort of polynomial time reducibility captured by the earlier mentioned $\ada x\ade y\bigl(y=f(x)\bigr)$. For instance, every predicate $p(x)$ is reducible to its complement $\gneg p(x)$ in the sense of (\ref{ffb1}). Namely, a polynomial time strategy for  
\[\ada x\bigl(p(x)\add \gneg p(x)\bigr)\mli \ada x\bigl(\gneg p(x)\add p(x)\bigr)\]
goes as follows. Wait till a value $n$ for $x$ is specified in the consequent. Then specify the same value for $x$ in the antecedent. Further wait till a 
$\add$-disjunct is selected in the antecedent. If the first (resp. second) disjunct is selected there, select the second (resp. first) $\add$-disjunct in the consequent, and celebrate victory.  On the other hand, we cannot say that every predicate $p(x)$ is also polynomial time reducible --- in the sense of  $\ada x\ade y\bigl(y=f(x)\bigr)$ --- to its complement. Take $p(x)$ to be any coNP-complete predicate. If it was polynomial time reducible to its complement, then we would have NP=coNP.

\section{The language of logic $\cltw$ and its semantics}\label{ss6}
%\marginpar{ss6}

Logic $\cltw$ will be axiomatically constructed in Section \ref{ss8}. The present section is merely devoted to its {\em language}. The building blocks of the formulas of the latter are:

\begin{itemize} 
\item {\bf Nonlogical predicate letters},\label{ipl} for which we use $p,q$  as metavariables. With each predicate letter is associated a fixed nonnegative integer called its {\bf arity}.\label{iar2} We assume that, for any $n$, there are infinitely many $n$-ary predicate letters.   
\item {\bf Function letters},\label{ifl} for which we use $f,g$ as metavariables. Again, each function letter comes with a fixed {\bf arity},\label{iar3} and  we assume that, for any $n$, there are infinitely many $n$-ary function letters.  
\item The binary {\bf logical predicate letter} $\equals $.
\item Infinitely many {\bf variables} and {\bf constants}.  These are the same as the ones fixed in Section \ref{nncg}.
\end{itemize}

{\bf Terms},\label{ipterm} for which we use $\tau,\psi,\xi,\chi,\theta,\eta$  as metavariables,  are built from variables, constants and function letters in the standard way.  An {\bf atomic formula} is $p(\tau_1,\ldots,\tau_n)$, where $p$ is an $n$-ary predicate letter and the $\tau_i$ are terms. 
When $p$ is $0$-ary, we write $p$ instead of $p()$. Also, we write $\tau_1\equals \tau_2$ instead of $\equals (\tau_1,\tau_2)$, and $\tau_1\notequals \tau_2$ instead of $\gneg (\tau_1\equals \tau_2)$. 
{\bf Formulas} are built from atomic formulas, propositional connectives $\twg,\tlg$ ($0$-ary), $\gneg$ ($1$-ary), $\mlc,\mld,\adc,\add$ ($2$-ary), variables and quantifiers $\cla,\cle,\ada,\ade$  in the standard way, with the exception that, officially, $\gneg$ is only allowed to be applied to atomic formulas. The definitions of {\em free} and {\em bound} occurrences of variables are also standard  (with $\ada,\ade$ acting as quantifiers along with $\cla,\cle$). A formula with no free occurrences of variables is said to be {\bf closed}.

Note that, terminologically, $\twg$ and $\tlg$ do not count as atoms. For us, atoms are formulas containing no logical operators. The formulas $\twg$ and $\tlg$ do not qualify because they {\em are} ($0$-ary) logical operators themselves.

$\gneg E$, where $E$ is not atomic, will be understood as a standard abbreviation: 
$\gneg\twg=\tlg$, $\gneg\gneg E= E$, $\gneg(A\mlc B)= \gneg A\mld \gneg B$, $\gneg \ada xE= \ade x\gneg E$, etc. And $E\mli F$ will be understood as an abbreviation of $\gneg E\mld F$.

Parentheses will often be omitted --- as we just did --- if there is no danger of ambiguity. When omitting parentheses, we assume that $\gneg$ and the quantifiers have the highest precedence, and $\mli$ has the lowest precedence. 
An expression $E_1\mlc\ldots\mlc E_n$, where $n\geq 2$, is to be understood as $E_1\mlc(E_2\mlc (\ldots\mlc(E_{n-1}\mlc E_n)\ldots ) )$. Sometimes we can write this  expression for an unspecified $n\geq 0$ (rather than $n\geq 2$). Such a formula, in the case of $n= 1$, should be understood as simply $E_1$. Similarly for $\mld,\adc,\add$.  As for the case of $n=0$, $\mlc$ and $\adc$ should be understood as $\twg$ while $\mld$ and $\add$ as $\tlg$.  

Sometimes a formula $F$ will be represented as $F(s_1,\ldots,s_n)$, where the $s_i$ are variables. 
When doing so, we do not necessarily mean that each  $s_i$ has a free occurrence in $F$, or that every variable occurring free in $F$ is among $s_1,\ldots,s_n$. However, it {\em will}  always be assumed (usually only implicitly) that the $s_i$ are pairwise distinct, and have no bound occurrences in $F$.  In the context set by the above representation, $F(\tau_1,\ldots,\tau_n)$ will mean the result of replacing, in $F$, each  occurrence of each $s_i$   by term $\tau_i$. When writing $F(\tau_1,\ldots,\tau_n)$, it will always be assumed (again, usually only implicitly) that the terms $\tau_1,\ldots,\tau_n$ contain no variables that have bound occurrences in $F$, so that there are no unpleasant collisions of variables when doing replacements.  

Similar --- well established in the literature --- notational conventions apply to terms.

A {\bf sequent} is an expression $E_1,\ldots,E_n\intimpl F$, where $E_1,\ldots,E_n$ ($n\geq 0$) and $F$ are formulas. Here $E_1,\ldots,E_n$ is said to be the {\bf antecedent} of the sequent, and $F$ said to be the {\bf succedent}. 

By a {\bf free} (resp. {\bf bound}) {\bf variable} of a sequent we shall mean a variable that has a free (resp. bound) occurrence in one of the formulas of the sequent. For safety and simplicity, throughout the rest of this paper we assume that the sets of all free and bound variables of any formula or sequent that we ever consider --- unless strictly implied otherwise by the context ---  are disjoint.   This restriction, of course, does not yield any loss of expressive power, as variables can always be renamed so as to satisfy this condition. 

An {\bf interpretation}\label{iint} is a pair $(U,^*)$, where $U=(U,^U)$ is a universe and $^*$ is a function that sends:
\begin{itemize}
\item  every $n$-ary function letter $f$ to a function \(f^*:\ U^n\rightarrow U\);
\item  every nonlogical $n$-ary  predicate letter $p$ to an $n$-ary predicate (elementary game) $p^*(s_1,\ldots,s_n)$ on $U$ which does not depend on any variables other than $s_1,\ldots,s_n$. 
 \end{itemize}

The above uniquely extends to  a  mapping that sends each term $\tau$ to a function $\tau^*$, and each formula $F$ to a game $F^*$, by stipulating that: 
\begin{enumerate}
\item $c^*=c^U$ (any constant $c$).
\item $s^* =  s$ (any variable $s$). 
\item Where $f$ is an $n$-ary function letter and $\tau_1,\ldots,\tau_n$ are terms, $\bigl(f(\tau_1,\ldots,\tau_n)\bigr)^* =  f^*(\tau_{1}^{*},\ldots,\tau_{n}^{*})$. 
\item Where   $\tau_1$ and $\tau_2$ are terms, $(\tau_1\equals \tau_2)^*$ is $\tau_{1}^{*}\equals \tau_{2}^{*}$. 
\item Where $p$ is an $n$-ary nonlogical  predicate letter  and $\tau_1,\ldots,\tau_n$ are terms, $\bigl(p(\tau_1,\ldots,\tau_n)\bigr)^* =  p^*(\tau_{1}^{*},\ldots,\tau_{n}^{*})$. 
\item $^{*}$ commutes with all logical operators, seeing them as the corresponding game operations: $\tlg^* =  \tlg$,  $(E_1\mlc\ldots\mlc E_n)^{*} =  E^{*}_{1}\mlc \ldots\mlc E^{*}_n$, $(\ada x E)^{*} =  \ada x(E^{*})$, etc. 
\end{enumerate}

While an interpretation is a pair $(U,^*)$, terminologically and notationally we will usually identify it with its second component and write $^*$ instead of $(U,^*)$, keeping in mind that every such ``interpretation'' $^*$ comes with a fixed universe $U$, said to be the {\bf universe of $^*$}.  When $O$ is a function letter, a predicate letter, a constant or a formula, and $O^* =  W$, we say that $^*$ {\bf interprets} $O$ as $W$. We can also refer to such a $W$ as 
``{\bf $O$ under interpretation $^*$}''.

When a given formula is represented as $F(x_1,\ldots,x_n)$, we will typically write $F^*(x_1,\ldots,x_n)$ instead of 
$\bigl(F(x_1,\ldots,x_n)\bigr)^*$. A similar practice will be used for terms as well.

We agree that, for a formula $F$, an interpretation $^*$ and an HPM $\cal M$, whenever we say that $\cal M$ is a {\bf  solution} of $F^*$ or write ${\cal M}\models F^*$, we mean that   $\cal M$ is a  solution of the (constant) game $\ada x_1\ldots\ada x_n(F^*)$, where $x_1,\ldots,x_n$ are exactly the free variables of $F$, listed according to their lexicographic order. We call the above game the {\bf $\ada$-closure} of $F^*$, and denote it by 
$\ada F^*$. The {\bf $\cla$-closure} $\cla F^*$ is defined similarly. The same notational convention extends from games to formulas. 

Note that, for any given   formula $F$, the $\legal{}{}$ component of the game $\ada F^*$ does not depend on the interpretation $^*$. Hence we can safely 
say ``legal run of $\ada F$'' 
without indicating an interpretation applied to the formula.

\section{The axiomatization of logic $\cltw$}\label{ss8}
%\marginpar{ss8}

 Our formulations  rely on some terminology and notation, explained below. 

A formula not containing any choice operators $\adc,\add,\ada,\ade$ --- i.e., a formula of the language of classical first order logic --- is said to be {\bf elementary}. A sequent is {\bf elementary} iff all of its formulas are so. 

The {\bf elementarization} \[\elz{F}\] of a formula $F$ is the result of replacing
in $F$ all $\add$- and $\ade$-subformulas by $\tlg$, and all $\adc$- and $\ada$-subformulas by $\twg$. Note that $\elz{F}$ is (indeed) an elementary formula. The {\bf elementarization} $\elz{G_1,\ldots,G_n\intimpl F}$ of a sequent 
$G_1,\ldots,G_n\intimpl F$ is the elementary formula \(\elz{G_1}\mlc\ldots\mlc \elz{G_n}\mli \elz{F}.\) 

A sequent  is said to be {\bf stable} iff its elementarization is classically valid.  By ``classical validity'', in view of G\"{o}del's completeness theorem,  we mean provability in some standard classical first-order calculus with constants, function letters and $\equals$, where $\equals$ is treated as the logical {\em identity} predicate (so that, say, $x\equals x$, $x\equals y\mli (E(x)\mli E(y))$, etc. are provable).

A {\bf surface occurrence} of a subformula is an occurrence that is not in the scope of any choice operators.
 
We will be using the notation \[F[E]\] to mean a formula $F$ together with some (single) fixed  surface occurrence of a subformula $E$. Using this notation sets a context, in which $F[H]$ will mean the result of replacing in $F[E]$ the (fixed) occurrence of $E$ by $H$.  Note that here we are talking about some {\em occurrence} of $E$. Only that occurrence gets replaced when moving from $F[E]$ to $F[H]$, even if the formula also had some other occurrences of $E$.

By a {\bf rule} (of inference) in this section we mean a binary relation $\mathbb{Y}{\cal R} X$, where $\mathbb{Y}=\seq{Y_1,\ldots,Y_n}$ is a finite sequence of sequents and $X$ is a sequent. Instances of such a relation are schematically written as 
\[\frac{Y_1,\ldots,Y_n}{X},\]
where $Y_1,\ldots,Y_n$ are called the {\bf premises}, and $X$ is  called the {\bf conclusion}. Whenever $\mathbb{Y}{\cal R}X$ holds, we say that $X$ {\bf follows} from $\mathbb{Y}$ by $\cal R$.  

Expressions such as $\vec{G},\vec{K},\ldots$ will usually stand for finite sequences of formulas. The standard meaning of an expression such as $\vec{G},F,\vec{K}$  should also be clear.

\begin{center}
\begin{picture}(100,30)

\put(0,10){\bf THE RULES OF $\cltw$}

\end{picture}
\end{center}

$\cltw$ has the six rules listed below, with the following additional conditions/explanations: 
\begin{enumerate}
\item In $\add$-Choose and $\adc$-Choose, $i\in\{0,1\}$.
\item In $\ade$-Choose and $\ada$-Choose,  $\mathfrak{t}$ is either a constant or a variable with no bound occurrences in the premise, and $H(\mathfrak{t})$ is the result of replacing by $\mathfrak{t}$ all free occurrences of $x$ in $H(x)$ (rather than vice versa).
\end{enumerate}
\begin{center}
\begin{picture}(287,70)

\put(14,50){\bf $\add$-Choose}
\put(12,30){$\vec{G}\ \intimpl\  F[H_i]$}
\put(0,22){\line(1,0){78}}
\put(0,8){$\vec{G}\ \intimpl \ F[H_0\add H_1]$}

\put(232,50){\bf $\adc$-Choose}
\put(212,30){$\vec{G},\ E[H_i],\ \vec{K}\  \intimpl \ F$}
\put(200,22){\line(1,0){113}}
\put(200,8){$\vec{G},\ E[H_0\adc H_1], \ \vec{K}\ \intimpl\ F$}

\end{picture}
\end{center}

\begin{center}
\begin{picture}(287,70)

\put(231,50){\bf $\ada$-Choose}
\put(210,30){$\vec{G},\ E[H(\mathfrak{t})],\ \vec{K}\ \intimpl\ F$}
\put(200,22){\line(1,0){113}}
\put(200,8){$\vec{G},\ E[\ada xH(x)],\ \vec{K}\ \intimpl\ F$}

\put(16,50){\bf $\ade$-Choose}
\put(08,30){$\vec{G}\ \intimpl\  F[H(\mathfrak{t})]$}
\put(0,22){\line(1,0){78}}
\put(0,8){$\vec{G}\ \intimpl \ F[\ade x H(x)]$}

\end{picture}
\end{center}

\begin{center}
\begin{picture}(74,70)

\put(12,50){\bf Replicate}
\put(8,8){$\vec{G},E,\vec{K}\intimpl F$}
\put(0,22){\line(1,0){69}}
\put(0,30){$\vec{G},E,\vec{K},E\intimpl F$}
\end{picture}
\end{center}

\begin{center}
\begin{picture}(300,70)
\put(140,50){\bf Wait}
\put(0,30){$Y_1,\ldots,Y_n$}
\put(0,22){\line(1,0){45}}
\put(55,20){($n\geq 0$), where all of the following five conditions are satisfied:}
\put(20,8){$X$}
\end{picture}
\end{center}

\begin{enumerate}
\item {\bf $\adc$-Condition:}  Whenever $X$ has the form $\vec{G}\intimpl F[H_0\adc H_1]$, both of the sequents $\vec{G}\intimpl F[H_0]$ and 
$\vec{G}\intimpl F[H_1]$ are among $Y_1,\ldots,Y_n$.
\item {\bf $\add$-Condition:} Whenever $X$ has the form $\vec{G},E[H_0\add H_1],\vec{K}\intimpl F$, both of the sequents $\vec{G},E[H_0],\vec{K}\intimpl F$ and 
$\vec{G},E[H_1],\vec{K}\intimpl F$ are among $Y_1,\ldots,Y_n$.
\item {\bf $\ada$-Condition:} Whenever $X$ has the form $\vec{G}\intimpl F[\ada xH(x)]$, for some variable $y$ not occurring in $X$, the sequent  $\vec{G}\intimpl F[H(y)]$ is among  $Y_1,\ldots,Y_n$. Here and below, $H(y)$ is the result of replacing by $y$ all free occurrences of $x$ in $H(x)$ (rather than vice versa).
\item {\bf $\ade$-Condition:} Whenever $X$ has the form $\vec{G},E[\ade xH(x)],\vec{K}\intimpl F$, for some variable $y$ not occurring in $X$, the sequent  $\vec{G},E[H(y)],\vec{K}\intimpl F$ is among  $Y_1,\ldots,Y_n$.
\item {\bf Stability Condition:} $X$ is stable.
\end{enumerate}

A {\bf $\cltw$-proof} of a sequent $X$ is a sequence $X_1,\ldots,X_n$ of sequents, with $X_n=X$, such that, each $X_i$ follows  by one of the rules of $\cltw$ from some (possibly empty in the case of Wait, and certainly empty in the case of $i=1$) set $\cal P$ of premises such that ${\cal P}\subseteq \{X_1,\ldots, X_{i-1}\}$.
When a $\cltw$-proof of $X$ exists, we say that $X$ is {\bf provable} in $\cltw$, and write $\cltw\vdash X$.

   A {\bf $\cltw$-proof} of a formula $F$ will be understood as a  $\cltw$-proof of the empty-antecedent sequent $\intimpl F$. Accordingly, $\cltw\vdash F$ means $\cltw\vdash\intimpl F$.

 $\cltw$ is a conservative extension of classical logic (see \cite{Japtowards}). Namely,  an elementary sequent $E_1,\ldots,E_n\intimpl F$ is provable in $\cltw$ iff the formula $E_1\mlc\ldots\mlc E_n\mli F$ is valid in the classical sense. It is also a conservative extension of the earlier known logic {\bf CL3} studied in \cite{Japtcs}.\footnote{An essentially the same logic, under the name {\bf L}, was in fact known as early as in \cite{Jap02}.} The latter is nothing but the empty-antecedent fragment of $\cltw$ without function letters and identity. 

\begin{example}\label{ecube} In this example, $\mult$ is a binary function letter and $^3$ is a unary function letter. We write $x\mult y$ and $x^3$ instead of $\mult(x,y)$ and $^3(x)$, respectively. The following sequence of sequents is a $\cltw$-proof (of its last sequent):\vspace{7pt}

\noindent 1. $\cla x \bigl(x^3\equals (x\mult x)\mult x\bigr),\    t\equals s\mult s, \  r\equals t\mult s \ \intimpl \  r\equals s^3$  \ \ {Wait: (no premises)} \vspace{3pt}

\noindent 2. $\cla x \bigl(x^3\equals (x\mult x)\mult x\bigr),\    t\equals s\mult s, \  r\equals t\mult s \ \intimpl \  \ade y(y\equals s^3)$  \ \   {$\ade$-Choose: 1}\vspace{3pt}

\noindent 3. $\cla x \bigl(x^3\equals (x\mult x)\mult x\bigr),\    t\equals s\mult s, \  \ade z (z\equals t\mult s) \ \intimpl \  \ade y(y\equals s^3)$  \ \ {Wait: 2} \vspace{3pt}

\noindent 4. $\cla x \bigl(x^3\equals (x\mult x)\mult x\bigr),\    t\equals s\mult s, \ \ada y \ade z (z\equals t\mult y) \ \intimpl \  \ade y(y\equals s^3)$  \ \ {$\ada$-Choose: 3}\vspace{3pt}

\noindent 5. $\cla x \bigl(x^3\equals (x\mult x)\mult x\bigr),\    t\equals s\mult s, \ \ada x\ada y \ade z (z\equals x\mult y) \ \intimpl \  \ade y(y\equals s^3)$  \ \ {$\ada$-Choose: 4}\vspace{3pt}

\noindent 6. $\cla x \bigl(x^3\equals (x\mult x)\mult x\bigr),\    \ade z (z\equals s\mult s), \ \ada x\ada y \ade z (z\equals x\mult y) \ \intimpl \  \ade y(y\equals s^3)$  \ \ {Wait: 5}\vspace{3pt}

\noindent 7. $\cla x \bigl(x^3\equals (x\mult x)\mult x\bigr),\  \ada y \ade z (z\equals s\mult y), \ \ada x\ada y \ade z (z\equals x\mult y) \ \intimpl \  \ade y(y\equals s^3)$  \ \ { $\ada$-Choose: 6}\vspace{3pt}

\noindent 8.  $\cla x \bigl(x^3\equals (x\mult x)\mult x\bigr),\ \ada x\ada y \ade z (z\equals x\mult y), \ \ada x\ada y \ade z (z\equals x\mult y) \ \intimpl \  \ade y(y\equals s^3)$ \ \ { $\ada$-Choose: 7}\vspace{3pt}

\noindent 9.  $\cla x \bigl(x^3\equals (x\mult x)\mult x\bigr),\ \ada x\ada y \ade z (z\equals x\mult y) \ \intimpl \ \ade y(y\equals s^3)$ \ \ {Replicate: 8}\vspace{3pt}

\noindent 10. $\cla x \bigl(x^3\equals (x\mult x)\mult x\bigr),\ \ada x\ada y \ade z (z\equals x\mult y) \ \intimpl \ \ada x\ade y(y\equals x^3)$ \ \ { Wait: 9}
\end{example}

\begin{example}\label{j28a}
%\marginpar{j28a}
The formula $\cla x\hspace{1pt}p(x)\mli\ada x\hspace{1pt}p(x)$ is provable in $\cltw$. It follows  from $\cla x\hspace{1pt}p(x)\mli p(y)$ by Wait. The latter, in turn, follows by Wait from the empty set of premises. 

On the other hand, the formula $\ada x\hspace{1pt}p(x)\mli\cla x\hspace{1pt}p(x)$, i.e. $\ade x\gneg p(x)\mld \cla x\hspace{1pt}p(x)$, in not provable. Indeed, its elementarization is $\tlg\mld \cla x\hspace{1pt}p(x)$, which is not classically valid.  Hence $\ade x\gneg p(x)\mld \cla x\hspace{1pt}p(x)$ cannot be derived by Wait. Replicate can also be dismissed for obvious reasons. This leaves us with $\ade$-Choose. But if $\ade x\gneg p(x)\mld \cla x\hspace{1pt}p(x)$ is derived  by $\ade$-Choose, then the premise should be $\gneg p(\mathfrak{t})\mld \cla x\hspace{1pt}p(x)$ for some variable or constant $\mathfrak{t}$. The latter, however, is a classically non-valid elementary formula and hence unprovable. 
\end{example}

\begin{example}\label{j28b}
%\marginpar{j28b} 
The formula $\ada x\ade y\bigl(p(x)\mli p(y)\bigr)$ is provable in $\cltw$ as follows:\vspace{7pt}

\noindent 1. $\begin{array}{l}
p(s)\mli p(s)
\end{array}$  \ \ Wait:\vspace{3pt}

\noindent 2. $\begin{array}{l}
\ade y\bigl(p(s)\mli p(y)\bigr)
\end{array}$  \ \ $\ade$-Choose: 1\vspace{3pt}

\noindent 3. $\begin{array}{l}
\ada x\ade y\bigl(p(x)\mli p(y)\bigr)
\end{array}$  \ \ Wait: 2\vspace{7pt}

On the other hand, the formula $\ade y\ada x \bigl(p(x)\mli p(y)\bigr)$ can be seen to be unprovable, even though its classical counterpart $\cle y\cla x \bigl(p(x)\mli p(y)\bigr)$ is a classically valid elementary formula and hence provable in $\cltw$.  
\end{example}

\begin{example}\label{j28c}
%\marginpar{j28c} 
While the formula $\cla x\cle y \bigl(y\equals f(x)\bigr)$  is classically valid and hence provable in $\cltw$, its constructive counterpart 
$\ada x\ade y \bigl(y\equals f(x)\bigr)$ can be easily seen to  be unprovable. This is no surprise. In view of the expected soundness of $\cltw$,  provability  of $\ada x\ade y \bigl(y\equals f(x)\bigr)$ would imply that every function $f$ is computable, which, of course, is not the case.     
\end{example}

\begin{exercise}\label{feb1a}
%\marginpar{feb1a}
To see the resource-consciousness of $\cltw$, show that it does not prove $p\adc q\mli (p\adc q)\mlc (p\adc q)$, even though this formula has the form $F\mli F\mlc F$ of a classical tautology. Then show that, in contrast, $\cltw$ proves the {\em sequent} $p\adc q\intimpl (p\adc q)\mlc (p\adc q)$ because, unlike the antecedent of a $\mli$-combination, the antecedent of a $\intimpl$-combination is reusable (trough Replicate). 
\end{exercise}

\begin{exercise}\label{feb1ae}
%\marginpar{feb1ae}
Show that $\cltw\vdash \ade x\ada y\hspace{2pt} p(x,y)\intimpl \ade x\bigl(\ada y\hspace{2pt}p(x,y)\mlc \ada y\hspace{2pt}p(x,y)\bigr)$. Then observe that, on the other hand,  $\cltw$ does not prove any of the formulas 
\[\begin{array}{rcl}
\ade x\ada y\hspace{2pt} p(x,y) & \mli & \ade x\bigl(\ada y\hspace{2pt}p(x,y)\mlc \ada y\hspace{2pt}p(x,y)\bigr);\\
\ade x\ada y\hspace{2pt} p(x,y)\ \mlc \ \ade x\ada y\hspace{2pt} p(x,y) & \mli & \ade x\bigl(\ada y\hspace{2pt}p(x,y)\mlc \ada y\hspace{2pt}p(x,y)\bigr);\\
\ade x\ada y\hspace{2pt} p(x,y)\ \mlc\ \ade x\ada y\hspace{2pt} p(x,y)\ \mlc\ \ade x\ada y\hspace{2pt} p(x,y) & \mli & \ade x\bigl(\ada y\hspace{2pt}p(x,y)\mlc \ada y\hspace{2pt}p(x,y)\bigr);\\
 & \ldots & 
\end{array}\]
\end{exercise}

\section{The adequacy of logic $\cltw$}\label{sadeq} 
%\marginpar{sadeq}
{\bf Logical Consequence} ({\bf LC}) is the following rule, with both the premises and the conclusion being formulas:    
\[\mbox{\em From $E_1,\ldots,E_n$ conclude $F$, as long as $\cltw$ proves {\em $E_1,\ldots,E_n\intimpl F$}}.\]
When $F$ follows from $E_1,\ldots,E_n$ by this rule, i.e., when $\cltw \vdash E_1,\ldots,E_n\intimpl F$, we say that $F$ is a {\bf logical consequence} of $E_1,\ldots,E_n$.     

The official terms of the language of $\cltw$, identified with their parse trees, are tree-style structures, so let us call them {\bf tree-terms}. A more general and economical way to represent terms, however, is to allow merging some or all identical-content nodes in such trees, thus turning them into (directed, acyclic,  rooted, edge-ordered multi-) graphs.  Let us call these unofficial sorts of terms   {\bf graph-terms}. The idea of representing linguistic objects in the form of graphs rather than trees is central in the approach called {\em cirquent calculus} (\cite{Cirq,Japdeep}), and has already  proven its worth. We once again find  the usefulness of that idea  in our present, complexity-sensitive context. Figure 1 illustrates two terms representing the same polynomial function $y^8$, with the term on the right being a tree-term and the term on the left being a graph-term. As this example suggests, graph-terms are generally exponentially smaller than the corresponding tree-terms, which explains our interest in the former.  Figure 1 also makes it unnecessary to formally define graph-terms, as their meaning must be perfectly clear after looking at this single example.   

\begin{center} \begin{picture}(353,130)

\put(16,112){\circle{12}}
\put(14,111){$y$}

\put(22,90){\line(1,1){10}}
\put(32,100){\vector(-1,1){10}}
\put(10,90){\line(-1,1){10}}
\put(0,100){\vector(1,1){10}}
\put(16,86){\circle{12}}
\put(12,83){$\mult$}

\put(22,64){\line(1,1){10}}
\put(32,74){\vector(-1,1){10}}
\put(10,64){\line(-1,1){10}}
\put(0,74){\vector(1,1){10}}
\put(16,60){\circle{12}}
\put(12,57){$\mult$}

\put(22,38){\line(1,1){10}}
\put(32,48){\vector(-1,1){10}}
\put(10,38){\line(-1,1){10}}
\put(0,48){\vector(1,1){10}}
\put(16,34){\circle{12}}
\put(12,31){$\mult$}

\put(104,112){\circle{12}}
\put(102,111){$y$}
\put(137,112){\circle{12}}
\put(135,111){$y$}
\put(116,90){\vector(-2,3){10}}
\put(125,90){\vector(2,3){10}}
\put(121,86){\circle{12}}
\put(117,83){$\mult$}

\put(174,112){\circle{12}}
\put(172,111){$y$}
\put(207,112){\circle{12}}
\put(205,111){$y$}
\put(186,90){\vector(-2,3){10}}
\put(195,90){\vector(2,3){10}}
\put(191,86){\circle{12}}
\put(187,83){$\mult$}

\put(151,64){\vector(-3,2){26}}
\put(161,64){\vector(3,2){26}}
\put(156,60){\circle{12}}
\put(152,57){$\mult$}

\put(244,112){\circle{12}}
\put(242,111){$y$}
\put(277,112){\circle{12}}
\put(275,111){$y$}
\put(256,90){\vector(-2,3){10}}
\put(265,90){\vector(2,3){10}}
\put(261,86){\circle{12}}
\put(257,83){$\mult$}

\put(314,112){\circle{12}}
\put(312,111){$y$}
\put(347,112){\circle{12}}
\put(345,111){$y$}
\put(326,90){\vector(-2,3){10}}
\put(335,90){\vector(2,3){10}}
\put(331,86){\circle{12}}
\put(327,83){$\mult$}

\put(291,64){\vector(-3,2){26}}
\put(301,64){\vector(3,2){26}}
\put(296,60){\circle{12}}
\put(292,57){$\mult$}

\put(221,38){\vector(-4,1){63}}
\put(231,38){\vector(4,1){63}}
\put(226,34){\circle{12}}
\put(222,31){$\mult$}

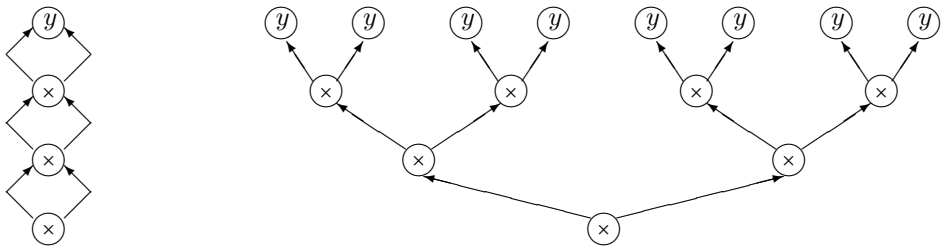
\put(50,10){{\bf Figure 1:} A graph-term and the corresponding tree-term}
\end{picture}\end{center}

By a {\bf polynomial graph-term} $\tau$ we shall mean a graph-term  not containing (at its leaves) any constants other than $0$, and not containing (at its internal nodes) any function letters other than $\successor$ (unary), $\plus$ (binary) and $\mult$ (binary).  The total number $k$ of the variables $y_1,\ldots,y_k$ occurring in (at the leaves of) $\tau$   is said to be the {\bf arity} of $\tau$. Subsequently we shall only be interested in unary ($1$-ary) polynomial graph-terms, and will typically omit the word ``unary''. Terminologically and notationally we shall usually identify such a term $\tau$ with the unary polynomial arithmetical function represented by it under the standard arithmetical interpretation ($x\successor$ means $x\plus 1$). So, for instance, either term of Figure 1 is a polynomial graph-term, representing --- and identified with --- the function $y^8$.

We generalize the above concept of a (unary) polynomial graph-term $\tau$ to that of a $(1,n)$-ary {\bf explicit polynomial functional} by allowing the term  $\tau$ to contain, on top of variables and $0,\successor,\plus,\mult$, additional $n$ ($n\geq 0$) unary function letters $f_1,\ldots,f_n$,  semantically treated as placeholders for unary arithmetical functions.  We say that such a $\tau$ {\bf depends on} $f_1,\ldots,f_n$.  Replacing $f_1,\ldots,f_n$ by names $g_1,\ldots,g_n$ of some particular unary arithmetical functions turns $\tau$ into  the corresponding unary arithmetical function, which we shall denote by $\tau(g_1,\ldots,g_n)$.  For instance, the term of Figure 2 is a 
$(1,2)$-ary explicit polynomial functional. Let us denote it by $\tau$. Then, if $g$ means ``square'' and $h$ means ``cube'',  $\tau(g,h)$ is  the unary arithmetical function  
 $(y^2\plus y^3)^3$. 

\begin{center} \begin{picture}(353,123)

\put(157,105){\circle{14}}
\put(155,104){$y$}

\put(144,87){\vector(1,1){11}}
\put(138,82){\circle{14}}
\put(134,80){$f_1$}

\put(170,87){\vector(-1,1){11}}
\put(176,82){\circle{14}}
\put(172,80){$f_2$}

\put(152,65){\vector(-1,1){11}}
\put(162,65){\vector(1,1){11}}

\put(157,60){\circle{14}}
\put(153,58){$\plus$}
\put(157,42){\vector(0,1){11}}
\put(157,34){\circle{14}}
\put(153,32){$f_2$}

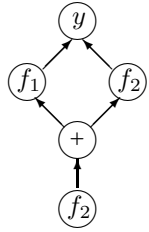
\put(0,10){{\bf Figure 2:} An explicit polynomial functional expressing $f_2\bigl(f_1(y)\plus f_2(y)\bigr)$}
\end{picture}\end{center}

We further generalize the concept of a (unary) polynomial graph-term $\tau$ to that of a (unary) {\bf explicit polynomial function}. The latter is defined as a nonempty sequence $\seq{\tau_{f_1},\ldots,\tau_{f_k}}$ of explicit polynomial functionals indexed by (associated with) pairwise distinct unary function letters $f_1,\ldots,f_k$, where each $\tau_{f_i}$ does not depend on any function letters that are not among $f_1,\ldots,f_{i-1}$. Note that this condition makes $\tau_{f_1}$ simply a polynomial graph-term, i.e., a $(1,0)$-ary explicit polynomial functional. (Meta)semantically, each index $f_i$ here is a name of (represents) a unary arithmetical function, and the corresponding $\tau_{f_i}$ is a definition of that function in terms of $0$, $\successor$, $\plus$, $\mult$ and some earlier-defined functions; then  $\tau$ itself is stipulated to represent the same function as  $f_k$ does. Again, a single example would be sufficient to fully clarify the denotation of each explicit polynomial function. Consider the explicit polynomial function  $\tau=\seq{\tau_{f_1},\tau_{f_1},\tau_{f_3}}$, where $\tau_{f_1}$ and $\tau_{f_2}$ are the $(1,0)$-ary explicit polynomial functionals of Figure 1, and $\tau_{f_3}$ is the $(1,2)$-ary explicit polynomial functional of Figure 2. Thus, both $f_1$ and $f_2$ represent the same unary polynomial function $y^8$, and $f_3$, i.e. $\tau$ itself, represents the unary polynomial function  $(y^8\plus y^8)^8$. Of course, every explicit polynomial function can be translated into an equivalent polynomial graph-term, but such a translation can increase the size exponentially. For our purposes, polynomial graph-terms (let alone tree-terms) are too inefficient means of representing polynomial functions. This explains our preference for explicit polynomial functions as a standard way of writing (in our metalanguage) polynomial terms.

As in the case of polynomial graph-terms, terminologically and notationally we shall usually identify an explicit polynomial function $\tau$ with the unary arithmetical function represented by it.  When $\tau$ is an explicit polynomial function and $\cal M$ is a $\tau$ time (resp. space) machine, we  say that $\tau$ is an {\bf explicit polynomial bound} for the time (resp. space) complexity of $\cal M$.

Another auxiliary concept that we are going to rely on in this section and later is that of a {\bf generalized HPM} ({\bf GHPM}). For a natural number $n$, an $n$-ary GHPM is defined in the same way as an HPM, with the difference that the former takes $n$ natural numbers as inputs (say, provided on a separate, read-only  {\em input tape}); such  inputs are present  at the very beginning of the work of the machine and remain unchanged throughout it. An ordinary HPM is thus nothing but a $0$-ary GHPM.  When $\cal M$ is an $n$-ary  GHPM and $c_1,\ldots,c_n$ are natural numbers, ${\cal M}(c_1,\ldots,c_n)$ denotes the HPM that works just like $\cal M$ in the scenario where the latter has received $c_1,\ldots,c_n$ as inputs. 
We will assume that some reasonable encoding (through natural numbers) of GHPMs is fixed. When $\cal M$ is a GHPM, $\code{{\cal M}}$ denotes its {\bf code}.

The paper \cite{Japlbcs} established the soundness  of $\cltw$ in the following strong sense:

\begin{theorem}\label{feb9d}  
%\marginpar{feb9d}
(\cite{Japlbcs}) If a formula $F$ is a logical consequence of formulas $E_1,\ldots,E_n$ and $^*$ is an interpretation such that each $E_{i}^{*}$ ($1\leq i\leq n$) is computable, then so is $F^*$.
Furthermore:
\begin{enumerate} 
\item There is an  efficient\footnote{Here and later in similar metacontexts, ``efficient'' means ``polynomial time''.}  procedure that takes an arbitrary $\cltw$-proof of an arbitrary sequent {\em $E_1,\ldots,E_n\intimpl F$}
 and constructs a $n$-ary GHPM $\cal M$, together with a $(1,n)$-ary explicit polynomial functional $\tau$,    such that, for any interpretation $^*$, any 
$n$-ary  GHPMs ${\cal N}_1,\ldots,{\cal N}_n$   and any unary arithmetical functions $g_1,\ldots,g_n$,   if each ${\cal N}_i(\code{{\cal N}_1},\ldots,\code{{\cal N}_n})$ is a $g_{i}$ time solution of $E_{i}^{*}$, then ${\cal M}(\code{{\cal N}_1},\ldots,\code{{\cal N}_n})$ is a $\tau(g_1,\ldots,g_n)$ time solution of $F^*$. 
\item The same holds for ``space'' instead of ``time''.   
\end{enumerate}
\end{theorem}

Among the  corollaries of the above theorem is that LC preserves both polynomial time and polynomial space computabilities, as well as $\Omega$-time and $\Omega$-space computabilities for any class $\Omega$ of functions containing all polynomial functions and closed under composition (such as, say, the class of all primitive recursive functions). Indeed, to see that LC preserves (for instance) polynomial time computability, assume $F$ is a logical consequence of $E_1,\ldots,E_n$, and the HPMs ${\cal N}_1,\ldots,{\cal N}_n$ are polynomial time solutions of $E_1,\ldots,E_n$. Of course, every such HPM ${\cal N}_i$ can as well be seen as an $n$-ary GHPM which simply ignores its inputs. So, each ${\cal N}_i(\code{{\cal N}_1},\ldots,\code{{\cal N}_n})$ solves $E_i$. Then, according to Theorem \ref{feb9d}, we can construct a polynomial time solution ${\cal M}(\code{{\cal N}_1},\ldots,\code{{\cal N}_n})$ of $F$. Furthermore, such a solution, together with an explicit polynomial bound for its time complexity, can be efficiently extracted from ${\cal N}_1,\ldots,{\cal N}_n$, explicit polynomial bounds for their time complexities, and the $\cltw$-proof of $E_1,\ldots,E_n\intimpl F$.   

Remember that, philosophically speaking, computational {\em resources} are symmetric to computational problems: what is a problem for one player to solve is a resource that the other player can use. Namely, having a problem $A$ as a computational resource intuitively means having the ability  to successfully solve/win $A$. For instance, as a resource, $\ada x\ade y(y=x^2)$ means the ability to tell the square of any number. 

Together with soundness, \cite{Japlbcs} also established the completeness  of $\cltw$. While in the present paper we treat 
$\intimpl$ merely as a syntactic expression separating the two parts of a sequent, in \cite{Japlbcs} it is seen as an operation on games, called  {\bf ultimate reduction}. For simplicity considerations, we do not want to reproduce the definition of $\intimpl$ here, which would be necessary to precisely state the completeness result for $\cltw$. We shall only point out that, intuitively, $A_1,\ldots,A_n\intimpl B$ is (indeed) the problem of reducing $B$ to $A_1,\ldots,A_n$ in the most general intuitive sense. It is similar to $A_1\mlc \ldots\mlc A_n\mli B$, but with the difference that, during playing this game, in the former $\pp$ can use any of the antecedental resources $A_i$ repeatedly, while in the latter at most once. The $A_i$s, as resources, are ``recyclable'', that is. Furthermore, they are ``recyclable'' in the strongest sense possible. Namely,  it is not necessary to restart $A_i$ from the beginning every time it is reused. $\pp$ may as well choose to continue $A_i$ --- in a new way --- from any of the previously reached positions. This corresponds to the way of reusage any purely software resource would offer in the presence of an advanced operating system and unlimited memory: one can start running process $A_i$, then fork it at any stage thus creating two  copies (branches) with a common past but possibly diverging futures, then fork any of the new branches again at any time, and so on. See \cite{Japfin} or \cite{Japlbcs} for more details and explanations. 

Anyway, the completeness result of \cite{Japlbcs} for $\cltw$ says that, if there is an HPM that solves  the problem $E^{*}_{1},\ldots,E_{n}^{*}\intimpl F^{*}$ for any interpretation $^*$ --- such an $^*$-independent HPM is said to be a {\em logical solution} of the sequent $E_{1},\ldots,E_{n}\intimpl F$ --- then $\cltw\vdash E_{1},\ldots,E_{n}\intimpl F$.     Furthermore, the same has been shown to hold even if HPMs are no longer required to follow algorithmic (let alone efficient) strategies --- for instance, if they are allowed to use oracles for whatever functions. This result, in view of the intuitions captured by the operation $\intimpl$, eventually translates in \cite{Japlbcs} into the following thesis relevant to our further purposes:

\begin{thesis}\label{thesis}
%\marginpar{thesis}
(\cite{Japlbcs}) Assume  $E_1,\ldots,E_n,F$ are formulas such that  there is a $^*$-independent (whatever interpretation $^*$) intuitive description and justification of a winning strategy for $F^*$, which relies on the availability and ``recyclability'' --- in the strongest sense possible --- of $E_{1}^{*},\ldots,E_{n}^{*}$ as computational resources.   Then $F$ is a logical consequence of $E_1,\ldots,E_n$. 
\end{thesis}

The above means that the rule of Logical Consequence lives up to its name, and that we can always reliably use intuition on games and strategies when reasoning about them in $\cltw$-based systems (where LC is a rule of inference) such as $\arfour$. Namely, once a formula $E$ is proven in such a system, it can be treated as a recyclable resource whose unlimited availability can be safely assumed for any new strategies that we construct. More precisely, a new strategy may assume that there is an (external) {\bf provider}  of the resource $E$ (call it an {\bf oracle} for $E$ if you prefer), capable of successfully playing $E$ for (against) us any time and in any number of sessions, whether we choose those sessions to evolve in a parallel or a branched/forked fashion.    

\begin{example}\label{intex}
%\marginpar{intex}
 Imagine a $\cltw$-based applied formal theory, in which we have already proven two facts: $\cla x\bigl(x^3\equals(x\mult x)\mult x\bigr)$ (the meaning of ``cube'' in terms of multiplication) and $\ada x\ada y\ade z(z\equals x\mult y)$ (the computability of multiplication), and now we want to derive $\ada x\ade y(y\equals x^3)$ (the computability of ``cube''). This is how we can reason to justify $\ada x\ade y(y\equals x^3)$:  
\begin{quote}{\em Consider any $s$ 
(selected by Environment for $x$ in $\ada x\ade y(y\equals x^3)$). We need to find  $s^3$. Using the resource $\ada x\ada y\ade z(z\equals x\mult y)$ twice, we first find the value $t$ of $s\mult s$, and then  the value $r$ of $t\mult s$. According to $\cla x(x^3\equals (x\mult x)\mult x)$, such an $r$ is the sought $s^3$.}
\end{quote}   

Thesis \ref{thesis} promises that the above intuitive argument will be translatable into a $\cltw$-proof of 
\[\cla x\bigl(x^3\equals(x\mult x)\mult x\bigr),\ \ada x\ada y\ade z(z\equals x\mult y)\ \intimpl \ \ada x\ade y(y\equals x^3) \]
(and hence the succedent will be derivable in the  theory by LC as the formulas of the antecedent are already proven). Such a proof indeed exists --- see 
 Example \ref{ecube}.
\end{example}

 $\cltw$ --- more precisely, the associated rule of LC --- is adequate because, on one hand, by Theorem \ref{feb9d}, it is sound for a wide spectrum of applied theories, including --- but not limited to --- polynomial-time-oriented ones,  and,   on the other hand, by   Thesis \ref{thesis},  it is as strong as a logical rule of inference could possibly be. 

\section{Theory $\arfour$  introduced}\label{ss11}
%\marginpar{ss11}

The language of $\arfour$, whose {\bf formulas} more specifically can be referred to as {\bf $\arfour$-formulas}, is obtained from the language of $\cltw$ by removing all nonlogical predicate letters (thus only leaving the logical predicate letter $\equals $),  removing all constants but $0$, and  removing all but three function letters, which are:
 
\begin{itemize}
\item $successor$, unary. We will write $\tau\successor $ for $successor(\tau)$.
\item $sum$, binary. We will write $\tau_1\plus \tau_2$ for $sum(\tau_1,\tau_2)$.
\item $product$, binary. We will write $\tau_1\mult  \tau_2$ for $product(\tau_1,\tau_2)$.
\end{itemize}

Thus, the language of $\arfour$ extends that of {\bf Peano arithmetic} $\pa$ (see, for example, \cite{Hajek}) through adding to it $\adc,\add,\ade,\ada$. 
Formulas that have no free occurrences of variables are said to be {\bf sentences}.\label{isentence}  

The concept of an interpretation explained earlier can now be restricted to interpretations that are only defined on  $\successor $, $\plus $ and  $\mult $, as the present language has no other nonlogical function or predicate letters. Of such interpretations, the {\bf standard interpretation}\label{isi} $^\dagger$ is the one whose universe is the ideal universe $\{0,1,10,11,100,\ldots\}$, with its elements identified with the corresponding natural numbers,  and which  interprets  $\successor $ as the standard successor ($x\plus 1$) function, interprets $\plus $ as the sum function,  and interprets $\mult $ as the product function.  For a 
$\arfour$-formula $F$, the {\bf standard interpretation of} 
$F$ is the game $F^\dagger$, which we typically write simply as $F$ unless doing so may cause ambiguity. 

Where $\tau$ is a term, we will be using $\tau \zero$ and $\tau \one$ as abbreviations for the terms $0\successor\successor\mult \tau$ and $(0\successor\successor\mult \tau)\successor$, respectively. The choice of this notation is related to the fact that, given any natural number $a$, the binary representation of $0\successor\successor \mult a$ (i.e., of $2a$) is nothing but the binary representation of $a$ with a ``$0$'' added on its right. Similarly, the binary representation of $(0\successor\successor\mult a)\successor$ is nothing but the binary representation of $a$ with a ``$1$'' added to it. Of course, here an exception is the case $a\equals 0$. It can be made an ordinary case by assuming that adding any number of $0$s at the beginning of a binary numeral $b$ results in a legitimate numeral representing the same number as $b$ does. 

The number $a\zero$ (i.e. $2a$) will be said to be the {\bf binary $0$-successor}\label{ibzs} of $a$, and $a\one $ (i.e. $2a+1$) said to be the {\bf binary $1$-successor}\label{ibos} of $a$; in turn, we can refer to $a$ as the {\bf binary predecessor}\label{ibp} of $a\zero$ and $a\one$. As for $a\successor $, we can refer to it as the {\bf unary successor}\label{ius} of $a$, and refer to $a$ as the {\bf unary predecessor}\label{iup} of $a\successor $. Every number has a binary predecessor, and every number except $0$ has a unary predecessor. 
Note that  the binary predecessor of a number is the result of deleting the last digit (if present) in its binary representation. Remember that the string $0$ for us is just another name of the empty string $\epsilon$ representing the number zero. So,  $0$ and $1$ are no exceptions to the above rule: deleting the last (and only) digit of $1$ results in $\epsilon$; and ``deleting the last digit'' in $0$, i.e. in $\epsilon$, again results in $\epsilon$, as there is no digit to delete.   

The language of $\pa$ is known to be  very expressive, despite its nonlogical vocabulary officially being limited to only $0,\successor ,\plus ,\mult $. Specifically, it allows us to express, in a certain  standard way, all  recursive functions and relations, and beyond. Relying on the common knowledge of the power of the language of \pa,  we will be using standard expressions such as $x\mleq y$, $y\mgreater x$, etc. in formulas as abbreviations of the corresponding proper expressions of the language. In our metalanguage, $|x|$\label{iii} will refer to the length of the binary numeral for the number represented by $x$. In other words, $|x|\equals\lceil log_2(x\plus 1)\rceil$ \ ($\lceil y\rceil$ means the smallest integer $z$ with $y\mleq z$). Expressions like $|x|$ we refer to as {\bf pseudoterms} --- officially they are not terms of the object language, but in many contexts still can be treated as such. Another example of a pseudoterm is $2^x$ with its standard meaning.  So, when we write, say, ``$|x|\mleq y$'', it is officially to be understood as an abbreviation of a standard formula of {\bf PA} saying that $|x|$ does not exceed $y$.

For a variable $x$, by a {\bf polynomial sizebound} for $x$ we shall mean a standard   formula of the language of $\pa$ saying that $|x|\mleq\tau(|y_1|,\ldots,|y_n|)$, where $y_1,\ldots,y_n$ are any variables different from $x$, and $\tau(|y_1|,\ldots,|y_n|)$ is any $(0,\successor,\plus,\mult)$-combination of $|y_1|,\ldots,|y_n|$.  For instance, $|x|\mleq |y|\plus |z|$ is a polynomial sizebound for $x$, which is a formula of $\pa$ saying that    the size of $x$ does not exceed the sum of the sizes of $y$ and $z$.  Now, we say that a $\arfour$-formula $F$ is {\bf polynomially bounded} iff:
\begin{itemize}
\item Whenever $\ada x G(x)$ is a subformula of $F$, $G(x)$ has the form $S(x) \mli H(x)$, where $S(x)$ is a polynomial sizebound for $x$. 
\item Whenever $\ade x G(x)$ is a subformula of $F$, $G(x)$ has the form $S(x) \mlc H(x)$, where $S(x)$ is a polynomial sizebound for $x$.  
\end{itemize}

Remember that, where $F$ is a formula, $\cla  F$ means the $\cla$-closure of $F$, i.e., $\cla x_1\ldots\cla x_n F$, where $x_1,\ldots,x_n$ are the free variables of $F$.  
Similarly for $\cle F$, $\ada F$, $\ade F$.
  
The {\bf axioms} of $\arfour$ are:

\begin{description}
\item[Axiom 1:] $\cla x(0\notequals x\successor )$
\item[Axiom 2:] $\cla x\cla y(x\successor \equals y\successor \mli x\equals y)$
\item[Axiom 3:] $\cla x(x\plus 0\equals x)$
\item[Axiom 4:] $\cla x\cla y\bigl( x\plus y\successor \equals (x\plus y)\successor \bigr)$
\item[Axiom 5:] $\cla x(x\mult 0\equals 0)$
\item[Axiom 6:] $\cla x\cla y\bigl(x\mult  y\successor \equals (x\mult  y)\plus x\bigr)$
\item[Axiom 7:] $\cla \Bigl(  F(0)\mlc \cla x\bigl(F(x)\mli F(x\successor )\bigr)\mli \cla xF(x)\Bigr) $ for each elementary formula $F(x)$
\item[Axiom 8:] $\ada x\ade y(y\equals x\successor)$
\item[Axiom 9:] $\ada x\ade y(y\equals x \zero)$
\end{description}

All of the above are thus {\em nonlogical axioms} ($\arfour$ has no logical axioms).   Note that the overall number of axioms is infinite rather than nine, because Axiom 7 is not a particular sentence but a scheme of sentences. Axioms 1-7 are nothing but  {\bf Peano axioms} --- the nonlogical axioms of $\pa$. $\arfour$ thus only has two {\bf extra-Peano axioms}: Axiom 8 and Axiom 9. In view of the forthcoming soundness theorem for $\arfour$, Axiom 8 says that the unary successor function is polynomial time computable, and Axiom 9 says the same about the binary $0$-successor function.

As for the {\bf rules of inference}, $\arfour$ has a single logical rule, which is our old friend Logical Consequence, and a single nonlogical rule, which we call {\bf $\arfour$-Induction}. The latter is 
\[\frac{\ada \bigl(F(0)\bigr)\hspace{30pt} \ada\bigl( F(x)\mli F(x\zero)\bigr)\hspace{30pt} \ada\bigl( F(x)\mli F(x\one)\bigr)}{\ada \bigl(F(x)\bigr)},\]
where $F(x)$ is any {\em polynomially bounded} formula. 

Here we shall say that $\ada \bigl(F(0)\bigr)$ is  the {\bf basis} of induction,  $\ada \bigl(F(x)\mli F(x\zero)\bigr)$ is the {\bf left inductive step}, and  $\ada \bigl(F(x)\mli F(x\one)\bigr)$ is the {\bf right inductive step}. The variable $x$ has a special status here, and we say that the conclusion follows from the premises by {\bf $\arfour$-Induction on} $x$. A reader familiar with Buss's \cite{Buss} bounded arithmetic will notice a resemblance between the PIND axiom scheme of the latter and our $\arfour$-Induction rule, even though the two beasts operate in very different environments, of course.

A sentence $F$ is considered {\bf provable} in $\arfour$ iff there is a sequence of sentences, called a {\bf $\arfour$-proof} of $F$, where each sentence is either an  axiom, or  follows from some previous sentences by one of the two  rules of inference, and where the last sentence is $F$. An {\bf extended $\arfour$-proof} is defined in the same way, only, with the additional requirement that each application of LC should come together with an attached $\cltw$-proof of the corresponding sequent. With some fixed, effective, sound and complete axiomatization $L$ of classical first order logic in mind, a {\bf superextended $\arfour$-proof} is an extended $\arfour$-proof with the additional requirement that every application of Wait in the justification of a $\cltw$-derivation in it comes with an $L$-proof of the elementarization of the conclusion. Note that the property of being a superextended proof is (efficiently) decidable, while the properties of being an extended proof or just a proof are only recursively enumerable.  We write $\arfour\vdash F$ to say that $F$ is provable (has a proof) in $\arfour$, and $\arfour\not\vdash F$ to say the opposite. 

Generally, as in the above definition of provability and proofs, in $\arfour$ we will only be interested in proving {\em sentences}.\footnote{In case we do not insist that every formula in an $\arfour$-proof be a sentence, one could show that, if $F$ is not a sentence, it is provable (in this relaxed sense) if and only if its $\adc$-closure is so, anyway.} So, for technical convenience, we agree that, from now on, whenever we write $\arfour\vdash F$ (or say ``$F$ is provable'') for a non-sentence $F$, it  simply means that $\arfour\vdash \ada F$. Similarly, when we say that a given strategy solves $F$, it is to be understood as that the strategy solves $\ada F$. Similarly, when we say that $F$ is a logical consequence of $E_1,\ldots,E_n$, what we typically mean is that $\ada F$ is a logical consequence of $\ada E_1,\ldots,\ada E_n$. To summarize, in the context of $\arfour$, any formula with free variables should be understood as its $\ada$-closure. An exception is when $F(x_1,\ldots,x_n)$ is an elementary formula and we say that $F(c_1,\ldots,c_n)$ is {\bf true} (whatever constants $c_1,\ldots,c_n$). This is to be understood as that the $\cla$-closure $\cla F(c_1,\ldots,c_n)$ of $F(c_1,\ldots,c_n)$ is true (in the standard model), for ``truth'' is only meaningful for elementary formulas (which $\ada F$ generally would not be). An important fact on which we will often rely yet only implicitly so, is that the sentence $\cla F\mli \ada F$ or the closed sequent $\cla F\intimpl \ada F$ is (always) $\cltw$-provable. In view of the soundness of $\cltw$, this means that whenever ($F$ is elementary and) $\cla F$ is true, $\ada F$ is automatically won by a strategy that does nothing.

\begin{example}\label{onesuc}
%\marginpar{onesuc}
The following sequence is a proof of $\ada x\ade y(y\equals x\one)$, i.e. of $\ada x\ade y\bigl(y\equals (x\zero)\successor\bigr)$  --- the sentence saying that the binary $1$-successor function is computable:\vspace{7pt}

\noindent I.   $\ada x\ade y (y\equals x\successor)$  \  Axiom 8  \vspace{3pt}

\noindent II.   $\ada x\ade y (y\equals  x\zero)$  \  Axiom 9  \vspace{3pt}

\noindent III.   $\ada x\ade y \bigl(y\equals (x\zero)\successor\bigr)$  \  LC: I,II  \vspace{7pt}

An extended version of the above proof would have to include a justification for step III where LC was used, such as the following one:\vspace{7pt}

\noindent 1.   $t\equals r\successor,\  r\equals  s\zero\ \intimpl\ t\equals (s\zero)\successor$ \  Wait:     \vspace{3pt}

\noindent 2.   $t\equals r\successor,\  r\equals  s\zero\ \intimpl\ \ade y \bigl(y\equals (s\zero)\successor\bigr)$ \  $\ade$-Choose: 1    \vspace{3pt}

\noindent 3.   $\ade y (y\equals r\successor),\  r\equals  s\zero\ \intimpl\ \ade y \bigl(y\equals (s\zero)\successor\bigr)$ \  Wait: 2    \vspace{3pt}

\noindent 4.   $\ada x\ade y (y\equals x\successor),\  r\equals  s\zero\ \intimpl\ \ade y \bigl(y\equals (s\zero)\successor\bigr)$ \  $\ada$-Choose: 3   \vspace{3pt}

\noindent 5.   $\ada x\ade y (y\equals x\successor),\ \ade y (y\equals  s\zero)\ \intimpl\ \ade y \bigl(y\equals (s\zero)\successor\bigr)$ \ Wait: 4  \vspace{3pt}

\noindent 6.   $\ada x\ade y (y\equals x\successor),\ \ada x\ade y (y\equals  x\zero)\ \intimpl\ \ade y \bigl(y\equals (s\zero)\successor\bigr)$  \ $\ada$-Choose: 5 \vspace{3pt}

\noindent 7.   $\ada x\ade y (y\equals x\successor),\ \ada x\ade y (y\equals  x\zero)\ \intimpl\ \ada x\ade y \bigl(y\equals (x\zero)\successor\bigr)$  \ Wait: \vspace{7pt}

As we just saw, (additionally) justifying an application of LC takes more space than the (non-extended) proof itself. And this would be a typical case for $\arfour$-proofs. Luckily, however, there is no real need to formally justify LC. Firstly, this is so because $\cltw$ is an analytic system, and proof-search in it is a routine (even if sometimes long) syntactic exercise. Secondly, in view of Thesis \ref{thesis}, there is no need to generate formal $\cltw$-proofs anyway: instead, we can  use intuition on games and strategies. In the present case, the whole (III+7)-step extended formal proof can be replaced by a short intuitive argument in the following style:
\begin{quote} {\em Consider any $s$ chosen by Environment for $x$. We need to compute $(s\zero)\successor$. Using Axiom 9, we   find the value $r$ of $s\zero$. Then, using Axiom 8, we find the value $t$   of $r\successor$. Such a $t$ is what we are looking for. }
\end{quote}

Furthermore, as we agreed to understand a proof of a non-sentence $F$ as a proof of its $\ada$-closure $\ada F$, the above argument can be shortened by considering simply $\ade y(y\equals x\one)$ instead of $\ada x\ade y(y\equals x\one)$. This would allow us to skip the phrase ``Consider any $s$ chosen by Environment for $x$'' and the necessity to introduce the new name/variable $s$. Instead, we can simply say ``Consider any $x$'', implicitly letting $x$ itself serve as a variable for the value chosen by Environment for $x$ when playing the $\ada$-closure of the game. Sometimes we can go even lazier and skip the phrase ``Consider any $x$'' altogether. 
\end{example}

In view of the following fact, an alternative way to present $\arfour$ would be to delete Axioms 1-7  and, instead, declare all (closed) theorems of {\bf PA} to be axioms of $\arfour$ along with Axioms 8 and 9:

\begin{fact}\label{nov7}
%\marginpar{nov7} 
Every (elementary $\arfour$-) sentence provable in $\pa$ 
is also provable in $\arfour$.
\end{fact}

\begin{proof} Suppose (the classical-logic-based) $\pa$ proves  $F$. By the deduction theorem for classical logic this means  that, for some nonlogical axioms $H_1,\ldots,H_n$ of $\pa$, the formula $H_1\mlc\ldots\mlc H_n\mli F$ is provable in classical first order logic. Hence $H_1\mlc\ldots\mlc H_n\intimpl F$ is 
provable in $\cltw$ by Wait from the empty set of premises. Hence $F$ is a logical consequence of the Peano axioms $H_1,\ldots,H_n$ of $\arfour$ and, as such, is provable by LC.
\end{proof}

The above fact, on which we will be implicitly relying in the sequel, allows us to construct ``lazy'' $\arfour$-proofs where some steps can be justified by simply indicating their provability in {\bf PA}. That is, we will  treat theorems of {\bf PA} as if they were axioms of $\arfour$. As {\bf PA} is well known and well studied, we safely assume that the reader has a good feel for what it can prove, so we do not usually further justify {\bf PA}-provability claims that we make. A reader less familiar with {\bf PA}, can take it as a rule of thumb that, despite G\"{o}del's incompleteness theorems, 
{\bf PA} proves every true number-theoretic fact that a contemporary high school  student can establish, or that mankind was or could be aware of before 1931.
 
\begin{example}\label{113}
%\marginpar{113}
The following sequence is a lazy proof of $\ada x(x\equals 0\add x\notequals 0)$   --- the formula saying that the ``zeroness'' predicate is decidable:\vspace{7pt}

\noindent I.   $0\equals 0\add 0\notequals 0$  \  LC:  \vspace{3pt}

\noindent II.   $\cla x(x\equals 0\mli x\zero\equals 0)$  \  $\pa$  \vspace{3pt}

\noindent III.   $\cla x(x\notequals 0\mli x\zero\notequals 0)$  \  $\pa$  \vspace{3pt}

\noindent IV.   $\ada x(x\equals 0\add x\notequals 0\mli x\zero\equals 0\add x\zero\notequals 0)$  \  LC: II,III  \vspace{3pt}

\noindent V.   $\cla x(x\one\notequals 0)$  \  $\pa$  \vspace{3pt}

\noindent VI.   $\ada x(x\equals 0\add x\notequals 0\mli x\one\equals 0\add x\one\notequals 0)$  \  LC: V  \vspace{3pt}

\noindent VII.   $\ada x(x\equals 0\add x\notequals 0)$  \  $\arfour$-Induction: I,IV,VI  \vspace{7pt}

An extended version of the above proof will include the following three additional justifications ($\cltw$-proofs):\vspace{5pt}

A justification for Step I:\vspace{5pt}

\noindent 1.   $0\equals 0$  \  Wait:  \vspace{3pt}

\noindent 2.   $0\equals 0\add 0\notequals 0$  \  $\add$-Choose: 1  \vspace{7pt}

A justification for Step IV:\vspace{5pt}

\noindent 1.   $\cla x(x\equals 0\mli x\zero\equals 0),\ \cla x(x\notequals 0\mli x\zero\notequals 0)\ \intimpl \ s\equals 0 \mli s\zero\equals 0$  \  Wait:  \vspace{3pt}

\noindent 2.   $\cla x(x\equals 0\mli x\zero\equals 0),\ \cla x(x\notequals 0\mli x\zero\notequals 0)\ \intimpl \ s\equals 0 \mli s\zero\equals 0\add s\zero\notequals 0$  \  $\add$-Choose: 1  \vspace{3pt}

\noindent 3.   $\cla x(x\equals 0\mli x\zero\equals 0),\ \cla x(x\notequals 0\mli x\zero\notequals 0)\ \intimpl \   s\notequals 0\mli   s\zero\notequals 0$  \  Wait:  \vspace{3pt}

\noindent 4.   $\cla x(x\equals 0\mli x\zero\equals 0),\ \cla x(x\notequals 0\mli x\zero\notequals 0)\ \intimpl \   s\notequals 0\mli s\zero\equals 0\add s\zero\notequals 0$  \  $\add$-Choose: 3  \vspace{3pt}

\noindent 5.   $\cla x(x\equals 0\mli x\zero\equals 0),\ \cla x(x\notequals 0\mli x\zero\notequals 0)\ \intimpl \ s\equals 0\add s\notequals 0\mli s\zero\equals 0\add s\zero\notequals 0$  \  Wait:  2,4\vspace{3pt}

\noindent 6.   $\cla x(x\equals 0\mli x\zero\equals 0),\ \cla x(x\notequals 0\mli x\zero\notequals 0)\ \intimpl \ \ada x(x\equals 0\add s\notequals 0\mli x\zero\equals 0\add x\zero\notequals 0)$  \  Wait:  5\vspace{7pt}

A justification for Step VI:\vspace{5pt}

\noindent 1.   $\cla x(x\one\notequals 0)\ \intimpl\  s\equals 0 \mli s\one\notequals 0$  \  Wait:  \vspace{3pt}

\noindent 2.   $\cla x(x\one\notequals 0)\ \intimpl\  s\notequals 0\mli s\one\notequals 0$  \  Wait:   \vspace{3pt}

\noindent 3.   $\cla x(x\one\notequals 0)\ \intimpl\  s\equals 0\add s\notequals 0\mli s\one\notequals 0$  \  Wait: 1,2  \vspace{3pt}

\noindent 4.   $\cla x(x\one\notequals 0)\ \intimpl\ s\equals 0\add s\notequals 0\mli s\one\equals 0\add s\one\notequals 0$ \  $\add$-Choose:  3\vspace{3pt}

\noindent 5.   $\cla x(x\one\notequals 0)\ \intimpl\ \ada x(x\equals 0\add x\notequals 0\mli x\one\equals 0\add x\one\notequals 0)$ \  Wait:  4\vspace{7pt}

But, again, in the sequel we will not generate formal (even if lazy) proofs in the above style but, instead, limit ourselves to informal arguments within $\arfour$. Below we illustrate such an argument for the present case. Further, as explained in the previous example, we prefer to deal with $x\equals 0\add x\notequals 0$ rather than $\ada x(x\equals 0\add x\notequals 0)$, and do not explicitly say ``consider any value $x$ chosen for the variable $x$ by Environment''.

{\em  We prove $x\equals 0\add x\notequals 0$ by $\cltw$-Induction on $x$. 

The \underline{Basis} $0\equals 0\add 0\notequals 0$ of induction is straightforward:  it is solved by choosing the left $\add$-disjunct.

Solving the \underline{left inductive step} $ x\equals 0\add x\notequals 0\mli x\zero\equals 0\add x\zero\notequals 0 $ means   solving the consequent using a single copy of the antecedent as a resource. From 
$\pa$, we know that $x\zero\equals 0$ iff $x\equals 0$. And whether $x\equals 0$ or not we can find out using the resource $x\equals 0\add x\notequals 0$. So, we can tell whether $x\zero\equals 0$ or ($\add$) $x\zero\notequals 0$, as desired.  

The \underline{right inductive step} $ x\equals 0\add x\notequals 0\mli x\one\equals 0\add x\one\notequals 0$ is solved by  simply choosing the right $\add$-disjunct in the consequent which, by $\pa$, is true. }  
  
In the future, our reliance on $\pa$ may not always be explicit as it was in the above informal argument. Further, as we advance, our reliance on some basic nonelementary facts such as the just-proven $\ada x(x\equals 0\add x\notequals 0)$ may increasingly become only implicit.   

Why did not we use the following, much simpler intuitive argument instead of the above one? 

\begin{quote}
{\em To solve  $x\equals 0\add x\notequals 0$, see if $x$ (the constant chosen by Environment for $x$) is $0$ or not. If yes,  choose the left $\add$-disjunct; otherwise choose the right $\add$-disjunct.}
\end{quote}

The above argument is not valid, in the sense that it is not ``purely logical'', and hence Thesis \ref{thesis} does not guarantee that it will be translatable into a formal proof. Namely, in being confident  that the above strategy wins the game, we rely on the extra-logical knowledge of the fact that different constants are names of different objects. This is indeed so for the standard model of arithmetic. But in some other models it is quite possible that, say, both constants $0$ and $10$ are names for the same object of the universe. Then, the above strategy prescribes to choose the false $10\notequals 0$ in the corresponding scenario.   

Thus, $\ada x(x\equals 0\add x\notequals 0)$, while provable in $\arfour$, is not logically valid, that is, is not provable in $\cltw$. The same applies to the more general principle $\ada x\ada y(x\equals y\add x\notequals y)$, whose $\arfour$-provability is shown later in Section \ref{s17}.
\end{example}

When a formula $F(x_1,\ldots,x_n)$ is a standard (in whatever informal sense) representation of an $n$-ary predicate $p$ and $\arfour\vdash \ada \bigl(F(x_1,\ldots,x_n)\add \gneg F(x_1,\ldots,x_n)\bigr)$, we say that $p$ is 
{\bf $\arfour$-provably decidable}. Similarly, when $F(y,x_1,\ldots,x_n)$ is a standard representation of the graph of an $n$-ary function $f$ and  $\arfour\vdash \ada \ade y F(y,x_1,\ldots,x_n)$, 
 we say that $f$ is 
{\bf $\arfour$-provably computable}. Since $\arfour$ is the only system of clarithmetic dealt with in this paper, the prefix ``$\arfour$'' before ``provably'' can be safely omitted.  So, for instance, Example \ref{113}  established the provable decidability of the ``zeroness'' predicate, and Example \ref{onesuc} established the provable computability of the binary $1$-successor function. The same terminology extends to partially defined functions as well. For instance, the later-proven Fact \ref{m4h} shows that $\arfour\vdash \ada x\bigl(x\notequals 0\mli \ade y(x\equals y\successor)\bigr)$. We understand this as the provable computability of the (partial) unary predecessor function.      

\begin{exercise} Let $\mbox{\em Even}(x)$ be an abbreviation of $\cle z(x\equals z\plus z)$, and $\mbox{\em Odd}(x)$ be an abbreviation of $\gneg \mbox{\em Even}(x)$. Find a lazy $\arfour$-proof of 
\(\cla x \Bigl(\mbox{\em Even}(x)\add \mbox{\em Odd}(x)\mli \ada y\bigl(\mbox{\em Even}(x\plus y)\add \mbox{\em Odd}(x\plus y)\bigr)\Bigr).\)
{\em Hint}: Relying on the $\pa$-provable fact that $x\zero$ is always even and $x\one$ is always odd, by $\arfour$-Induction prove $\ada x\bigl(\mbox{\em Even}(x)\add \mbox{\em Odd}(x)\bigr)$. The target sentence is a logical consequence of the latter and the $\pa$-provable fact that the sum of two numbers is even iff both numbers are even or both are odd.  
\end{exercise}

By an {\bf arithmetical problem} in this paper we mean a game $A$ such that, for some sentence $F$ of the language of $\arfour$, \ $A= F^\dagger$ (remember that $^\dagger$ is the standard interpretation). Such a sentence $F$ is said a {\bf representation}\label{irepresentation} of $A$.
We say that an arithmetical problem $A$ is {\bf provable} in $\arfour$ iff it has a $\arfour$-provable representation.  
In these terms, a central result of the present paper sounds as follows:

\begin{theorem}\label{tt1}
%\marginpar{tt1}
An arithmetical problem has a polynomial time solution iff it is provable in $\arfour$. 

Furthermore, there is an efficient procedure that takes an arbitrary extended $\arfour$-proof of an arbitrary sentence $X$ and constructs a   
 solution of $X$ (of $X^\dagger$, that is) together with an explicit polynomial bound for its time complexity. 
\end{theorem}

\begin{proof} The soundness (``if'') part of this theorem will be proven in Section \ref{sectsound}, and the completeness (``only if'') part in Section 
\ref{sectcompl}.\vspace{-7pt}
\end{proof}

\section{Some more taste of \arfour}\label{s17}
%\marginpar{s17} 

In this section we establish several $\arfour$-provability facts. In view of the soundness of $\arfour$, each such fact tells us about the efficient solvability  of the associated number-theoretic  computational problem. 

As mentioned, the present work has been written not merely as a research paper but also as an advanced tutorial on CoL-based applied theories. The series of proofs given in this section can be treated as exercises aimed at developing the reader's feel for CoL-based systems and $\arfour$ in particular. But these proofs also provide certain necessary results relied upon later, in our proof of the extensional completeness of $\arfour$. 

As noted earlier, we shall exclusively rely on informal reasoning in $\arfour$, remembering that behind every such piece of reasoning is a formal $\arfour$-proof.

\begin{fact}\label{comlog}
%\marginpar{comlog}
 $\arfour\vdash \ade y(y\equals |x|)$ \ (i.e., $\arfour\vdash \ada x\ade y(y\equals |x|)$).
\end{fact}

\begin{proof} Argue in $\arfour$. By $\arfour$-Induction on $x$, we first want to prove $ \ade y(|y|\mleq |x|\mlc y\equals |x|)$. 

The basis   $ \ade y(|y|\mleq |0|\mlc y\equals |0|)$ is solved by selecting $0$ for $y$. Such a move brings the game down to $|0|\mleq 0\mlc 0\equals |0|$. From $\pa$, we know that the latter is true. So, we win. 

The left inductive step is $\ade y(|y|\mleq |x|\mlc y\equals |x|)\mli \ade y(|y|\mleq |x\zero|\mlc y\equals |x\zero|)$.  Using Example \ref{113}, we figure our whether $x\equals 0$ or not. If $x\equals 0$, then (we know from $\pa$ that) $|x\zero|=0$,  so we win the game by selecting $0$ for $y$ in the consequent. Now suppose $x\notequals 0$. We wait till Environment selects a constant $a$ for $y$ in the antecedent\footnote{In informal arguments like this, we usually do not try to be consistent in using different metavariables for constants and variables. Notice that the status of $x$ in the present informal argument is also ``constant'' (chosen by Environment for the variable $x$) just like the status of $a$ but we, out of reluctance to introduce new names, continue using  the expression ``$x$'' for it, even though in earlier sections we tried to reserve the metanames $x,y,z,\ldots$ (as opposed to $a,b,c,\ldots$) for variables rather than constants.} (if this does not happen, we win).  
Then, using Axiom 8, we compute the value (constant) $b$ with $b\equals a\successor$, and choose $b$ for $y$ in the consequent. We win because, from $\pa$, we know that  when $x$ is not $0$, the resulting position  
$|a|\mleq |x|\mlc a\equals |x|\mli  |b|\mleq |x\zero|\mlc b\equals |x\zero|$, i.e. $|a|\mleq |x|\mlc a\equals |x|\mli  |a\successor|\mleq |x\zero|\mlc a\successor\equals |x\zero|$ (whatever the value of $x$ is) is true. 
The right inductive step is similar but simpler, as we do not need to separately consider the case of $x\equals 0$.

So, we know how to win (the $\ada$-closure of) $\ade y(|y|\mleq |x|\mlc y\equals |x|)$. But then, of course, ignoring the  $|y|\mleq |x|$ conjunct, we also know how to win (the $\ada$-closure of the weaker)  $\ade y(y\equals |x|)$. Putting it in precise terms, $\ada x\ade y(y\equals |x|)$ is an (easy) logical consequence of $\ada x\ade y(|y|\mleq |x|\mlc y\equals |x|)$ --- that is, the sequent  $\ada x\ade y(|y|\mleq |x|\mlc y\equals |x|)\intimpl \ada x\ade y(y\equals |x|)$ is (easily) provable in $\cltw$. 
\end{proof}

Compare the $\ada$-closures of the following formulas: 
\begin{eqnarray}
\label{f1}\cle y(x\equals y\zero\mld x\equals y\one)  & & \\
\label{f2}\ade y(x\equals y\zero\mld x\equals y\one)  & & \\
\label{f3} \cle y(x\equals y\zero\add x\equals y\one) & & \\
\label{f4} \ade y(x\equals y\zero\add x\equals y\one) & &  
\end{eqnarray}
All four sentences ``say the same'' about the arbitrary ($\ada$) number represented by $x$, but in different ways. (\ref{f1}) is the weakest, least informative, of the four.  It says that  $x$ has a binary predecessor $y$, and that $x$ is even (i.e., is the binary $0$-successor of its binary predecessor) or odd (i.e., is the binary $1$-successor of its binary predecessor). This is an almost trivial piece of information. (\ref{f2}) and (\ref{f3}) carry stronger information. According to (\ref{f2}), $x$ not just merely {\em has} a binary predecessor $y$, but  such a predecessor can be actually and efficiently {\em found}. (\ref{f3}) strengthens (\ref{f1}) in another way. It says that $x$ can be efficiently determined to be even or odd. As for (\ref{f4}), it is the strongest. It carries two pieces of good news at once: we can efficiently find the binary predecessor $y$ of $x$ and, simultaneously, tell whether $x$ is even or odd. According to the following fact, (\ref{f4}) is provable. As we may guess, so are the weaker (\ref{f3}), (\ref{f2}), (\ref{f1}).

\begin{fact}\label{m4gg}
%\marginpar{m4gg}
 $\arfour\vdash \ade y(x\equals y\zero\add x\equals y\one)$.
\end{fact}

\begin{proof} Argue in $\arfour$. By $\arfour$-Induction on $x$, we first want to prove $ \ade y\bigl(|y|\mleq |x|\mlc (x\equals y\zero\add x\equals y\one)\bigr)$.

The basis 
 $\ade y\bigl(|y|\mleq |0|\mlc (0\equals y\zero\add 0\equals y\one)\bigr)$ is obviously solved by selecting $0$ for $y$ and then choosing the left $\add$-disjunct, which results in the true (according to $\pa$) sentence $|0|\mleq |0|\mlc 0\equals 0\zero $. 
The left inductive step  $  \ade y(x\equals y\zero\add x\equals y\one)\mli \ade y(x\zero\equals y\zero\add x\zero\equals y\one) $ is solved by selecting for $y$ the same constant as the one selected by Environment for $x$, and then choosing the left $\add$-disjunct. Similarly, the right inductive step  
$ \ade y(x\equals y\zero\add x\equals y\one)\mli \ade y(x\one\equals y\zero\add x\one\equals y\one) $ is solved by selecting for $y$ the same constant as the one selected by Environment for $x$, and then choosing the right $\add$-disjunct.

Now,  $ \ade y (x\equals y\zero\add x\equals y\one) $ is a straightforward logical consequence of $ \ade y\bigl(|y|\mleq |x|\mlc (x\equals y\zero\add x\equals y\one)\bigr)$.
\end{proof}

The preceding fact established the provable computability of binary predecessor. The following fact does the same for unary predecessor:

\begin{fact}\label{m4h}
%\marginpar{m4h}
 $\arfour\vdash   x\notequals 0\mli \ade y(x\equals y\successor) $.
\end{fact}

\begin{proof} Argue in $\arfour$. By $\arfour$-Induction on $x$, we  want to prove 
$x\notequals 0\mli \ade y(|y|\mleq |x|\mlc x\equals y\successor)$, 
from which the target $ x\notequals 0\mli \ade y(x\equals y\successor)$ follows immediately by LC.

The basis   $ 0\notequals 0\mli \ade y(|y|\mleq |0|\mlc 0\equals y\successor)$ is solved trivially by a strategy that does nothing. A win is guaranteed because the antecedent is false.  

The left inductive step is $ \bigl(x\notequals 0\mli \ade y(|y|\mleq |x|\mlc x\equals y\successor)\bigr)\mli \bigl(x\zero\notequals 0\mli \ade y(|y|\mleq |x\zero|\mlc x\zero\equals y\successor)\bigr) $.   If $x\zero\notequals 0$   (and if not, we win the game), then --- according to $\pa$ --- $x\notequals 0$. So, Environment will have to choose a constant $a$ for $y$ in the antecedent, or else it loses. We may assume that $a$ is (indeed) the unary predecessor of $x$, or else, again, having chosen a wrong $a$, Environment loses. We know from $\pa$ that then the unary predecessor $b$ of $x\zero$ equals $a\one$, and that $|b|\mleq |x\zero |$. This $b$  can be computed using (the resource provided by)  Example \ref{onesuc}. We choose $b$ for $y$ in the consequent and win. 

The right inductive step   $ \bigl(x\notequals 0\mli \ade y(|y|\mleq |x|\mlc x\equals y\successor)\bigr)\mli \bigl(x\one\notequals 0\mli \ade y(|y|\mleq |x\one|\mlc x\one\equals y\successor)\bigr) $ is even easier to handle.   Prom $\pa$, the unary predecessor of $x\one$ is $x\zero$, and $|x\zero |\mleq |x\one |$. Using Axiom 9, we compute the value $b$ of $x\zero$ and choose $b$ for $y$ in the consequent.\vspace{-7pt}
\end{proof}

\begin{fact}\label{comad}
%\marginpar{comad}
 $\arfour\vdash  \ade z(z\equals x\plus y)$.
\end{fact}

\begin{proof} Argue in $\arfour$. By $\arfour$-Induction on $x$, we  want to prove 
\[\ada y \bigl(|y|\mleq |t|\mli \ade z(|z|\mleq |x|\plus |y|\mlc z\equals x\plus y)\bigr),\]
from which, together with the $\pa$-provable fact $\cla y(|y|\mleq |y|)$, the target $\ade z(z\equals x\plus y)$  follows by LC.

The  basis is $\ada y \bigl(|y|\mleq |t|\mli \ade z(|z|\mleq |0|\plus |y|\mlc z\equals 0\plus y)\bigr)$, which (as always) we prefer to simply write as $|y|\mleq |t|\mli \ade z(|z|\mleq |0|\plus |y|\mlc z\equals 0\plus y) $. It is won by selecting (the value of) $y$ for $z$, because we know, from $\pa$, that the resulting $|y|\mleq |t|\mli  |y|\mleq |0|\plus |y|\mlc y\equals 0\plus y)$ is true.

In inductive steps, we will rely on the fact that $\pa$ proves (the $\cla$-closures of) the following formulas:
\begin{eqnarray}
s\zero\plus r\zero\equals(s\plus r)\zero, & \mbox{i.e.,} & 2s\plus 2r\equals 2(s\plus r);\label{obs1}\\
s\zero\plus r\one\equals(s\plus r)\one, & \mbox{i.e.,} & 2s\plus (2r\plus 1)\equals 2(s\plus r)\plus 1;\label{obs2}\\
s\one\plus r\zero\equals(s\plus r)\one , & \mbox{i.e.,} & (2s\plus 1)\plus 2r\equals 2(s\plus r)\plus 1 ;\label{obs3}\\
s\one\plus r\one\equals\bigl((s\plus r)\one\bigr)\successor, & \mbox{i.e.,} & (2s\plus 1)\plus (2r\plus 1)\equals \bigl(2(s\plus r)\plus 1\bigr)\plus 1.\label{obs4}
\end{eqnarray}

The left inductive step is 
%\marginpar{m4e}
\begin{equation}\label{m4e}
\ada y \bigl(|y|\mleq |t|\mli \ade z(|z|\mleq |x|\plus |y|\mlc z\equals x\plus y)\bigr)\mli \ada y \bigl(|y|\mleq |t|\mli \ade z(|z|\mleq |x\zero|\plus |y|\mlc z\equals x\zero\plus y)\bigr).
\end{equation}
To solve it, we wait till Environment chooses a constant $a$ for $y$ in the consequent, after which (\ref{m4e}) will be brought down to 
%\marginpar{m4f}
\begin{equation}\label{m4f}
\ada y \bigl(|y|\mleq |t|\mli \ade z(|z|\mleq |x|\plus |y|\mlc z\equals x\plus y)\bigr)\mli  \bigl(|a|\mleq |t|\mli \ade z(|z|\mleq |x\zero|\plus |a|\mlc z\equals x\zero\plus a)\bigr).
\end{equation}
Using Fact \ref{m4gg}, we find the binary predecessor $b$ of $a$, for which we will also know whether  $a\equals b\zero$ or (``or'' in the strong sense of $\add$) $a\equals b\one$.  We specify $y$ as $b$ in the antecedent of (\ref{m4f}), and wait till Environment selects a value $c$ for $z$ there.  

If 
$a\equals b\zero$, the game by now will be brought down to 
\[
 \bigl(|b|\mleq |t|\mli(|c|\mleq |x|\plus |b|\mlc c\equals x\plus b)\bigr)\mli  \bigl(|b\zero|\mleq |t|\mli \ade z(|z|\mleq |x\zero|\plus |b\zero|\mlc z\equals x\zero\plus b\zero)\bigr).
\]
Using Axiom 9, we   compute the value $d$ of $c\zero$, and specify $z$ as $d$ in the consequent. The resulting position 
\[\bigl(|b|\mleq |t|\mli(|c|\mleq |x|\plus |b|\mlc c\equals x\plus b)\bigr)\mli  \bigl(|b\zero|\mleq |t|\mli  |c\zero|\mleq |x\zero|\plus |b\zero|\mlc c\zero\equals x\zero\plus b\zero\bigr),\]
in view of (\ref{obs1}) (and certain additional, straightforward $\pa$-provable facts), is true, so we win.   

Quite similarly, if $a\equals b\one$, the game by now will be brought down to 
\[\bigl(|b|\mleq |t|\mli(|c|\mleq |x|\plus |b|\mlc c\equals x\plus b)\bigr)\mli  \bigl(|b\one|\mleq |t|\mli \ade z(|z|\mleq |x\zero|\plus |b\one|\mlc z\equals x\zero\plus b\one)\bigr).\]
Using Example \ref{onesuc}, we   compute the value $d$ of $c\one$, and  specify $z$ as $d$ in the consequent. The resulting position 
\[ \bigl(|b|\mleq |t|\mli(|c|\mleq |x|\plus |b|\mlc c\equals x\plus b)\bigr)\mli  \bigl(|b\one|\mleq |t|\mli |c\one|\mleq |x\zero|\plus |b\one|\mlc c\one\equals x\zero\plus b\one\bigr),\]
in view of (\ref{obs2}), is true, so we win.   

The right inductive step will be handled in a similar way, only relying on (\ref{obs3}) and (\ref{obs4}) instead of (\ref{obs1}) and (\ref{obs2}).\vspace{-7pt}  
\end{proof}

\begin{fact}\label{commul}
%\marginpar{commul}
 $\arfour\vdash \ade z(z \equals x\mult  y)$.
\end{fact}

\begin{proof} The general scheme of proof here is very similar to the one employed in the proof of Fact \ref{comad}, and details are left as an exercise to the reader. Here we shall only point out that, the four basic $\pa$-provable facts that play the same role here as facts (\ref{obs1})-(\ref{obs4}) in the proof of Fact \ref{comad} are the following: 
\begin{eqnarray*}
s\zero\mult  r\zero\equals (s\mult  r)\zero\zero, & \mbox{i.e.,} & 2s\mult 2r\equals 4(s\mult r); \\
s\zero\mult  r\one\equals (s\mult  r)\zero\zero\plus s\zero, & \mbox{i.e.,} & 2s\mult (2r\plus 1)\equals 4(s\mult r)\plus 2s; \\
s\one\mult  r\zero\equals (s\mult  r)\zero\zero\plus r\zero , & \mbox{i.e.,} & (2s\plus 1)\mult 2r\equals 4(s\mult r)\plus 2r; \\
s\one\mult  r\one\equals (s\mult  r)\zero\zero\plus (s\plus r)\one, & \mbox{i.e.,} & (2s\plus 1)\mult (2r\plus 1)\equals 4(s\mult r)\plus 2(s\plus r)\plus 1. 
\end{eqnarray*}
Also, where the previous proof relied on Axiom 9 and  Example \ref{onesuc}, the present proof, in addition, will rely on Fact \ref{comad}.  
\end{proof}

The following fact establishes the provable computability of all (polynomial) functions represented through  terms: 

\begin{fact}\label{com}
%\marginpar{com}
For any term $\tau$ (not containing $z$),  
 $\arfour\vdash  \ade z (z\equals \tau) $.
\end{fact}

\begin{proof} We prove this fact by (meta)induction on the complexity of $\tau$. The base cases are those of $\tau$ being the constant $0$ or a variable $x$. Both of the corresponding sentences $\ade z(z\equals 0)$ and $\ada x\ade z(z\equals x)$ are provable in $\cltw$ and hence also in $\arfour$.  Next, assume  $\tau$ is $\theta\successor$. By the induction hypothesis, $\arfour$ proves $ \ade z(z\equals \theta)$. $\arfour$ also proves $\ade y(y\equals x\successor)$ (Axiom 8). The desired $\ade z(z\equals \theta\successor)$ is a logical consequence of these two. The remaining cases of $\tau$ being $\theta_1\plus \theta_2$ or $\theta_1\mult\theta_2$ are handled in a similar way, relying on Facts \ref{comad} and \ref{commul}, respectively.   
 \end{proof}

The formula of the following fact, as a computational problem, is about finding the (nonnegative) difference $z$  between any two numbers $x$ and $y$ and then telling whether this difference is $x\minus y$ or $y\minus x$.  

\begin{fact}\label{minus}
%\marginpar{minus}
 $\arfour\vdash \ade z(x\equals y\plus z\add y\equals x\plus z)$.  
\end{fact}

\begin{proof}  Argue in $\arfour$. By $\arfour$-Induction on $x$, we want to show 
%\marginpar{m5a}
\begin{equation}\label{m5a}
\ada y\Bigl(|y|\mleq |t|\mli \ade z\bigl(|z|\mleq |x|\plus |y|\mlc( x\equals y\plus z\add y\equals x\plus z)\bigr)\Bigr). 
\end{equation}

The {\em basis}  $|y|\mleq |t|\mli \ade z\bigl(|z|\mleq |0|\plus |y|\mlc (0\equals y\plus z\add y\equals 0\plus z)\bigr)$ is obviously\footnote{Here and often elsewhere  implicitly relying on $\pa$.} solved by the strategy that chooses the value of $y$ for the variable $z$  and then selects the right $\add$-disjunct.  

To solve the {\em left inductive step}
%\marginpar{m5b}
\begin{equation}\label{m5b}
\begin{array}{l}
\ada y\Bigl(|y|\mleq |t|\mli \ade z\bigl(|z|\mleq |x|\plus |y|\mlc (x\equals y\plus z\add y\equals x\plus z)\bigr)\Bigr) \mli \\ 
\ada y\Bigl(|y|\mleq |t|\mli \ade z\bigl(|z|\mleq |x\zero|\plus |y|\mlc (x\zero\equals y\plus z\add y\equals x\zero\plus z)\bigr)\Bigr),  
\end{array}
\end{equation}
 we wait  till Environment specifies a constant $a$ for $y$ in the consequent. Then, using Fact \ref{m4gg}, we compute the binary predecessor $b$ of $a$, and also figure out whether $a=b\zero$ or ($\add$) $a\equals b\one$. 

{\em Case 1}: $a\equals b\zero$. We specify $y$ as $b$ in the antecedent of (\ref{m5b}), this way forcing Environment to choose a constant $c$ for $z$ there (unless $|b|\mleq |t|$ is false, in which case $|a|\mleq |t|$ is also false and we win), and also choose one of the disjuncts of $ x\equals b\plus c\add b\equals x\plus c$. Using Axiom 9, we calculate the value $d$ of $c\zero$, and specify $z$ as $d$ in the consequent of (\ref{m5b}). Further, if Environment has chosen $ x\equals b\plus c$ in the antecedent, we choose the left $\add$-disjunct in the consequent.  This means that, by now, 
(\ref{m5b}) is brought down to 
\[(|b|\mleq |t|\mli  |c|\mleq |x|\plus |b|\mlc x\equals b\plus c ) \mli  
 (|b\zero|\mleq t\mli  |c\zero|\mleq |x\zero|\plus |b\zero|\mlc x\zero\equals b\zero\plus c\zero) .\] 
From $\pa$, the above is true, so we win. Similarly, if Environment has chosen $b\equals x\plus c$  in the antecedent of (\ref{m5b}), then we choose the right $\add$-disjunct in the consequent, and again win.

{\em Case 2}: $a\equals b\one$. Again, we specify $y$ as $b$ in the antecedent of (\ref{m5b}), this way forcing Environment to choose a constant $c$ for $z$ there, and also to choose one of the disjuncts of $ x\equals b\plus c\add b\equals x\plus c$. 

{\em Subcase 2.1}:  $b\equals x\plus c$ is chosen. Using Example \ref{onesuc}, we calculate the value $d$ of $c\one$, specify $z$ as $d$ in the consequent of (\ref{m5b}), and choose the right $\add$-disjunct there.  By now, 
(\ref{m5b}) is brought down to 
\[  (|b|\mleq |t|\mli  |c|\mleq |x|\plus |b|\mlc b\equals x\plus c ) \mli  
 (|b\one|\mleq |t|\mli  |c\one|\mleq |x\zero|\plus |b\one|\mlc b\one\equals x\zero\plus c\one) . 
\]
According to $\pa$, the above is true, so we win. 

{\em Subcase 2.2}: $x\equals b\plus c$ is chosen. First, using Example \ref{113}, we figure our whether $c\equals 0$ or ($\add$) $c\notequals 0$. 

{\em Subsubcase 2.2.1}: $c\equals 0$. Using Axiom 8, we calculate the value $d$ of $0\successor$, specify $z$ as $d$ in the consequent of (\ref{m5b}), and choose the right $\add$-disjunct there.  This means that, by now, 
(\ref{m5b}) is brought down to 
\[\begin{array}{l}
 (|b|\mleq |t|\mli    |0|\mleq |x|\plus |b|\mlc x\equals b\plus 0     ) \mli   
 (|b\one|\mleq |t|\mli    |0\successor|\mleq |x\zero|\plus |b\one|\mlc   b\one\equals x\zero\plus 0\successor) 
\end{array}
\]
which, by $\pa$, is true, so we win. 

{\em Subsubcase 2.2.2}: $c\notequals 0$. Using Axiom 9 and Fact \ref{m4h}, we calculate $d$ with $ c\zero\equals d\successor$, i.e. $d\equals c\zero\minus 1$,  specify $z$ as $d$ in the consequent of (\ref{m5b}), and choose the left $\add$-disjunct there.  This means that, by now, 
(\ref{m5b}) is brought down to the following true (by $\pa$) position, so we win:
\[
  (|b|\mleq |t|\mli  |c|\mleq |x|\plus |b|\mlc x\equals b\plus c ) \mli  
 \bigl(|b\one|\mleq |t|\mli  |c\zero\minus 1|\mleq |x\zero|\plus |b\one|\mlc x\zero\equals b\one\plus (c\zero\minus 1)\bigr) . 
\]

The {\em right inductive step} is handled in a rather similar way, and it is left as an exercise.

Thus, we have proven (\ref{m5a}). Now, the target sentence  $\ada x \ada y\ade z( x\equals y\plus z\add y\equals x\plus z)\bigr)$ can be easily seen to be a logical consequence of (\ref{m5a}) and the $\pa$-provable fact $\cla y(|y|\mleq |y|)$.\vspace{-7pt}
\end{proof}

\begin{fact}\label{comid}
%\marginpar{comid}
 $\arfour\vdash  x\equals y\add x\notequals y $.
\end{fact}

\begin{proof} 
Argue in $\arfour$. In order to solve $x\equals y\add x\notequals y$, using Fact \ref{minus}, we find the difference $a$ between $x$ and $y$. Further, using Example \ref{113}, we figure out whether $a\equals 0$ or $a\notequals 0$. If $a\equals 0$, we choose the left $\add$-disjunct, otherwise we choose the right $\add$-disjunct. 
\end{proof}

For natural numbers $n$ and $i$ --- as always identified with the corresponding binary numerals --- such that $i\mless |n|$,  in our metalanguage,  we let 
$[n]_i$ mean bit $\# i$ of $n$, where    the count of the bits of $n$ starts from $0$ rather than $1$, and proceeds from left to right. So, for instance, if $n=100$, then $1$ is its bit $\#0$, and the $0$s are its bits $\#1$ and $\#2$. We  treat $[n]_i$ as a pseudoterm just like  $|x|$, meaning that we can feel free to write expressions such as $[x]_y\equals z$, understood as abbreviations, in formulas of $\arfour$.

\begin{fact}\label{combit}
%\marginpar{combit}
 $\arfour\vdash  y\mless |x|\mli \ade z(z\equals [x]_y)$.  
\end{fact}

\begin{proof} Argue in $\arfour$. By induction on $x$, we want to show 
%\marginpar{m5c}
\begin{equation}\label{m5c}
\ada y \Bigl(|y|\mleq |x|\mli \bigl(y\mless |x|\mli \ade z(|z|\mleq 0\successor \mlc z\equals [x]_y)\bigr)\Bigr).
\end{equation}
The basis 
 $|y|\mleq |0|\mli \bigl(y\mless |0|\mli \ade z(|z|\mleq 0\successor \mlc z\equals [0]_y)\bigr)$ is obviously solved by a strategy that makes no moves.

To solve the left inductive step  
\[\ada y \Bigl(|y|\mleq |x|\mli \bigl(y\mless |x|\mli \ade z(|z|\mleq 0\successor \mlc z\equals [x]_y)\bigr)\Bigr)\mli  
\ada y \Bigl(|y|\mleq |x\zero|\mli \bigl(y\mless |x\zero|\mli \ade z(|z|\mleq 0\successor \mlc z\equals [x\zero]_y)\bigr)\Bigr),\]
we wait till Environment selects a value $a$ for $y$ in the consequent. 
Using   Facts \ref{comlog} and \ref{comid}, we figure out whether $a\equals |x|$ or not. If yes, we select $0$ for $z$ in the consequent. If not, we choose $a$ for  $y$  in the antecedent   and wait till Environment responds by selecting a constant $b$ for $z$ there, after which we choose the same constant $b$ for $z$ in the consequent.  With a little thought, this strategy can be seen to win. 

The right inductive step has a similar strategy, with the difference that, if $a\equals |x|$, it chooses the value of $0\successor$ (found using Axiom 8) for $z$ in the consequent. 

Now, the target $y\mless |x|\mli \ade z(z\equals [x]_y)$ can be seen to be a logical consequence of (\ref{m5c}) and the $\pa$-provable fact $\cla (y\mless |x|\mli |y|\mleq |x|) $. 
\end{proof}

The exponentiation function $2^x$ increases the size of its argument exponentially and hence, in view of the soundness of $\arfour$, cannot be provably computable. According to the following fact, however, the same is not the case for a  limited version of the function:  

\begin{fact}\label{comexp} $\arfour$ proves both of the following:
%\marginpar{comexp}
\begin{eqnarray}
&  x\mleq |z|\mli \ade y  (y\equals 2^{x}) ; & \label{sos2}\\
&   \ade y  (y\equals 2^{|r|}). & \label{sos1}
\end{eqnarray}  
\end{fact}

\begin{proof} Argue in $\arfour$. By $\arfour$-induction on $x$, we  want to prove 
$x\mleq |z|\mli \ade y  (|y|\mleq |z|\successor \mlc y\equals 2^{x})$, from which (\ref{sos2}) immediately follows by LC. 

 The  basis  $0\mleq |z|\mli \ade y  (|y|\mleq |z|\successor \mlc y\equals 2^{0})$ is obviously solved by choosing the value $a$ of $0\successor$ for $y$. Such an $a$ can be found using Axiom 8. 
The left inductive step is \[\bigl(x\mleq |z|\mli \ade y  (|y|\mleq |z|\successor \mlc y\equals 2^{x})\bigr)\mli \bigl(x\zero\mleq |z|\mli \ade y  (|y|\mleq |z|\successor \mlc y\equals 2^{x\izero})\bigr).\] To solve it, we wait till Environment chooses a value for $y$ in the antecedent. If such a value is never chosen, Environment loses unless  $x\mleq |z|$ is false. But, if $x\mleq |z|$ is false, then so is $x\zero\mleq |z|$, and we win. So, assume $a$ is the constant chosen by Environment in the antecedent for $y$. Using Fact \ref{commul}, we compute $b$ with $b\equals a^2$, and choose $b$ for $y$ in the consequent. We win because the game will have evolved to  $(x\mleq |z|\mli  |a|\mleq |z|\successor \mlc a\equals 2^{x}) \mli (x\zero\mleq |z|\mli   |b|\mleq |z|\successor \mlc b\equals 2^{x\izero})$, i.e. 
\[(x\mleq |z|\mli  |a|\mleq |z|\successor \mlc a\equals 2^{x}) \mli (2x \mleq |z|\mli   |a^2|\mleq |z|\successor \mlc a^2\equals 2^{2x})\] 
which, by $\pa$, is true. 
The right inductive step is similar, with the difference that here we shall choose $b$ to be $2a^2$ rather than $a^2$, computing which will take Axiom 9 in addition to Fact \ref{commul}. 

Thus, (\ref{sos2}) is proven. Now we solve (\ref{sos1}), i.e. $\ade y  (y\equals 2^{|r|})$, as follows. First, using Fact \ref{comlog}, we find the value  $a$ of  $|r|$. Next, using 
(\ref{sos2}) --- namely, specifying its $z$ and $x$ as $r$ and $a$, respectively --- we compute the value $b$ with $b\equals 2^a$, i.e. $b\equals 2^{|r|}$, and choose that $b$ for  $y$ in $\ade y  (y\equals 2^{|r|})$.\vspace{-7pt}  
\end{proof}

We generalize the earlier notation $[x]_{y}$ to $[x]_{y}^{z}$, additionally to  $y\mless |x|$ requiring that $y\plus z\mleq |x|$. It means ``the  substring  of $x$ of length $z$ which starts at the $y$th bit''.   For instance, if $x\equals 111010$, then $[x]_{2}^{3}\equals 101$, $[x]_{0}^{6}\equals x$ and, for each $i\in\{0,\ldots,5\}$, $[x]_{i}^{0}\equals 0$. As always, we identify the bit string $[x]_{y}^{z}$ with the number it represents in the binary notation.  Note that the old $[x]_y$ is the special case of $[x]_{y}^{z}$ with $z\equals 1$. 
The following fact states the provable computability of the function $[x]_{y}^{z}$.

\begin{fact}\label{m6a}
%\marginpar{m6a}
$\arfour\vdash y\mless |x|\mlc y\plus z\mleq |x| \mli \ade t(t\equals [x]_{y}^{z} )$.
\end{fact}

\begin{proof} Argue in $\arfour$.  First, by $\arfour$-induction on $r$, we want to prove  
%\marginpar{m6r}
\begin{equation}\label{m6r}
y\mless |x|\mlc  y\plus |r|\mleq |x| \mli \ade t(|t|\mleq |x|\mlc t\equals [x]_{y}^{|r|} ).
\end{equation}

The basis  $y\mless |x|\mlc  y\plus |0|\mleq |x| \mli \ade t(|t|\mleq |x|\mlc t\equals [x]_{y}^{|0|} )$ is solved straightforwardly by choosing $0$ for $t$. That is because   $0$ stands for the empty bit string, and so does $ [x]_{y}^{|0|}$ for whatever $x,y$ with $y\mless |x|$. 

The right inductive step is 
%\marginpar{m6s}
\begin{equation}\label{m6s}
\bigl(y\mless |x|\mlc y\plus |r|\mleq |x| \mli \ade t(|t|\mleq |x|\mlc t\equals [x]_{y}^{|r|} )\bigr)\mli \bigl(y\mless |x|\mlc y\plus |r\one|\mleq |x| \mli \ade t(|t|\mleq |x|\mlc t\equals [x]_{y}^{|r\ione|} )\bigr).
\end{equation}
Assume $y\mless |x|\mlc y\plus |r\one|\mleq |x|$. Then we also have $y\mless |x|\mlc y\plus |r|\mleq |x|$. Using the antecedental resource $\ade t(|t|\mleq |x|$ 
%[!]
$\mlc t\equals [x]_{y}^{|r|} )$, we find  $a$ with $|a|\mleq |x|$ such that $a\equals [x]_{y}^{|r|}$. Using   Facts \ref{comlog}, \ref{comad} and \ref{combit}, we further find $b$ with $b\equals [x]_{y\iplus |r|}$. Now the sought value of $t$ is the value of $a\zero\plus b$, which we compute using Axiom 9 and Fact \ref{comad}.

To solve the left inductive step, first we figure out (using Example \ref{113}) whether $r\equals 0$ or not. If $r\equals 0$, we ignore the antecedent and act in the consequent as the strategy for the basis of induction did. Otherwise, we act as the strategy for the right inductive step did.  

(\ref{m6r}) is thus proven. To solve the target  $y\mless |x|\mlc  y\plus z\mleq |x| \mli \ade t(t\equals [x]_{y}^{z} )$,  assume $y\mless |x|\mlc y\plus z\mleq |x|$. Then $z\mleq |x|$. Using
Facts \ref{m4gg} and \ref{comexp},  we find the value $a$ of the binary predecessor of $2^z$. Note that $|a|\equals z$. Now, using (\ref{m6r}), we find the value $b$ of $[x]_{y}^{|a|}$, i.e. of $[x]_{y}^{z}$.   Selecting $b$ for $t$ solves the problem.
\end{proof}

\section{The soundness of $\arfour$}\label{sectsound}
%\marginpar{sectsound}
This section is devoted to proving the soundness part of Theorem \ref{tt1}. It means showing that any $\arfour$-provable sentence $X$ (as always, identified with its standard interpretation $X^\dagger$) has a polynomial time solution, and that, furthermore, such a solution for $X$, together with an explicit polynomial bound $\tau$ for its time complexity, can be efficiently extracted from any extended $\arfour$-proof of $X$. Consider any sentence $X$ with a fixed $\arfour$-proof. 

For presentational considerations,  by  induction on the length of the  proof of $X$, we will first simply show that a polynomial time solution of $X$ exists. Only after that, at the end of this section,  we will show that such a solution, together with an explicit polynomial bound for its time complexity, is or can be constructed efficiently.  

Assume $X$ is an axiom of $\arfour$. If $X$ is a Peano axiom, then it is a true  elementary sentence and therefore is won by a machine that makes no moves.  
If $X$ is  $\ada x\ade y(y\equals x\successor)$ (Axiom 8), then it is won by a machine that (for the constant $x$ chosen by Environment for the variable $x$)    computes  the value $a$ of $x\plus 1$, and makes $a$ as its only move in the play. Similarly, if $X$ is $\ada x\ade y(y\equals x\zero)$ (Axiom 9),  it is won by a machine that computes  the value $a$ of $2x$, and makes $a$ as its only move in the play. Needless to point out that all of the above  machines  run in polynomial time.

Next, suppose $X$ is obtained from premises $Y_1,\ldots,Y_n$ by LC. By the induction hypothesis, for each $i\in\{1,\ldots,n\}$, we already have a solution (HPM) ${\cal N}_i$ of $Y_i$ together with an explicit polynomial bound  $\xi_i$ for the time complexity of ${\cal N}_i$. Of course, every such HPM ${\cal N}_i$ can as well be seen as an $n$-ary GHPM that simply ignores its inputs. Then, by Theorem \ref{feb9d}, we can (efficiently) construct an $n$-ary GHPM $\cal M$, together with an explicit polynomial bound $\tau(\xi_1,\ldots,\xi_n)$ for the time complexity of the HPM 
${\cal M}(\code{{\cal N}_1},\ldots,\code{{\cal N}_n})$ such that the latter solves $X$.  

Finally, suppose $X$ is (the $\ada$-closure of) $F(x)$, where $F(x)$ is a polynomially bounded formula, and $X$ is obtained by $\arfour$-Induction on $x$. So, the premises are (the $\ada$-closures of) $F(0)$, $F(x)\mli F(x\zero)$ and $F(x)\mli F(x\one)$. By the induction hypothesis, there are HPMs ${\cal N}$, ${\cal K}_0$,${\cal K}_1$ --- with explicit polynomial bounds $\xi,\zeta_0,\zeta_1$ for their time complexities, respectively --- that solve these three premises, respectively.  Fix them.  

We need certain auxiliary concepts. Consider  any polynomially bounded formula $H$,  any legal position $\Phi$ of $\ada H$, and any legal move $\alpha$ (by whichever player $\xx$) in position $\Phi$. 
In this context, we say that $\alpha$ is {\bf unreasonable} if it signifies a choice of a constant $c$ for a variable $y$ in a $\ade y\bigl(S(y,\vec{z})\mlc G\bigr)$ or $\ada y\bigl(S(y,\vec{z})\mli G\bigr)$ (depending on whether $\xx=\pp$ or $\xx=\oo$) subcomponent of $H$, such that $c$ violates the conditions on its size imposed by the sizebound $S(y,z)$. Rather than trying to turn this otherwise clear intuitive explanation into a strict definition, providing an example would be sufficient. Let  $H$ be the formula $0\equals 0\mlc \ade y(|y|\mleq |z|\successor\mlc z\equals y\zero)$. Then the move $1.1111$ is unreasonable in position $\seq{\oo 11}$. That is because
this move signifies choosing the constant $1111$ for $y$. And the move $\oo 11$ of the position has set the value of $z$ to $11$ and hence the value of $|z|$ to $2$. So, the condition $|y|\mleq |z|\successor$, i.e. $|1111|\mleq |11|\successor$, is violated.  Any other move $1.n$ with $|n|\mgreater 3$, such as $1.1000$ or $1.111111111$, would also be unreasonable in that position.

We replace  ${\cal N}$   by its ``{\bf reasonable counterpart}'' ${\cal N}'$ --- an HPM  which never makes unreasonable moves but otherwise is essentially the same as $\cal N$.   Namely, ${\cal N}'$ is a machine that works just like $\cal N$, with the difference that, every time $\cal N$ makes an unreasonable move that chooses some (offensively long) constant $c$ for a variable bound by a (bounded) quantifier of $F(0)$, ${\cal N}'$ chooses (the always safe) $0$ instead. Note that this does not decrease the chances of the machine to win, as unreasonable moves always result in the corresponding subgames' being lost, anyway. Obviously ${\cal N}'$ can be efficiently constructed from $\cal N$. Further, the corresponding explicit polynomial bound $\xi'$ can also be efficiently indicated (the latter will depend on $\xi$ and the sizebounds of the $\ade$-bound variables of $F(0)$). 
In a similar fashion, we replace ${\cal K}_0$,${\cal K}_1$ by their ``reasonable counterparts'' ${\cal K}'_0$,${\cal K}'_1$ and the corresponding explicit polynomial bounds $\zeta'_0,\zeta'_1$ for their time complexities. For simplicity, we further replace the three bounds $\xi',\zeta'_0,\zeta'_1$ by the (generously taken) common  bound $\phi\equals \xi'\plus \zeta'_0\plus\zeta'_1$ for the time complexities of all three machines ${\cal N}'$, ${\cal K}'_0$ and ${\cal K}'_1$.

We now describe an HPM $\cal M$ that solves the conclusion $F(x)$. In this description, 
we use the term ``{\bf synchronizing}''\label{imatching} to mean applying copycat between two (sub)games of the form $A$ and $\gneg A$. This means copying one player's moves in $A$ as the other player's moves in $\gneg A$, and vice versa. The effect achieved this way is that the games to which $A$ and $\gneg A$  eventually evolve (the final positions hit by them, that is) will be of the form $A'$ and $\gneg A'$, that is, one will remain the negation of the other, so that one will be won by a given player iff the other is lost by the same player.  {\bf Moderated synchronization} means the same, with the only difference that, whenever a player makes an unreasonable move by choosing an (offensively long) constant $c$ for a variable bound by a bounded quantifier, the move is copied by the synchronizer with $c$ replaced by $0$. 

Throughout our description and analysis of the work of $\cal M$, we assume that its adversary never makes illegal moves, for otherwise $\cal M$ easily detects illegal behavior and retires with victory. 

At the beginning,  $\cal M$  waits for Environment to choose constants for the free variables of $F(x)$.   Assume $k$ is the length of the constant chosen for the variable $x$, and the bits of that constant, in the left-to-right order, are   $b_1,b_{2},\ldots b_k$. We shall also assume here that $k\notequals 0$, for otherwise the case is straightforward.   Let $d_0$ be the constant $0$ and, for each $i\in\{1,\ldots,k\}$, let $d_i$ be the constant $b_1\ldots b_i$. So, the constant chosen by Environment for $x$ is $d_k$. For each $i\in\{1,\ldots,k\}$, let ${\cal K}'_{b_i}$ stand for ${\cal K}'_0$ if $b_i=0$, and for ${\cal K}'_1$ if $b_i=1$. Similarly, let $xb_i$ stand for $x\zero$ if $b_i=0$, and for $x\one$ if $b_i=1$. 

After Environment chooses constants for all free variables of $F(x)$, the work of $\cal M$ consists in continuously polling its run tape to see if Environment has made any new moves, combined with simulating, in parallel, one play of $\ada\bigl(F(0)\bigr)$  by ${\cal N}'$ and --- for each $i\in\{1,\ldots,k\}$ --- 
one play  of $\ada\bigl(F(x)\mli F(x b_i)\bigr)$ by ${\cal K}'_{b_i}$. In the simulation of ${\cal N}'$, $\cal M$ lets the imaginary adversary of 
${\cal N}'$ choose, at the very beginning of the play, the same constants for the free variables of $F(0)$ as $\cal M$'s adversary chose for those variables in the real play.  In the simulation of each ${\cal K}'_{b_i}$, $\cal M$ lets the imaginary adversary of ${\cal K}'_{b_i}$
 choose, at the very beginning of the play, the constant $d_{i - 1}$ for $x$ and the same constants for all other free variables of $F(x)\mli F(x b_i)$ as $\cal M$'s adversary chose for those variables in the real play.

 Let $F'(x)$ be the result of substituting (see Definition \ref{sov}) in $F(x)$ each free variable of $F(x)$ other than $x$ by the constant chosen by (the real) Environment for that variable. Thus, after Environment's initial moves, $\ada\bigl(F(x)\bigr)$ has been brought down to $F'(d_k)$. Similarly, after the initial moves by the imaginary adversary of ${\cal N}'$, $\ada\bigl(F(0)\bigr)$ will be brought down to $F'(0)$. And similarly, for each $i\in\{1,\ldots,k\}$, after the initial moves by the imaginary adversary of ${\cal K}'_{b_i}$, $\ada\bigl(F(x)\mli F(x b_i)\bigr)$ will be brought down to $F'(d_{i-1})\mli F'(d_i)$.  

What $\cal M$ does after the above initial moves in the real and simulated plays is that it synchronizes $k\plus 1$ pairs of (sub)games, real or imaginary. Namely:

\begin{itemize}
\item It synchronizes --- in the {\em moderated} sense ---  the consequent of the imaginary play of $F'(d_{k-1})\mli F'(d_k)$ by ${\cal K}'_{b_k}$ with the real play of $F'(d_k)$.  
\item For each $i\in\{1,\ldots,k\minus 1\}$,  it synchronizes the consequent of the imaginary play of  $F'(d_{i-1})\mli F'(d_i)$ by ${\cal K}'_{b_i}$ with the antecedent of the 
 imaginary play of $F'(d_{i})\mli F'(d_{i+1})$ by ${\cal K}'_{b_{i+1}}$. 
\item It synchronizes  the  imaginary play of $F'(0)$ (i.e. of $F'(d_0)$) by ${\cal N}'$ with the antecedent of the imaginary play of $F'(d_{0})\mli F'(d_1)$ by ${\cal K}'_{b_1}$. 
\end{itemize}

Below is an illustration of such synchronization arrangements --- indicated by arcs --- for the case $d_k\equals d_4\equals 1001$:

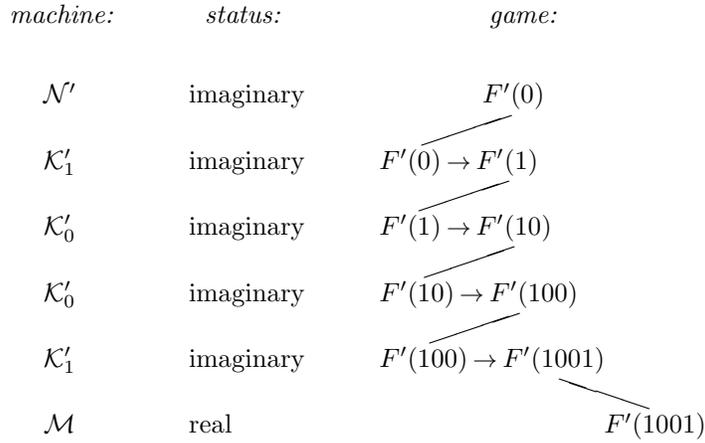
\begin{center}
\begin{picture}(260,175)
\put(0,165){\em machine:\hspace{30pt} status: \hspace{70pt} game:}
\put(13,135){${\cal N}'$}
\put(68,135){imaginary}
\put(179,135){$F'(0)$}
\put(13,110){${\cal K}'_1$}
\put(68,110){imaginary}
\put(140,110){$F'(0)\mli F'(1)$}

\put(13,85){${\cal K}'_0$}
\put(68,85){imaginary}
\put(140,85){$F'(1)\mli F'(10)$}

\put(13,60){${\cal K}'_0$}
\put(68,60){imaginary}
\put(140,60){$F'(10)\mli F'(100)$}

\put(13,35){${\cal K}'_1$}
\put(68,35){imaginary}
\put(140,35){$F'(100)\mli F'(1001)$}

\put(13,10){${\cal M}$}
\put(68,10){real}
\put(225,10){$F'(1001)$}

\put(156,119){\line(3,1){34}}

\put(155,94){\line(3,1){34}}

\put(157,69){\line(3,1){34}}

\put(159,44){\line(3,1){34}}

\put(242,19){\line(-3,1){34}}

\end{picture}
\end{center}

This completes our description of $\cal M$. 
%Of  course, $\cal M$  can be constructed efficiently from $F(x)$ and (${\cal N}',{\cal K}'_0,{\cal K}'_1$ and hence) ${\cal N},{\cal K}_0,{\cal K}_1$.  
Remembering our assumption that (${\cal N},{\cal K}_0,{\cal K}_1$ and hence) ${\cal N}',{\cal K}'_0,{\cal K}'_1$ win the corresponding games, with a little thought it can be seen  that $\cal M$ wins $F'(d_k)$ and hence $\ada\bigl(F(x)\bigr)$, as desired. It now remains to show that the time complexity of $\cal M$ is also as desired.  

Remembering that the machines ${\cal N}',{\cal K}'_0,{\cal K}'_1$ are ``reasonable'' and that the synchronization between the real play of $F'(d_k)$ and the consequent of $F'(d_{k-1})\mli F'(d_k)$ is moderated, one can easily write a term $\eta(\ell)$ with a single variable $\ell$ such that, if $\ell$ is greater than or equal to the size of any of the constants chosen by Environment for the free variables of $F(x)$,  the sizes of no moves ever made by $\cal M$ or the simulated ${\cal N}',{\cal K}'_0,{\cal K}'_1$ exceed $\eta(\ell)$. For instance, if $F(x)$ is \(\ade u\bigl(|u|\mleq |x|\mult |z|\mlc\ada v(|v|\mleq |u|\plus |x|\mli G)
\bigl)\) where $G$ is elementary, then $\eta(\ell)$ can be taken to be $\ell\mult \ell\plus \ell\plus 0\successor\successor\successor\successor$ (here $0\successor\successor\successor\successor$ is to account for the size of the prefix  ``$.$'', ``$1.$''  or ``$0.1.$'' that any legal move by $\pp$ in any of the plays that we consider would take).  

For the rest of this proof, pick and fix an arbitrary play (computation branch) of $\cal M$, and an arbitrary clock cycle $\mathfrak{c}$ on which $\cal M$ makes a move $\alpha$ in the real play of $F(x)$. Let $\hbar$ and $\ell$ be the timecost and the background (see Section \ref{s7}) of this move, respectively. Let 
$d_0,\ldots,d_k$  be as in the description of the work of $\cal M$. Note that $\ell$ is not smaller than the size of the greatest of the constants chosen by Environment for the free variables of $F(x)$. Hence, where $\eta$ is as in the previous paragraph, we have:
%\marginpar{m12a}
\begin{equation}\label{m12a}
\mbox{\em The sizes of no moves ever made by $\cal M$ or the simulated ${\cal N}',{\cal K}'_0,{\cal K}'_1$ exceed $\eta(\ell)$.}
\end{equation}   

The polling, simulation and copycat performed by $\cal M$ do impose some time overhead. But the latter is only (fixed) polynomial and, in our subsequent analysis, can be safely ignored. Namely, for the sake of simplicity, we are going to pretend that  $\cal M$ copies moves in its copycat routine instantaneously (as soon as detected), and that the times that $\cal M$ ever spends ``thinking'' about what move to make are the times during which it is waiting for simulated machines to make one or several moves. Furthermore, we will pretend that the polling and the several simulations happen in a truly parallel fashion, in the sense that $\cal M$ spends a single clock cycle on tracing a single computation step of  {\em all} $k\plus 1$ machines simultaneously, as well as on checking out its run tape to see if Environment has made a new move.

Let $\beta_1,\ldots,\beta_m$ be the moves by simulated machines that $\cal M$ detects by time $\mathfrak{c}$, arranged according to the times $t_1\mleq \ldots\mleq t_m$ of their detections (which, by our simplifying assumptions, coincide with the timestamps of those moves in the corresponding simulated plays).      
Let   $d\equals\mathfrak{c}\minus \hbar$. Let $j$ be the smallest integer among $1,\ldots,m$ such that $t_j\mgeq d$. Since each simulated machine  runs in time $\phi$, in view of (\ref{m12a}) it is clear that $t_j\minus d$ does not exceed $\phi(\eta(\ell))$. Nor does $t_{i}\minus t_{i\iminus 1}$ for any $i$ with 
$j\mless i\mleq m$. Therefore $t_m\minus d\mleq  (m\minus j\plus 1)\mult \phi(\eta(\ell))$. Since $m,j\mgeq 1$, let us be generous and simply say that $t_m\minus d\mleq  m\mult \phi(\eta(\ell))$.
But notice that $\beta_m$ is a move made by ${\cal K}'_{b_k}$ in the consequent of $F'(d_{k-1})\mli F(d_k)$, immediately (by our simplifying assumptions) copied by $\cal M$ in the real play when it made its move $\alpha$. In other words, $\mathfrak{c}\equals t_m$. And $\mathfrak{c}\minus d\equals \hbar$. So, $\hbar$ does not exceed $m\mult  \phi(\eta(\ell))$.   And, by (\ref{m12a}), the size of $\alpha$ does not exceed $m\mult  \phi(\eta(\ell))$, either. But observe that $k\mleq \ell$, and that $m$ cannot exceed   $k\plus 1$ times the depth (see Section \ref{nncg}) $\mathfrak{d}$ of $F(0)$; therefore,    $m\mleq  \mathfrak{d}\mult (\ell\plus 1)$. Thus, (as long as we pretend that there is no polling/simulation/copycat overhead) neither the timecost nor the size of $\alpha$ exceed $\mathfrak{d}\mult (\ell\plus 1)\mult \phi(\eta(\ell))$. 

An upper bound for the above function $\mathfrak{d}\mult (\ell\plus 1)\mult \phi(\eta(\ell))$, even after ``correcting'' the latter so as to precisely  account for the so far suppressed polling/simulation/copycat overhead, can be written as an explicit polynomial function $\tau$.   The latter expresses the sought   polynomial bound for the time complexity of $\cal M$. 
%It remains to notice the term $\tau(\ell)$ can be constructed efficiently. 

Thus, we have shown how to construct, from a proof of $X$, an HPM $\cal M$ and an explicit  polynomial function $\tau$ such that $\cal M$ solves $X$ in time $\tau$. Obviously our construction is effective. It remains to see that it also is --- or, at least, can be made --- efficient. 
Of course, at every step of our inductive construction (for each sentence  of the proof, that is),  the solution $\cal M$ of the step and its time complexity bound $\tau$ is obtained efficiently from previously constructed $\cal M$s and $\tau$s. This, however, does not guarantee that the entire construction will be efficient as well. For instance, if the proof has $n$ steps and the size of each HPM $\cal M$ that we construct for each step is twice the size of the previously constructed HPMs, then the size of the eventual HPM will exceed $2^n$ and thus the construction will not be efficient, even if each of the $n$ steps of it is so. 

A trick that we can use to avoid an exponential growth of the sizes of the machines that we construct and thus achieve the efficiency of the entire construction is to deal with GHPMs instead of HPMs. Namely, assume the proof of $X$ is the sequence $X_1,\ldots,X_n$ of sentences, with $X=X_n$. Let ${\cal M}_1,\ldots,{\cal M}_n$ be the HPMs constructed as we constructed ${\cal M}$s earlier at the corresponding  steps of our induction. Remember that each such ${\cal M}_i$ was defined in terms of  ${\cal M}_{j_1},\ldots,{\cal M}_{j_k}$for some $j_1,\ldots,j_k< i$. For simplicity and uniformity, we may just as well say that each ${\cal M}_i$ was defined in terms of all  ${\cal M}_{1},\ldots,{\cal M}_{n}$, with those ${\cal M}_j$s that were not among ${\cal M}_{j_1},\ldots,{\cal M}_{j_k}$ simply ignored in the description of the work of ${\cal M}_i$. 
Now, for each such ${\cal M}_i$, let ${\cal M}'_i$ be the $n$-ary GHPM whose description is obtained from that of ${\cal M}_i$  by replacing  each reference to (any previously constructed) ${\cal M}_j$  by ``${\cal M}'_j(\code{{\cal M}'_1},\ldots,\code{{\cal M}'_n})$ where, for each $e\in\{1,\ldots,n\}$,  ${\cal M}'_e$ is the machine encoded by the $e$th input''.\footnote{For simplicity, here we assume that every number is a code of some $n$-ary GHPM; alternatively, ${\cal M}'_i$ can be defined so that it does nothing if any of its relevant inputs is not the code of some $n$-ary GHPM.} As it is easy to see by induction on $i$, ${\cal M}_i$ and ${\cal M}'_i(\code{{\cal M}'_1},\ldots,\code{{\cal M}'_n})$ are essentially the same, in the sense that our earlier analysis of the play and time complexity of the former applies to the latter just as well. So, ${\cal M}'_n(\code{{\cal M}'_1},\ldots,\code{{\cal M}'_n})$ wins $X_n$, i.e. $X$. At the same time, note that the size of each GHPM ${\cal M}'_i$ is independent of the sizes of the other (previously constructed) GHPMs. Based on this fact, with some analysis, one can see that then the HPM ${\cal M}'_n(\code{{\cal M}'_1},\ldots,\code{{\cal M}'_n})$  is indeed constructed efficiently. 

As for the explicit polynomial bounds $\tau_1,\ldots,\tau_n$ for the time complexities of the $n$ HPMs ${\cal M}'_1(\code{{\cal M}'_1},\ldots,\code{{\cal M}'_n})$, \ldots, ${\cal M}'_n(\code{{\cal M}'_1},\ldots,\code{{\cal M}'_n})$,  their sizes  can be easily seen to be polynomial in the size of the proof. That is because, for each $i\in\{1,\ldots,n\}$, the size of   $\tau_i$  only increases the sizes of the earlier constructed $\tau_j$s by adding (rather than multiplying by) a certain polynomial quantity.\footnote{The fact that we represent complexity bounds as explicit polynomial functions rather than polynomial tree-terms or even graph-terms is relevant here.} Thus, the explicit bound $\tau_n$ for the time  complexity of the eventual HPM    ${\cal M}'_n(\code{{\cal M}_1},\ldots,\code{{\cal M}_n})$ is indeed constructed efficiently. 

\section{The extensional completeness of $\arfour$}\label{sectcompl}\label{s19}
%\marginpar{s19} 
This section is devoted to proving the completeness part of Theorem \ref{tt1}. It means showing that, for any arithmetical problem $A$ that has a polynomial time solution,  there is a theorem of $\arfour$ which, under the standard interpretation, equals (``expresses'') $A$. 

\subsection{$X$, $\cal X$ and $\chi$}
So, let us pick an arbitrary polynomial-time-solvable arithmetical problem $A$. By definition, $A$ is an arithmetical problem because, for some sentence $X$ of the language of $\arfour$, $A=X^\dagger$. For the rest of this section, we fix such a sentence \(X,\label{ix}\)  and fix \({\cal X}\) as an HPM that solves $A$ (and hence $X^\dagger$) in polynomial time. Specifically, we assume that $\cal X$ runs in time 
$\chi$, 
where $\chi$, which we also fix for the rest of this section,  is a single-variable term of the language of $\pa$ --- and hence can as well be seen/written as an explicit polynomial function --- with $\chi(x)\mgeq x$ for all $x$.  For readability, we also agree that, throughout the rest of this section, ``{\bf formula}'' exclusively means a subformula of $X$, in which some variables may  be renamed. 

$X$ may not necessarily be provable in $\arfour$, and our goal is to construct another sentence $\overline{X}$ for which, just like for $X$, we have $A=\overline{X}^\dagger$ and which, perhaps unlike $X$, is  provable in $\arfour$.

Remember our convention about identifying formulas  of the language of $\arfour$ with (the games that are) their standard interpretations. So, in the sequel, just as we have done so far, we shall typically write $E,F,\ldots$ to mean either $E,F,\ldots$ or $E^\dagger,F^\dagger,\ldots$.  
Similar conventions apply to terms as well. In fact, we have already used this convention when saying that $\cal X$ runs in time $\chi$. What was really meant was that it runs in time $\chi^\dagger$.

\subsection{Preliminary insights}\label{gggg}
%\marginpar{gggg}
 Our proof is a little long and, in the process of going through it, it is easy to get lost in the forest and stop seeing it for the trees. Therefore, it might be worthwhile to try to get some preliminary insights into the basic idea behind this proof before venturing into its details.

Let us consider the simplest nontrivial special case  where $X$ is \[\ada x\bigl(Y(x)\add Z(x)\bigr) \] for some elementary formulas $Y(x)$ and $Z(x)$
(perhaps $Z(x)$ is $\gneg Y(x)$, in which case $X$ expresses   an ordinary decision problem --- the problem of deciding  the predicate $Y(x)$). 

The assertion ``$\cal X$ does not win $X$ in time $\chi$'' can be formalized in the language of $\pa$ through as a certain sentence $\mathbb{L}$. Then we let the earlier mentioned  $\overline{X}$ be the sentence 
\[\ada x\Bigl(\bigl(Y(x)\mld \mathbb{L}\bigr)\add\bigl(Z(x)\mld\mathbb{L}\bigr)\Bigr).\]
Since $\cal X$ {\em does} win the game $X$ in time $\chi$, $\mathbb{L}$ is false. Hence $Y(x)\mld \mathbb{L}$ is equivalent to $Y(x)$, and  $Z(x)\mld \mathbb{L}$ is equivalent to $Z(x)$. This means that $\overline{X}$ and $X$, as games, are the same, that is,  $\overline{X}^\dagger=X^\dagger$. It now remains to understand why $\arfour\vdash \overline{X}$.  

A central lemma here is one establishing that the work of $\cal X$ is ``{\em provably traceable}''.\label{iprovtr1} Roughly, in our present case this means the provability of the fact that, for any ($\ada$) value chosen  for $x$ by Environment --- let us continue referring to that value as $x$ --- we can tell ($\ade$)  the configuration of $\cal X$ in the corresponding play of $\overline{X}$ at any given time $t$. Letting $\cal X$ work for $\chi(x)$ steps,  one of the following four eventual scenarios should take place, and the provable traceability of the work of $\cal X$ can be shown to imply that $\arfour$ proves the $\add$-disjunction of sentences describing those scenarios:

\begin{description}
\item[Scenario 1:] $\cal X$ makes the move $0$ (and no other moves).
\item[Scenario 2:] $\cal X$ makes the move $1$ (and no other moves).
\item[Scenario 3:] $\cal X$ does not make any moves.  
\item[Scenario 4:] $\cal X$ makes an illegal move (perhaps after   first making a legal move $0$ or $1$). 
\end{description}

In the case of Scenario 1, the play over $\overline{X}$ hits $Y(x)\mld \mathbb{L}$. And $\arfour$ --- in fact, $\pa$ --- proves that, in this case, $Y(x)\mld \mathbb{L}$ is true. The truth of $Y(x)\mld \mathbb{L}$ is indeed very easily established: if it was false, then $Y(x)$ should be false, but then the play of $\cal X$ over $X$ hits the false $Y(x)$ and hence is lost, but then $\mathbb{L}$ is true, but then $Y(x)\mld \mathbb{L}$ is true. Thus, 
\(\arfour\vdash (\mbox{\em Scenario 1})\mli Y(x)\mld \mathbb{L}$, from which, by LC, $\arfour\vdash (\mbox{\em Scenario 1})\mli \overline{X}$.  
The case of Scenario 2 is symmetric.

In the case of Scenario  3, ($\arfour$ proves that) $\cal X$ loses, i.e.   $\mathbb{L}$ is true, and hence, say, $Y(x)\mld \mathbb{L}$ (or $Z(x)\mld \mathbb{L}$ if you like) is true. That is, $\arfour\vdash   (\mbox{\em Scenario 3})\mli Y(x)\mld \mathbb{L}$, from which, by LC, $\arfour\vdash (\mbox{\em Scenario 3})\mli\overline{X}$.  
The case of Scenario 4 is similar. 

Thus, for each $i\in\{1,2,3,4\}$, $\arfour\vdash (\mbox{\em Scenario i})\mli \overline{X}$.  
And, as mentioned,  we also have  
\[\arfour\vdash (\mbox{\em Scenario 1})\add (\mbox{\em Scenario 2})\add (\mbox{\em Scenario 3})\add (\mbox{\em Scenario 4}).\]
The desired $\arfour\vdash \overline{X}$ follows from the above provabilities by LC.

The above was about the pathologically simple case of $X=\ada x\bigl(Y(x)\add Z(x)\bigr)$, and the general case will be much more complex, of course. Among other things,  showing the provability of $\overline{X}$ would require a certain metainduction on its complexity, which we did not need in the present case. But the   idea that we have just tried to explain would still remain valid and central, only requiring certain --- nontrivial yet doable ---  adjustments and refinements.

\subsection{The sentence $\mathbb{L}$}

By a {\bf literal} we mean $\twg$, $\tlg$, or an atomic formula  with or without negation $\gneg$. By a {\bf politeral}\label{ipoliteral} of a formula we mean a positive (not in the scope of $\gneg$) occurrence of a literal in it. For instance, the occurrence of $p$, as well as of $\gneg q$ (but not $q$), is a politeral of $p\mlc\gneg q$.
While a politeral is not merely a literal but a literal $L$  {\em together} with a fixed occurrence, we shall often refer to it just by the name $L$ of the literal, assuming that it is clear from the context which (positive) occurrence of $L$ is meant.

We assume that the reader is sufficiently familiar with G\"{o}del's technique of encoding and arithmetizing. Using that technique, 
we can construct a sentence   \(\mathbb{L}\label{illl}\) of the language of $\pa$     which asserts ``$\cal X$ does not win $X$ in time $\chi$''.  
Namely, let $E_1(\vec{x}),\ldots,E_n(\vec{x})$ be  all   subformulas of $X$, where all free variables of each $E_i(\vec{x}) $ are among $\vec{x}$ (but not necessarily vice versa). 
Then  
$\mathbb{L}$   is the $\mld$-disjunction of  natural formalizations of the following statements: 
\begin{quote} {\em 
\begin{enumerate}
\item There is a $\pp$-illegal position of $X$ spelled on the run tape of $\cal X$ on some clock  cycle of some computation branch of $\cal X$. 
\item There is a clock cycle $c$ in some computation branch of $\cal X$ on which $\cal X$ makes a move whose timecost exceeds   $\chi(\ell)$, where $\ell$ is the background of $c$. 
\item There is a (finite) legal run $\Gamma$ of $X$ generated by $\cal X$ and a tuple $\vec{c}$ of constants ($\vec{c}$ of the same length as $\vec{x}$) such that:
\begin{itemize}
\item  $\seq{\Gamma}X=E_1(\vec{c})$, and we have $\gneg  \elz{E_1(\vec{c})} $ (i.e., $ \elz{E_1(\vec{c})} $ is false),  
\item  or $\ldots$, or 
\item $\seq{\Gamma}X=E_n(\vec{c})$, and we have $\gneg  \elz{E_n(\vec{c})} $ (i.e., $ \elz{E_n(\vec{c})} $ is false). 
\end{itemize}
\end{enumerate} } 
\end{quote}

   \subsection{The overline notation}
As we remember, our goal is to construct a formula $\overline{X}$ which expresses the same problem as $X$ does and which is provable in $\arfour$. For any   formula $E$ --- including $X$ --- we let 
\(\overline{E}\label{ipver}\)
be the result of replacing in $E$ every politeral $L$ by $L\mld\mathbb{L}$.

\begin{lemma}\label{august12}
%\marginpar{august12}
Any  literal  $L $ is equivalent (in the standard model of arithmetic) to  $L\mld\mathbb{L}$. 
\end{lemma}

\begin{proof} That $L$ implies $L\mld\mathbb{L}$ is immediate, as the former is a disjunct of the latter. For the opposite direction, suppose     $L\mld\mathbb{L}$ is true at a given valuation $e$. Its second disjunct cannot be true, because $\cal X$ {\em does} win $X$ in time $\chi$, contrary to what $\mathbb{L}$ asserts.  So, the first disjunct, i.e. $L$, is true.\vspace{-7pt} 
\end{proof}

\begin{lemma}\label{august12a}
%\marginpar{august12a}
For any  formula $E$, including $X$, we have  $E^\dagger=\overline{E}^\dagger$. 
\end{lemma}

\begin{proof} Immediately from Lemma \ref{august12} by induction on the complexity of $E$.
\end{proof}

In view of the above lemma, what now remains to do for the completion of our completeness proof is to show that $\arfour\vdash\overline{X}$. The rest of the present section is entirely devoted to this task.

\begin{lemma}\label{jan4d}
%\marginpar{jan4d}
For any    formula $E$, $\arfour\vdash \mathbb{L} \mli \cla \overline{E}$. 
\end{lemma}

\begin{proof} Induction on the complexity of $E$. The base, which is about the cases where $E$ is a literal, is straightforward, as then $\mathbb{L}$ is a disjunct of $\overline{E}$.
If $E$ has the form $H_0\mlc H_1$, $H_0\mld H_1$, $H_0\adc H_1$ or $H_0\add H_1$ then, by the induction hypothesis, $\arfour$ proves $\mathbb{L} \mli  \cla \overline{H_0}$ and $\mathbb{L} \mli  \cla \overline{H_1}$, from which $\mathbb{L} \mli  \cla \overline{E}$ follows by LC.   Similarly, if $E$ has the form 
$\cla xH(x)$, $\cle x H(x)$, $\ada xH(x)$ or $\ade xH(x)$, then, by the induction hypothesis,  $\arfour$ proves $\mathbb{L} \mli  \cla \overline{H(x)}$, from which  
$\mathbb{L} \mli  \cla \overline{E}$ follows by LC.\vspace{-7pt}
\end{proof}

\subsection{The single-circle and double-circle notations}\label{sds}
%\marginpar{sds}
The way we encode configurations through natural numbers will be precisely described later in Section \ref{sA1}. For now, it would be sufficient to say that 
the size of the code of a configuration is always greater than the background (see Section \ref{s7}) of the corresponding clock cycle in the corresponding play. For readability, we will often identify configurations with their codes and say something like ``$a$ is a configuration'' when what is precisely meant is ``$a$ is the code of a configuration''.  

By a {\bf legitimate configuration} we shall mean a configuration of $\cal X$ that might have occurred in some computation branch $B$ of $\cal X$ such that the run spelled by $B$ is a legal run of $X$. The {\bf yield} of such a configuration is the game $\seq{\Phi}X$, where  $\Phi$ is the position spelled on the run tape in that configuration. 

By a {\bf deterministic successor} of a legitimate configuration $x$ we mean the configuration $y$ such that $y$ immediately follows $x$ (in one transition) in the scenario where Environment does not move during the cycle described by $x$. For $n\geq 0$, the {\bf $n$th deterministic successor} of $x$ is defined inductively by stipulating that the $0$th deterministic successor of $x$ is $x$, and the $(n\plus 1)$th deterministic successor of $x$ is the deterministic successor of the $n$th deterministic successor of $x$. 

Let $E(\vec{s})$ be a   formula all of whose  free variables are among  $\vec{s}$ (but not necessarily vice versa), and $z$ be a variable not among $\vec{s}$. We 
 will write   $E^\circ(z,\vec{s})$ 
to denote an elementary formula whose free variables are $z$ and those of $E(\vec{s})$,  and which is a natural arithmetization of the predicate that,  for any constants $a,\vec{c}$ in the roles of $z,\vec{s}$, holds (that is, $E^\circ(a,\vec{c})$ is true) iff $a$ is a legitimate configuration and its yield is $E(\vec{c})$.
Further,  we 
 will write   $E^{\circ}_{\circ}(z,\vec{s})$ 
to denote an elementary formula whose free variables are $z$ and those of $E(\vec{s})$, and which  is a natural arithmetization of the predicate that,  for any constants $a,\vec{c}$  in the roles of $z,\vec{s}$, holds iff $E^\circ(a,\vec{c})\mlc E^\circ(b,\vec{c})$ is true, where $b$ is the $\chi(|a|)$th deterministic successor of $a$.

We say that a formula $E$  is {\bf critical}\label{icritical} iff one of the following conditions is satisfied:\vspace{-3pt} 
\begin{itemize}
\item $E$ is of the form $G_0\add G_1$ or $\ade y G$;\vspace{-5pt}
\item $E$ is of the form $\cla y G$ or $\cle y G$, and $G$ is critical;\vspace{-5pt}
\item $E$ is of the form  $G_0\mld G_1$, and both $G_0$ and $G_1$ are critical;\vspace{-5pt}
\item $E$ is of the form  $G_0\mlc G_1$, and at least one of $G_0,G_1$ is critical.
\end{itemize}

\begin{lemma}\label{august20b}
%\marginpar{august20b}
Assume $E(\vec{s})$ is a non-critical  formula all of whose  free variables are among  $\vec{s}$. 
Then
\[\pa\vdash \cla\bigl(E^{\circ}_{\circ}(z,\vec{s})  \mli \elz{\overline{E(\vec{s})}}\bigr).\] 
\end{lemma}

\begin{proof} Assume the conditions of the lemma. Argue in $\pa$. Consider arbitrary $(\cla$) values of $z$ and $\vec{s}$, which we continue writing as $z$ and $\vec{s}$. Suppose, for a contradiction, that $E^{\circ}_{\circ}(z,\vec{s})$ is true but $\elz{\overline{E(\vec{s})}}$ is false.   The falsity of $\elz{\overline{E(\vec{s})}}$ implies the falsity of $\elz{E(\vec{s})}$. This is so because the only difference between the two formulas is that, wherever the latter has some politeral $L$, the former has a $\mld$-disjunction containing $L$ as a disjunct.  

The truth of $E^{\circ}_{\circ}(z,\vec{s})$ implies that $\cal X$ reaches the configuration (computation step) $z$ and, in the scenario where Environment does not move, $\cal X$ does not move either for at least $\chi(|z|)$ steps afterwards. If $\cal X$ does not move even after $\chi(|z|)$ steps, then it has lost the game, because the eventual position hit by the latter is $E(\vec{s})$ and the elementarization of the latter is false (it is not hard to see that every such game is indeed lost). And if $\cal X$ does make a move sometime after $\chi(|z|)$ steps, then it violates its time complexity bound $\chi$, because the background of that move is smaller than  $|z|$ but the timecost is at least $\chi(|z|)$. Thus, in either case, $\cal X$ does not win $X$ in time $\chi$, that is,  
%\marginpar{jan4b}
\begin{equation}\label{jan4b}
\mbox{\em $\mathbb{L}$ is true.}
\end{equation}

Consider any non-critical   formula $G$.  By induction on the complexity of $G$, we are going to show that $\elz{\overline {G}}$ is true for any ($\cla$) values of its free variables. Indeed:

If $G$ is a literal, then  $\elz{\overline {G}}$ is $G\mld \mathbb{L}$ which, by (\ref{jan4b}), is true. 

If $G$ is $H_0\adc H_1$ or $\ada xH(x)$, then $\elz{\overline {G}}$ is $\twg$ and is thus true.

$G$ cannot be $H_0\add H_1$ or $\ade xH(x)$, because then it would be critical. 

If $G$ is $\cla yH(y)$ or $\cle y H(y)$, then $\elz{\overline {G}}$ is $\cla y\elz{\overline{H(y)}}$ or $\cle y\elz{\overline{H(y)}}$. In either case $\elz{\overline{G}}$ is true because, by the induction hypothesis, $\elz{\overline{H(y)}}$ is true for every value of its free variables, including variable $y$.

 If $G$ is $H_0\mlc H_1$, then both  $H_0$ and $H_1$ are non-critical. Hence, by the induction hypothesis, both $\elz{\overline{H_0}}$ and 
$\elz{\overline{H_1}}$ are  true. Hence so is  $\elz{\overline{H_0}}\mlc\elz{\overline{H_1}}$ which, in turn, is nothing but $\elz{\overline{G}}$. 

Finally, if  $G$ is $H_0\mld H_1$, then one of the formulas $H_i$ is non-critical. Hence, by the induction hypothesis, $\elz{\overline{H_i}}$ is  true. Hence so is  $\elz{\overline{H_0}}\mld\elz{\overline{H_1}}$ which, in turn, is nothing but $\elz{\overline{G}}$.

Thus, for any non-critical formula $G$, $\elz{\overline {G}}$ is true. This includes the case $G= E(\vec{s}) $  which, however,  contradicts our assumption that $\elz{E(\vec{s})}$ is false.\vspace{-7pt}  
\end{proof}

\begin{lemma}\label{august20a}
%\marginpar{august20a}
Assume $E(\vec{s})$ is a  critical   formula all of whose  free variables are among  $\vec{s}$. Then
\begin{equation}\label{m13a}
\arfour\vdash   \cle E^{\circ}_{\circ} (z,\vec{s})  \mli \cla\overline{E(\vec{s})}.
\end{equation} 
\end{lemma}

\begin{proof} Assume the conditions of the lemma. By induction on complexity, one can easily see that the elementarization of any critical formula is false. Thus, for whatever ($\cla$) values of $\vec{s}$, $\elz{E(\vec{s})}$ is false. Arguing further as we did in the proof of Lemma \ref{august20b} when deriving (\ref{jan4b}), we find that,  if  $E^{\circ}_{\circ} (z,\vec{s})$ is true for whatever ($\cle$) values of $z$ and $\vec{s}$,  then so is $\mathbb{L}$. And this argument can be formalized in $\pa$, so that we have 
\(\pa\vdash\cle E^{\circ}_{\circ} (z,\vec{s})\mli \mathbb{L}.\)
This, together with Lemma \ref{jan4d}, can be easily seen to imply (\ref{m13a}) by LC.\vspace{-7pt}   
\end{proof}

\subsection{Q.E.D.}
In this subsection we finish the extensional  completeness proof for $\arfour$. Well, almost finish. The point is that our argument relies on Lemma    \ref{m2a}   whose  proof  is postponed to Appendix \ref{sA}. Here we only present a brief intuitive explanation of the proof idea for it. A reader satisfied by  our explanation will have no reasons to go through the technical appendix given at the end of this paper, whose only purpose is to provide a relatively detailed proof for Lemma    \ref{m2a}.  

Let $E$ be a formula not containing the variable $y$. We say that a formula $H$ is a {\bf $(\oo,y)$-development} of  $E$ iff $H$ is the result of replacing in $E$: 
\begin{itemize}
\item either a surface occurrence of a subformula $F_0\adc F_1$ by $F_i$ ($i=0$ or $i=1$), 
\item  or  a surface occurrence of a subformula $\ada xF(x)$ by $F(y)$.   
\end{itemize}

{\bf $(\pp ,y)$-development} is defined in the same way, only with $\add,\ade$ instead of $\adc,\ada$. 

\begin{lemma}\label{m2a}
%\marginpar{m2a}
Assume $E(\vec{s})$ is a formula all of whose free variables are among $\vec{s}$, and $y$ is a variable not occurring in $E(\vec{s})$. Then:      

(a) For every $(\oo,y)$-development $H_{i}(y,\vec{s})$ of $E(\vec{s})$,  $\arfour$ proves $E^{\circ}_{\circ}  (z,\vec{s})  \mli \ade u H_{i}^{\circ}(u,y,\vec{s})\).

(b) Where $H_1(y,\vec{s}),\ldots,H_n(y,\vec{s})$ are all of the $(\pp,y)$-developments of $E(\vec{s})$, $\arfour$ proves
%\marginpar{m2e}
\begin{equation}\label{m2e}   E^{\circ}  (z,\vec{s})  \mli E^{\circ}_{\circ}  (z,\vec{s})\add \mathbb{L}\add \ade u\ade y H_{1}^{\circ}(u,y,\vec{s})\add\ldots\add\ade u\ade y H_{n}^{\circ}(u,y,\vec{s})  . \end{equation}
\end{lemma}

\begin{idea} $E^{\circ} (z,\vec{s})$ implies that $z$ is a configuration reached by $\cal X$ in some play, and the game by that time has been brought down to $E (\vec{s})$.  $E^{\circ}_{\circ} (z,\vec{s})$ additionally implies that this situation persists ``for a while'' after $z$. 

For clause (a), assume $E^{\circ}_{\circ} (z,\vec{s})$. For any  $(\oo,y)$-development $H_{i}(y,\vec{s})$ of $E(\vec{s})$ and any value  of $y$, $H_{i}(y,\vec{s})$ is the game to which $E (\vec{s})$ is brought down by a certain labmove $\oo\alpha$. To solve $\ade u H_{i}^{\circ}(u,y,\vec{s})$ --- i.e., make  $H_{i}^{\circ}(u,y,\vec{s})$ true --- we can choose   $u$  to be the result of appending such a labmove $\oo\alpha$ to the run tape content of the deterministic successor of configuration $z$.  After properly formalizing encoding for configurations, this argument can be reproduced in $\arfour$. 

For clause (b), assume $E^{\circ}(z,\vec{s})$. We can trace, within $\arfour$, the work of $\cal X$ for ``sufficiently many'' --- namely, $\chi(|z|)$  --- steps in the scenario where Environment does not move. If  $\cal X$ does not move during those $\chi(|z|)$ steps either, then $E^{\circ}_{\circ}  (z,\vec{s})$ is true and we can choose it in the consequent of (\ref{m2e}). Suppose now $\cal X$ makes a move $\alpha$ within $\chi(|z|)$ steps. If $\alpha$ is illegal, then $\mathbb{L}$ is true, and we choose the latter in the consequent of (\ref{m2e}). Otherwise, if $\alpha$ is a legal move, then it brings $E(\vec{s})$ down to one of (the instances of) its $(\pp,y)$-developments $H_{i}(y,\vec{s})$ for a certain value of $y$. We choose the corresponding $\add$-disjunct in the consequent of  (\ref{m2e}); further, by tracing the work of $\cal X$, we will be able to compute the value of the configuration $u$ in which $\cal X$ made the above move, as well the above-mentioned ``certain value of $y$''. By choosing these values for the variables $u$ and $y$ in  $\ade u\ade y H_{1}^{\circ}(u,y,\vec{s})$, we win the game. Again, after properly formalizing encoding for configurations, the entire argument can be reproduced in $\arfour$.\vspace{-7pt}   
\end{idea}

\begin{lemma}\label{m2c}
%\marginpar{m2a}
Assume $E(\vec{s})$ is a formula all of whose free variables are among $\vec{s}$. 
Then $\arfour$ proves $E^{\circ}  (z,\vec{s})  \mli \overline{E(\vec{s})}$.  
\end{lemma}

\begin{proof} We prove this lemma by (meta)induction on the complexity of $E(\vec{s})$. By the induction hypothesis, for any $(\oo,y)$- or $(\pp,y)$-development $H_i(y,\vec{s})$ of $E(\vec{s})$ (if there are any), $\arfour$ proves 
%\marginpar{m2d}
\begin{equation}\label{m2d}
H_{i}^{\circ}  (u,y,\vec{s})  \mli \overline{H_i(y,\vec{s})}.
\end{equation}

Argue in $\arfour$ to justify $E^{\circ}  (z,\vec{s})  \mli \overline{E(\vec{s})}$. Consider any values (constants) $b$ and $\vec{a}$ chosen by Environment for $z$ and $\vec{s}$, respectively.\footnote{Here, unlike the earlier followed practice, for safety, we are reluctant to use the names $z$ and $\vec{s}$ for those constants.}   Assume  $E^{\circ}  (b,\vec{a})$ is true (if not, we win). Our goal is to show how to win $\overline{E(\vec{a})}$.  From clause (b) of Lemma 
\ref{m2a},  the resource (\ref{m2e}) with $b,\vec{a}$ plugged for $z,\vec{s}$ is at our disposal. Since the  antecedent of the latter is true, the provider of (\ref{m2e}) will have to choose one of the $\add$-disjuncts in the consequent  
\begin{equation} \label{m2f}   E^{\circ}_{\circ}  (b,\vec{a})\add \mathbb{L}\add \ade u\ade y H_{1}^{\circ}(u,y,\vec{a})\add\ldots\add\ade u\ade y H_{n}^{\circ}(u,y,\vec{a})  . \end{equation}

{\em Case 1}: $\mathbb{L}$ is chosen in (\ref{m2f}). It has to be true, or else the provider loses. By Lemma \ref{jan4d}, we also have the resource $\mathbb{L}\mli \cla \overline{E(\vec{s})}$. Since its antecedent $\mathbb{L}$ is true, we know how to win the consequent $\cla \overline{E(\vec{s})}$. But a strategy that wins the latter, of course, also wins our target    $\overline{E(\vec{a})}$.   

{\em Case 2}: One of $\ade u\ade y H_{i}^{\circ}(u,y,\vec{a})$ is chosen in (\ref{m2f}). This should be followed by a further choice of some constants $c$ and $d$ for $u$ and $y$, yielding $H_{i}^{\circ}(c,d,\vec{a})$. Plugging $\vec{a}$, $c$ and $d$ for $\vec{s}$, $u$ and $y$ in (\ref{m2d}), we get $H_{i}^{\circ}  (c,d,\vec{a})  \mli \overline{H_i(d,\vec{a})}$. We may assume that the antecedent of the latter is true, or else the provider of (\ref{m2f}) lied when bringing the latter down to  $H_{i}^{\circ}(c,d,\vec{a})$. Thus, the consequental resource $\overline{H_i(d,\vec{a})}$ is at our disposal. But, remembering that the formula $H_i(y,\vec{s})$ is a $(\pp,y)$-development of the formula $E(\vec{s})$, we can now win $\overline{E(\vec{a})}$ by making a move $\alpha$  that brings the latter down to  $\overline{H_i(d,\vec{a})}$, which we already know how to win.   For example, imagine $E(\vec{s})$ is $Y(\vec{s})\mli Z(\vec{s})\add T(\vec{s})$ and $H_i(y,\vec{s})$ is $Y(\vec{s})\mli Z(\vec{s})$, so that   $\overline{E(\vec{a})}$ is $\overline{Y(\vec{a})}\mli \overline{Z(\vec{a})}\add \overline{T(\vec{a})}$ and $\overline{H_i(d,\vec{a})}$ is $\overline{Y(\vec{a})}\mli \overline{Z(\vec{a})}$. Then the above move $\alpha$ will be ``$1.0$''. It indeed brings $\overline{Y(\vec{a})}\mli \overline{Z(\vec{a})}\add \overline{T(\vec{a})}$ down to $\overline{Y(\vec{a})}\mli \overline{Z(\vec{a})}$. As another example, imagine $E(\vec{s})$ is $Y(\vec{s})\mli \ade w Z(w,\vec{s})$ and $H_i(y,\vec{s})$ is $Y(\vec{s})\mli Z(y,\vec{s})$, so that   $\overline{E(\vec{a})}$ is $\overline{Y(\vec{a})}\mli \ade w\overline{Z(w,\vec{a})}$ and $\overline{H_i(d,\vec{a})}$ is $\overline{Y(\vec{a})}\mli \overline{Z(d,\vec{a})}$. Then the above move $\alpha$ will be ``$1.d$''. It indeed brings $\overline{Y(\vec{a})}\mli \ade w\overline{Z(w,\vec{a})}$ down to $\overline{Y(\vec{a})}\mli \overline{Z(d,\vec{a})}$.

{\em Case 3}:  $E^{\circ}_{\circ}  (b,\vec{a})$ is chosen in (\ref{m2f}). It has to be true, or else the provider loses. 

{\em Subcase 3.1}: The formula $E(\vec{s})$ is critical. Since $E^{\circ}_{\circ}  (b,\vec{a})$ is true, so is $\cle E^{\circ}_{\circ}  (z,\vec{s})$. By  Lemma \ref{august20a}, we also have $\cle E^{\circ}_{\circ}  (z,\vec{s})\mli  \cla \overline{E(\vec{s})}$. So, we have a strategy that wins $ \cla \overline{E(\vec{s})}$. Of course, the same strategy also wins $\overline{E(\vec{a})}$. 

{\em Subcase 3.2}:  The formula $E(\vec{s})$ is not critical. By Lemma \ref{august20b}, we find that the elementarization of $\overline{E(\vec{a})}$ is true. This obviously means that if Environment does not move in  $\overline{E(\vec{a})}$, we win the latter. So, assume Environment makes a move $\alpha$ in $\overline{E(\vec{a})}$. 
The move should be legal, or else we  win. Of course, the same move is a legal move of $ E(\vec{a})$ and, for one of the $(\oo,y)$-developments $H_i(y,\vec{s})$ of the formula $E(\vec{s})$ and some constant $c$, it brings $ E(\vec{a})$  down to $H_i(c,\vec{a})$ as well as $ \overline{E(\vec{a})}$  down to $\overline{H_i(c,\vec{a})}$.  For example, if  $E(\vec{s})$ is $Y(\vec{s})\mli Z(\vec{s})\adc T(\vec{s})$, $\alpha$ could be the move ``$1.0$'', which brings $Y(\vec{a})\mli Z(\vec{a})\adc T(\vec{a})$ down to $Y(\vec{a})\mli Z(\vec{a})$; the formula $Y(\vec{s})\mli Z(\vec{s})$ is indeed a $(\oo,y)$-development of the formula $Y(\vec{s})\mli Z(\vec{s})\adc T(\vec{s})$. As another example, imagine $E(\vec{s})$ is $Y(\vec{s})\mli \ada w Z(w,\vec{s})$. Then the above move $\alpha$ could be ``$1.c$'',  which brings $Y(\vec{a})\mli \ada wZ(w,\vec{a})$ down to $Y(\vec{a})\mli Z(c,\vec{a})$; the formula $Y(\vec{s})\mli Z(y,\vec{s})$ is indeed a $(\oo,y)$-development of the formula $Y(\vec{s})\mli \ada wZ(w,\vec{s})$. Fix the above formula $H_i(y,\vec{s})$ and constant $c$. Choosing $b$, $\vec{a}$ and $c$ for $z$, $\vec{s}$ and $y$ in the resource provided 
by clause (a) of Lemma 
\ref{m2a}, we get the resource $E^{\circ}_{\circ}  (b,\vec{a})  \mli \ade u H_{i}^{\circ}(u,c,\vec{a})\). Since the antecedent   of the latter is true by our assumptions,  the consequent $\ade u H_{i}^{\circ}(u,c,\vec{a})$ is at our disposal. The provider will have to choose a constant $d$ for $u$ in it such that $H_{i}^{\circ}(d,c,\vec{a})$ is true. Hence, by choosing $d$, $c$ and $\vec{a}$ for $u$, $y$ and $\vec{s}$ in (\ref{m2d}), we get the resource $\overline{H_i(c,\vec{a})}$. That is, we have a strategy for the game $\overline{H_i(c,\vec{a})}$ to which $\overline{E(\vec{a})}$ has evolved after Environment's move $\alpha$. We switch to that strategy and win. 
\end{proof}
 
Now we are ready to claim the target result of this section.  Let $a$ be the code of the start configuration of $\cal X$, and $\hat{a}$ be a standard variable-free term representing $a$, such as $0$ followed by $a$  ``$\successor$''s.  Of course, $\pa$ and hence $\arfour$ proves $X^\circ(\hat{a})$. By Fact \ref{com}, $\arfour$ proves $\ade z(z\equals\hat{a})$. 
 By Lemma \ref{m2c},   $\arfour$ also proves $\ada z\bigl(X^{\circ}  (z)  \mli \overline{X}\bigr)$. These three can be seen to imply  $\overline{X}$ by LC. Thus, $\arfour\vdash\overline{X}$, as desired.

\section{Inherent extensional incompleteness in the general case}\label{sincom}
%\marginpar{sincom}
The extensional completeness of $\arfour$ is not a result that could be taken for granted. In this short section we argue  that, if one replaces {\em polynomial time computability}  by simply {\em computability} in our semantical treatment of  $\arfour$,  extensional completeness is impossible to achieve for whatever recursively axiomatizable sound extension of $\arfour$ or other systems of clarithmetic.

Our extensional incompleteness argument goes like this. Consider any system $\bf S$ in the style of $\arfour$ whose  proof predicate --- throughout this section understood in the ``superextended'' sense of Section \ref{ss11} --- is decidable and hence the theoremhood predicate is recursively enumerable.   Assume $\bf S$ is sound in the same strong sense as $\arfour$ --- that is, there is an effective  procedure that extracts an algorithmic solution (HPM) for the problem represented by any sentence $F$ from any  $\bf S$-proof of $F$. 
 
Let then $A(s)$ be the predicate which is true iff:
\begin{itemize}
\item $s$ is (the code of) an $\bf S$-proof of some sentence of the form $\ada x\bigl(\gneg E(x)\add E(x)\bigr)$, where $E$ is elementary,  
\item and $E(s)$  is false.
\end{itemize}   

On our assumption of the soundness of $\bf S$, $A(s)$ is a decidable predicate. Namely, it  is decided by a procedure that first checks if $s$ is the code of  an  $\bf S$-proof of some sentence of the form $\ada x\bigl(\gneg E(x)\add E(x)\bigr)$, where $E$ is elementary. If not, it rejects. If yes, the procedure extracts from $s$ an HPM $\cal H$ which solves $\ada x\bigl(\gneg E(x)\add E(x)\bigr)$, and then simulates the play of $\cal H$ in the scenario where,  at the very beginning of the play, Environments makes the move $s$, thus bringing the game down to $\gneg E(s)\add E(s)$. If, in this play, $\cal H$ responds by choosing $\gneg E(s)$, then the procedure accepts $s$; and if $\cal H$ responds by choosing $E(s)$, then the procedure rejects $s$. Obviously this procedure indeed decides the predicate $A$.

 Now, assume that $\bf S$ is extensionally complete. Since $A$ is decidable, the problem $\ada x\bigl(\gneg A(x)\add A(x)\bigr)$ has an algorithmic solution. So, for some sentence $F$ with $F^\dagger=\ada x\bigl(\gneg A(x)\add A(x)\bigr)$ and some $c$, we should have that $c$ is (the code of) an   {\bf S}-proof of $F$. Obviously $F$ should have the form $\ada x\bigl(\gneg E(x)\add E(x)\bigr)$, where $E(x)$ is an elementary formula with $E^\dagger(x)=A(x)$. We are now dealing with the absurd of $A(c)$ being true iff it is false.

\section{On the intensional strength of $\arfour$}\label{sculprit}
%\marginpar{sculprit}

\begin{theorem}\label{feb15}
%\marginpar{feb15}
Let $X$ and $\mathbb{L}$ be as in Section \ref{s19}.  Then $\arfour\vdash \gneg \mathbb{L}\mli X$.
\end{theorem}

\begin{proof} Let $X$ and $\mathbb{L}$ be as in Section \ref{s19}, and so be the meaning of the overline notation. First, by induction on the complexity of $E$, we  want to show that
%\marginpar{m14a}
\begin{equation}\label{m14a}
\mbox{\em For any formula $E$, $\arfour\vdash \cla (\overline{E}\mlc \gneg \mathbb{L}\mli E)$.} 
\end{equation}
If  $E$ is a literal, then $\cla(\overline{X}\mlc \gneg \mathbb{L}\mli E)$ is nothing but 
$\cla\bigl((E\mld \mathbb{L})\mlc \gneg \mathbb{L}\mli E\bigr)$. This is a classically valid elementary sentence, and hence it is provable in $\arfour$ (by LC from the empty set of premises). 
Next, suppose $E$ is $F_0\mlc F_1$. By the induction hypothesis, $\arfour$ proves both $\cla (\overline{F_0}\mlc \gneg \mathbb{L}\mli F_0)$ and  $\cla (\overline{F_1}\mlc \gneg \mathbb{L}\mli F_1)$. These two, by LC, imply $\cla \bigl((\overline{F_0}\mlc\overline{F_1})\mlc \gneg \mathbb{L}\mli F_0\mlc F_1\bigr)$. And the latter is nothing but the desired $\cla (\overline{E}\mlc \gneg \mathbb{L}\mli E)$. 
The remaining cases where $E$ is $F_0\mld F_1$, $F_0\adc F_1$, $F_0\add F_1$, $\ada xF(x)$, $\ade xF(x)$, $\cla xF(x)$ or $\cle xF(x)$ are handled in a similar way.
 (\ref{m14a}) is thus proven.

(\ref{m14a}) implies that $\arfour$ proves $\overline{X}\mlc \gneg \mathbb{L}\mli X$.  As established in Section \ref{s19}, $\arfour$ also proves $\overline{X}$. From these two, by LC, $\arfour$ proves  $\gneg \mathbb{L}\mli X$, as desired. 
\end{proof} 

Remember that, in Section \ref{s19},  $X$ was an arbitrary $\arfour$-sentence assumed to have a polynomial time solution  under the standard interpretation $^\dagger$. And $\gneg\mathbb{L}$ was a certain true sentence of the language of classical Peano arithmetic. We showed in that section that $\arfour$ proved a certain sentence $\overline{X}$ with $\overline{X}^\dagger=X^\dagger$. That is, we showed that $X$ was ``extensionally provable''. 
According to our present Theorem \ref{feb15}, in order to make $X$ also provable in the intensional sense, all we need is to add to the axioms of $\arfour$ the true elementary sentence $\gneg\mathbb{L}$. 

In philosophical terms, the import of Theorem \ref{feb15} is that the culprit of the intensional incompleteness of $\arfour$ is the (G\"{o}del's) incompleteness of its classical, elementary part. Otherwise, the ``nonelementary rest'' of $\arfour$ --- the two extra-Peano axioms and the $\arfour$-Induction  rule ---  as a bridge from classical arithmetic to polynomial-time-computability-oriented clarithmetic, is 
complete in a certain very strong and natural sense. Namely, it guarantees not only extensional but also intensional provability of every polynomial time computable problem   as long as all necessary  true elementary sentences are taken care of. This means that if, instead of $\pa$,  we take the truth arithmetic {\bf Th(N)} (the set of all true sentences of the language of $\pa$) as the base arithmetical theory,   the corresponding version of $\arfour$ will be not only extensionally, but also intensionally complete. Unfortunately, however, such a system will no longer be recursively axiomatizable.

To summarize, in order to make $\arfour$ intensionally stronger, it would be sufficient to add to it new true elementary (classical) sentences only.  Note that this sort of an extension, even if in a language more expressive than that of $\pa$, would automatically remain sound and extensionally complete: virtually nothing in this paper relies on the fact that $\pa$ is not stronger than it really is. Thus, basing applied theories on CoL allows us to construct ever more expressive and intensionally  strong  theories without worrying about how to preserve soundness and extensional completeness. Among the main goals of this paper was to illustrate the scalability of CoL rather than the virtues of the particular system $\arfour$ based on it. The latter is in a sense arbitrary, as is $\pa$ itself: in the role of the classical part of $\arfour$, we could have chosen not only any true extension of $\pa$, but certain weaker-than-$\pa$ theories as well, for our proof of the extensional completeness of $\arfour$ does not require the full strength of $\pa$. The reason for not having done so is purely ``pedagogical'': $\pa$ is the simplest and best known arithmetical theory, and reasoning in it is much more relaxed, easy and safe than in weaker versions. $\arfour$ is thus the simplest and nicest representative of the wide class of clarithmetical theories for polynomial time computability, all enjoying the same relevant properties  as $\arfour$ does.  

As pointed out in Section \ref{intr}, among the potential applications of $\arfour$-style systems is using them as formal tools (say, after developing reasonable theorem-provers) for systematically finding efficient solutions for problems, and the stronger such a system is, the better the chances that a solution for a given problem will be found. Of course, what matters in this context is intensional rather than extensional strength. So, perfect strength is not achievable, but we can keep moving ever closer to it.  

One may ask why not think of simply using $\pa$ (or even, say, {\bf ZFC}) instead of $\arfour$ for the same purposes: after all,   $\pa$ is strong enough to allow us reason about polynomial time computability. This is true,  but $\pa$ is far from being a reasonable alternative to $\arfour$. First of all, 
 as a tool for finding  solutions, $\pa$ is very indirect and hence hopelessly inefficient. Pick any of the basic functions of Section \ref{s17} and try to generate, in $\pa$, a full formal proof of the fact that the function is polynomial-time computable (or even just {\em express} this fact) to understand the difference. Such a proof would have to proceed by clumsy reasoning about {\em non-number} objects such as Turing machines and computations, which, only by good luck, happen to be amenable to being understood as numbers through encoding. In contrast, reasoning in $\arfour$ is directly about numbers and their properties, without having to encode any foreign (metaarithmetical or complexity-theoretical) beasts and then try to reason about them as if they were just kind and innocent natural numbers. Secondly, even if an unimaginably strong theorem-prover succeeded in finding such a proof, there would be no direct use of it because, from a proof of the existence of a solution we cannot directly extract a solution. Furthermore, even knowing that a given HPM $\cal X$ solves the problem in {\em some} polynomial time $\chi$, would have no practical significance without knowing {\em what} particular polynomial $\chi$ is, in order to assess whether it is reasonable for our purposes, or takes us beyond the number of nanoseconds in the lifespan of the universe (after all, $\ell^{999^{999}}$ is also a polynomial function!). In order to actually obtain a solution and a polynomial bound for it,   one would need a {\bf constructive} proof, that is, not just a proof that a polynomial function $\chi$ and a $\chi$-time solution exist, but a proof of the fact that certain particular numbers $a$ and $b$ are (the codes of) a polynomial term $\chi$ and a $\chi$-time solution $\cal X$. Otherwise, a theorem-prover would have to be used not just once for a single target formula, but an indefinite (intractably many) number of times, once per each possible pair of values of $a,b$ until the ``right'' values are encountered.  To summarize, $\pa$ does not provide any reasonable mechanism for handling queries in the style ``{\em find} a polynomial time solution for problem $A$'': in its standard form, $\pa$ is merely a YES/NO kind of a ``device''.   

The above dark picture can be somewhat brightened by switching from $\pa$ to Heyting's arithmetic {\bf HA} --- the version of $\pa$  based on intuitionistic logic instead of classical logic,  which is known to  allow us to directly extract, from a proof of a formula $\cle  xF(x)$, a particular value of $x$ for which $F(x)$ is true. But the question is why intuitionistic logic and not CoL? Both claim to be ``constructive logics'', but the constructivistic claims of CoL have a clear semantical meaning and justification, while intuitionistic logic is essentially an ad hoc invention whose constructivistic claims are mainly based on certain syntactic and hence circular considerations,\footnote{What creates circularity is the common-sense fact that syntax is merely to serve a meaningful semantics, rather than vice versa. It is hard   not to remember the following words from \cite{Japfin} here: ``The reason for the failure of $P\add\gneg P$ in CoL is not that this principle \ldots is not included in its axioms. Rather, the failure of this principle is exactly the reason why this principle, or anything else entailing it, would not be among the axioms of a sound system for CoL''.}   without being supported by a convincing and complete {\em formal} constructive semantics.   And, while {\bf HA} is immune to the second one of the two problems pointed out in the previous paragraph,  it  still suffers from the first problem. At the same time, as a reasoning tool, {\bf HA} is inferior to $\pa$, for it is intensionally weaker and, from the point of view of the philosophy of CoL, is so for no good reasons. As a simple example, consider the function $f$ defined by ``$f(x)=x$ if $\pa$ is either consistent or inconsistent, and $f(x)=2x$ otherwise''. This is a legitimately defined function, and we all --- just as $\pa$ --- know that extensionally it is the same as the identity function $f(x)=x$. Yet, {\bf HA} can be seen to fail to prove --- in the intensional sense --- its computability.  Despite its name, intuitionistic logic is not so ``intuitive'' after all!

A natural question to ask is: {\em Is there a sentence $X$ of the language of $\arfour$ whose polynomial time computability is constructively provable in $\pa$ yet $X$ is not provable in $\arfour$?} Remember that, as we agreed just a while ago, by {\bf constructive provability} of the polynomial time computability of $X$ in $\pa$ we mean that, for some particular HPM $\cal X$ and a particular polynomial (term) $\chi$, $\pa$ proves that $\cal X$ is   a $\chi$-time solution of $X$. If the answer to this question was positive, then $\pa$, while indirect and inefficient, would still have at least {\em something} to say   in its defense when competing with $\arfour$ as a problem-solving tool. But, as seen from the following theorem, the answer to the question is negative:

\begin{theorem}\label{jan30}
%\marginpar{jan30}
Let $X$ be any sentence of the language of $\arfour$ such that $\pa$ constructively proves (in the above sense) the polynomial time computability of $X$. Then $\arfour\vdash X$. 
\end{theorem}

\begin{proof} Consider   any sentence $X$ of the language of $\arfour$. Assume $\pa$ constructively proves the polynomial time computability of $X$, meaning that,
for a certain HPM $\cal X$ and a certain term $\chi$,  
$\pa$ proves that $\cal X$ solves $X$ in    time $\chi$. But this is exactly what the sentence $\mathbb{L}$ of Section  \ref{s19} denies. So, $\pa\vdash\gneg \mathbb{L}$. But, by  
Theorem \ref{feb15}, we also have $\arfour\vdash\gneg \mathbb{L}\mli X$. Consequently, $\arfour\vdash X$.  
\end{proof} 

An import of the above theorem is that, if we tried to add to $\arfour$ some new nonelementary axioms in order to achieve a properly greater intensional strength, the fact that such axioms are computable in time $\chi$ for some particular polynomial $\chi$ would have to be unprovable in $\pa$, and hence would have to be ``very nontrivial''.  The same applies to attempts to extend $\arfour$ through some new rules of inference.  

\section{Give Caesar what belongs to Caesar}\label{scesar} Beginning from Buss's seminal work \cite{Buss}, many complexity-sensitive or complexity-oriented arithmetical and logical systems have been developed by various authors  (\cite{bbb2,Bussint, bbb3, bbb4, bbb5, bbb6, Sch} and more). Most of those achieve control over complexity in ways very different from ours, such as by type information rather than by explicit bounds on quantifiers, for which reason we do not attempt any direct comparison here. As pointed out in Section \ref{intr}, Buss's  \cite{Buss} original work on bounded arithmetic is closest to --- and the most immediate precursor of --- our approach. In fact,  in a broad sense,  $\arfour$ {\em is} a system of  bounded arithmetic, only based on CoL instead of classical logic or intuitionistic logic on which the other systems of bounded arithmetic have been traditionally based.

The main relevant results in the studies of classical-logic-based bounded arithmetic, extensive surveys of which can be found in \cite{Hajek,bbb7}, can be summarized saying that, by appropriately weakening the induction axiom of $\pa$ and  restricting it to bounded formulas of certain forms, and correspondingly readjusting the nonlogical vocabulary and axioms of $\pa$, certain  soundness and completeness   for the resulting system(s) $\bf S$ can be achieved. Such soundness results typically read like   ``If $\bf S$ proves a formula of the form $\cla x\cle y F(x,y)$, where $F$ satisfies such and such constraints, then there is function of such and such computational complexity which, for each $a$, returns a $b$ with $F(a,b)$''. And completeness results typically read  like ``For any function $f$ of such and such computational complexity, there is an $\bf S$-provable formula of the form  $\cla x\cle y F(x,y)$ such that, for any $a$ and $b$, \  $F(a,b)$ is true iff $b=f(a)$''.   
  
Among the characteristics that  make our approach very different from the above (as well as any other complexity-oriented systems of arithmetic known to the author), one should point out that it {\em extends} rather than {\em restricts} the language and the deductive power of $\pa$. Restricting  $\pa$ can be seen as   throwing out the baby with the bath water. Not only does it expel from the system many complexity-theoretically unsound yet otherwise meaningful and useful theorems, but it also reduces --- even if only in the intensional rather than extensional sense --- the class of complexity-theoretically  correct provable principles. This is a necessary sacrifice, related to the inability of the underlying classical logic  to clearly differentiate between constructive ($\adc,\add,\ada,\ade$) and ``ordinary'', non-constructive versions ($\mlc,\mld,\cla,\cle$) of operators. The inadequacy of classical logic as a basis for constructive systems also shows itself in the fact that  the above-mentioned soundness and completeness results are only partial.

The above problem of partiality is partially overcome when one bases a complexity-oriented arithmetic on intuitionistic logic (\cite{Bussint,Sch}) instead of classical logic.  In this case, soundness/completeness extends to all formulas of the form $\cla x\cle y F(x,y)$, without the ``$F$ satisfies such and such constraints'' condition (the reason why we still consider this sort of soundness/completeness partial is that it remains   limited to formulas of the form $\cla x\cle y F(x,y)$, i.e. functions, which, for us, are only special cases of computational problems). However, for reasons pointed out in the previous section, switching to intuitionistic logic signifies throwing out even more of the ``baby'' from the bath tub, further decreasing the intensional strength of the theory and probably losing its intuitive clarity or appeal in the eyes of the classically-minded majority. 

Both classical-logic-based and intuitionistic-logic-based  systems of bounded arithmetic happen to be {\em inherently weak} theories, as opposed to our CoL-based version, which is as strong as G\"{o}del's incompleteness phenomenon permits, and which can be indefinitely strengthened without losing computational soundness. 
We owe this achievement to the fact that  CoL gives Caesar  what belongs to Caesar, and  God  what belongs to God. As we had a chance to see throughout this paper, classical ($\mlc,\mld,\cla,\cle$) and constructive ($\adc,\add,\ada,\ade$) logical constructs can peacefully coexist and complement each other in one natural system that seamlessly extends the classical, constructive, resource- and complexity-conscious visions and concepts, and does so not by mechanically putting things together, but rather on the basis of one natural, all-unifying, complete game semantics. Unlike most other approaches where only few, special-form expressions (if any) have clear computational interpretations, in our case every formula is a meaningful computational problem. Further, we can capture not only computational problems in the traditional sense, but also problems in the more general --- interactive --- sense.

Classical logic and classical arithmetic, so close (unlike, say, intuitionistic logic or {\bf HA}) to the heart and mind of all of us,  do not at all need to be rejected or tampered with in order to achieve constructive heights. Just the opposite, they can be put in faithful and useful service to this noble goal. Our heavy reliance on reasoning in $\pa$ throughout this paper is an eloquent illustration of it. Overall, the present work can be seen as an illustration of the  fruitfulness  of two independently conceived lines of thought ---  bounded arithmetic and computability logic --- through a successful marriage between them.  

The forthcoming paper \cite{cla5} constructs three new, incrementally strong systems of clarithmetic, named {\bf CLA5}, {\bf CLA6} and {\bf CLA7}. In the same sense as $\arfour$ is sound and complete with respect to polynomial time computability, these systems are shown to be sound and complete with respect to polynomial space computability, elementary recursive time computability and primitive recursive time computability, respectively (as for elementary recursive space and primitive recursive space, they simply coincide with elementary recursive time and primitive recursive time). The simplicity and elegance of those systems serves as additional empirical evidence for the naturalness of basing applied theories on CoL instead of the more traditional alternatives, and for the flexibility and scalability of our approach. All three systems, on top of Axioms 1-7, have Axiom 8 as the only extra-Peano axiom (Axiom 9 simply becomes derivable and hence redundant due to the presence of a stronger induction rule). $\cltw$ continues to serve as the logical basis for these systems, and what varies is only the induction rule. The induction rule of {\bf CLA5} differs from that of $\arfour$ in that, while the (two) inductive steps of the latter are based on binary successors, the (single) inductive step of the former is based on unary successor, i.e., is the kind old $F(x)\mli F(x\successor)$, with $F(x)$ still required to be a polynomially bounded formula. The system {\bf CLA6} is obtained from {\bf CLA5} by relaxing this requirement in the induction rule and, instead, requiring that $F(x)$ be exponentially bounded. And the system {\bf CLA7} is obtained from {\bf CLA6} by removing all conditions on $F(x)$ whatsoever, thus leaving the realm of bounded arithmetic. 

The earlier mentioned system {\bf CLA1} of \cite{Japtowards} further strengthens the above series. Its logical basis, just like that of all clarithmetical theories we have seen, is $\cltw$. And the nonlogical axioms, just as in  the case of {\bf CLA5}, {\bf CLA6} and {\bf CLA7}, are Axioms 1-8 of Section \ref{ss11}.  As we may guess, the only difference between {\bf CLA1} and the weaker systems {\bf CLA4}-{\bf CLA7} is (again) related to how the induction rule operates. Here the difference is of a qualitative character due to the fact that {\bf CLA1}, unlike {\bf CLA4}-{\bf CLA7}, is a {\em natural deduction} system. Namely, while an inductive step of {\bf CLA7} is a {\em formula} $F(x)\mli F(x\successor)$, the corresponding inductive step in {\bf CLA1} is a {\em derivation} of $F(x\successor)$ from $F(x)$. In classical systems, according to the deduction theorem, a formula $E$ is derivable from a formula $G$ iff the formula $G\mli F$ is provable, so switching to natural deduction in the style of {\bf CLA1} would create no difference. The situation, however, is very different in (the resource-conscious) CoL-based systems, where deriving $E$ from $G$ is generally easier than proving $G\mli E$. This is so because a derivation may ``recycle'' its premises while, on the other hand, the antecedent of a $\mli$-combination  may be ``unrecyclable''. For instance, $E\mlc E$ is always derivable from $E$ but, as we had a chance to see from Exercise \ref{feb1a},  $E\mli E\mlc E$ is not always provable (and/or valid). While derivability of $F(x\successor)$ from $F(x)$ thus does not generally imply provability of $F(x)\mli F(x\successor)$, the latter {\em does} always imply the former. Consequently, {\bf CLA1} is at least as strong as {\bf CLA7}, meaning that {\bf CLA1}, just like {\bf CLA7}, can extensionally prove all primitive recursive time (and/or space) computable problems. A natural expectation here is that, at the same time, {\bf CLA1} takes us ``far beyond'' primitive recursive time (and/or space) computability, even though exactly {\em how far} still remains to be understood.  

\appendix

\section{Appendix}\label{sA}
%\marginpar{sA}
Throughout the rest of this appendix, the sole purpose of which is to prove Lemma \ref{m2a}, $X$, $\cal X$ and $\chi$ are  as in  Section \ref{s19}. The terms ``configuration'', ``state'', ``tape symbol'' etc. exclusively refer to ones of $\cal X$. We assume that $0$ and $1$ are among the tape symbols. $\mbox{\scriptsize {\bf blank}}$ will stand for the blank tape symbol. 

The proofs given in this appendix will heavily and repeatedly rely on $\pa$ and the results of Section \ref{s17}. It is important to note that, almost always,  this reliance will be only {\em implicit}.

\subsection{Encoding configurations}\label{sA1}
%\marginpar{sA1}
In order to prove Lemma \ref{m2a}, we need to introduce a system of encoding for various objects of relevance. Whenever $O$ is such an object, $\code{O}$ will stand for its code. 

 Let   $\cal A$ be the set consisting of all states (of $\cal X$), and four versions $\hat{a}$, $\check{a}$, $\underline{\hat{a}}$, $\underline{\check{a}}$ of every tape symbol $a$.   As opposed to {\em tape symbols}, we refer to the elements of $\cal A$ (simply) as {\bf symbols}. As we are going to see shortly, in our encoding of configurations,  $\hat{a}$ (resp. $\check{a}$) means the tape symbol $a$ written on the work (resp. run) tape, and the presence (resp. absence) of an underline  indicates that the head of the corresponding tape is (resp. is not) currently looking at the cell containing $a$. 

We extend the \ $\hat{ }$, $\check{ }$ notation from tape symbols to strings of tape symbols. Namely, for any such string  $\alpha$,  $\hat{\alpha}$ means the result of replacing every symbol $a$ by $\hat{a}$ in  $\alpha$. Similarly for $\check{\alpha}$. 

We pick and fix a sufficiently large integer $\mathfrak{k}$  and, with $\mathfrak{K}$ standing for $2^{\mathfrak{k}}$ throughout the rest of this paper, encode each symbol $a$ as a natural number $\code{a}$ with $|\code{a}|=\mathfrak{K}$.    As practiced earlier, terminologically and notationally we identify such an  $\code{a}$ with the corresponding binary numeral. Thus, the codes of symbols are bit strings, all of length $\mathfrak{K}$ and none starting with a $0$. Needless to mention that different symbols are required to have different codes.
 
Further, where $a_1,\ldots,a_k$ is a sequence of symbols, we encode it as the binary numeral --- again, identified with the corresponding number --- $\code{a_1}\ldots\code{a_k}$. We will not always be careful about differentiating objects from their codes, and may say something like ``the symbol $b$'' where, strictly speaking, the code of that symbol is meant, or vice versa. 

We need to make clear what, exactly, is meant by a configuration. According to our earlier informal explanation, this is a full description of some ``current'' situation in $\cal X$, namely, a list indicating the state of $\cal X$, the locations of its two scanning heads, and the contents of its two tapes. The tapes, however, are infinite, and we need to agree on how to represent their contents by finite means. Remember our convention that a head of an HPM can never move past the leftmost blank cell of the corresponding tape, and that the work-tape head can never write the blank symbol. This means that every cell to the left of a blank cell will also be blank and, accordingly, when describing a configuration, it would be sufficient to describe the contents of its tapes up to (including) the leftmost blank  cells.   Precisely, we agree to understand each configuration $C$ as the following sequence of symbols:
%\marginpar{m66}
\begin{equation}\label{m66}
a,\ \hat{b}_{0},\ \ldots\ \hat{b}_{i\iminus 1},\ \underline{\hat{b}}_{i},\ \hat{b}_{i\iplus 1},\ \ldots,\ \hat{b}_{m},\  \ \check{c}_{0},\ 
\ldots\ \check{c}_{j\iminus 1},\ \underline{\check{c}}_{j},\ \check{c}_{j\iplus 1},\ \ldots\ \check{c}_{n}
\end{equation}  
where, in the context of $C$,  $a$ is the (``current'') state of $\cal X$,   $b_{0}\ldots b_m$ (resp. $c_0\ldots c_n$) are the contents of cells $\# 0$ through $\# m$ 
(resp. $\# n$)  of the work (resp. run) tape, and $i$ (resp. $j$) is the cell $\#$ of the cell scanned by the head of the work (resp. run) tape. In addition,  both $b_{m}=c_{n}=\mbox{\scriptsize {\bf blank}}$ while no other $b_k$ or $c_k$ is \mbox{\scriptsize {\bf blank}}. We encode the above configuration as any other sequence of symbols, i.e., as the binary numeral 
\[\code{a}\ \code{\hat{b}_{0}}\ \ldots\ \code{\hat{b}_{i\iminus 1}}\ \code{\underline{\hat{b}}_{i}}\ \code{\hat{b}_{i\iplus 1}}\ \ldots\ \code{\hat{b}_{m}}\  \ \code{\check{c}_{0}}\ 
\ldots\ \code{\check{c} _{j\iminus 1}}\ \code{\underline{\check{c}}_{j}}\ \code{\check{c}_{j\iplus 1}}\ \ldots\ \code{\check{c}_{n}}.\]

In the sequel, we will be using the pseudoterm $x\circ y$ and  several  elementary formulas with special names, each one being a natural arithmetization of the corresponding predicate shown below:    
\begin{itemize}
\item $x\circ y$ abbreviates $x\mult 2^{|y|}\plus y$. Note that, when $x$ and $y$ are the codes of some sequences $a_1,\ldots,a_m$ and $b_1,\ldots,b_n$ of symbols, $x\circ y$ is the code of the {\bf concatenation} $a_1,\ldots,a_m,b_1,\ldots,b_n$ of those sequences.\vspace{-4pt}  
\item $\mathbb{N}(x,y)$ says ``if $x$ is $\code{\hat{b}_1,\ldots,\hat{b}_k}$ for some symbols $b_1,\ldots,b_k$ ($k\geq 0$), then $y$ is $\code{\check{b}_1,\ldots,\check{b}_k}$ .'' So, for instance, $\mathbb{D}(\code{\hat{1},\hat{.},\hat{0}},\code{\check{1},\check{.},\check{0}})$ is true.\vspace{-4pt} 
\item $\mathbb{C}(x)$ says ``$x$ is the code of a configuration''.\vspace{-4pt}
\item  $\mathbb{I}(x,y)$ says ``$x$ is the code of a configuration of the form (\ref{m66}), and $i\equals y$''.\vspace{-4pt} 
\item $\mathbb{J}(x,y)$ says ``$x$ is  the code of a configuration of the form (\ref{m66}), and $j\equals y$''.\vspace{-4pt}
\item $\mathbb{M}(x,y)$ says ``$x$ is the code of a configuration of the form (\ref{m66}), and $m\equals y$''.\vspace{-4pt}   
\item $\mathbb{E}(x,y)$ says ``$y$ is the code of the sequence of symbols resulting from changing every $0$ to $\check{0}$ and every $1$ to $\check{1}$ in the binary numeral representing number  $x$''. So, for instance, $\mathbb{E}(101,\code{\check{1},\check{0},\check{1}})$ is true.\vspace{-4pt}
\item $\mathbb{D}(x,y)$ says ``$x$ is $\code{\hat{b}_1,\ldots,\hat{b}_k}$ for some bits  $b_1,\ldots,b_k$ ($k\geq 0$) where $b_1$ (if present) is $1$, and $y$ is the number represented  by the numeral $b_1\ldots b_k$''. So, for instance, $\mathbb{D}(\code{\hat{1},\hat{0},\hat{1}},101)$ is true.\vspace{-4pt}
\item $\mathbb{S}(x,y)$ says ``$x$ is the code of a  configuration and $y$ is the code of the deterministic successor (see Section \ref{sds}) of that configuration''.\vspace{-4pt}
\item $\mathbb{A}(z,x,y)$ says ``$z$ is the code of a legitimate configuration $C$ (see Section \ref{sds}), $x$ is the code of the $y$th deterministic successor of $C$, and, for any $i$ with $0\mleq i\mleq y$, the state of the $i$th deterministic successor of $C$ is not a move state''. $\mathbb{A}(z,x,y)$ thus asserts that, after the (legitimate) configuration $z$, if Environment does not move, $\cal X$ reaches the configuration $x$ within $y$ ($y\mgeq 0$) steps, and it does not move during those steps, either.\vspace{-4pt}  
\item $\mathbb{A}'(z,y)$ abbreviates $\cle x \mathbb{A}(z,x,y)$. $\mathbb{A}'(z,y)$ thus says that, after reaching the (legitimate) configuration $z$, during the subsequent $y$ (including $0$) steps, if Environment does not move, neither does $\cal X$.\vspace{-4pt}   
\item $\mathbb{B}(z,x)$ says ``$z$ is the code of a legitimate configuration $C$, $x$ is the code of the $i$th deterministic successor of $C$ for some $i\mgeq 0$ and, for each   
$j$ with $0\mleq j\mleq i$, the state of the $j$th deterministic successor of $C$ is a move state iff $j\equals i$''.   $\mathbb{B}(z,x)$  thus asserts that $z$ is a legitimate configuration and $x$ is the earliest configuration after (and including) $z$ in which $\cal X$ moves in the scenario where Environment does not move.  
\end{itemize}

\begin{lemma}\label{m6bb} $\arfour \vdash \ade z(z\equals x\circ y)$.
%\marginpar{m6bb}
\end{lemma}

\begin{proof} Immediate in view of the results of Section \ref{s17}. \end{proof}

In the sequel, whenever we write $\mathfrak{K}$  within a formula or while reasoning in $\arfour$, it is to be understood as a standard variable-free term representing it. For clarity, let us say that this term is $0$ followed by $\mathfrak{K}$ \ $\successor$s.  We shall implicitly rely on the fact  that $\arfour\vdash\ade z(z\equals \mathfrak{K})$ (Fact \ref{com}). Similarly for $\mathfrak{k}$, as well as $\code{b}$ where $b$ is a symbol. 

\begin{lemma}\label{m8d} $\arfour \vdash   \ade y\mathbb{N}(x,y)$.
%\marginpar{m8d}
\end{lemma}

\begin{proof}   Argue in $\arfour$. By $\arfour$-Induction on $z$, we first want to show 
%\marginpar{m19a}
\begin{equation}\label{m19a}
\mathfrak{K}\mult |z|\mleq |x|\mli \ade y\bigl(|y|\mleq |x|\mlc  \mathbb{N}([x]_{0}^{\mathfrak{K}\imult |z|},y)\bigr).
\end{equation}
The base $\mathfrak{K}\mult |0|\mleq |x|\mli \ade y\bigl(|y|\mleq |x|\mlc  \mathbb{N}([x]_{0}^{\mathfrak{K}\imult |0|},y)\bigr)$ is solved by choosing $0$ (the code of the empty sequence of symbols) for $y$. 
The left inductive step is 
\[\Bigr(\mathfrak{K}\mult |z|\mleq |x|\mli \ade y\bigl(|y|\mleq |x|\mlc  \mathbb{N}([x]_{0}^{\mathfrak{K}\imult |z|},y)\bigr)\Bigr)\mli \Bigr(\mathfrak{K}\mult |z\zero|\mleq |x|\mli \ade y\bigl(|y|\mleq |x|\mlc  \mathbb{N}([x]_{0}^{\mathfrak{K}\imult |z\izero|},y)\bigr)\Bigr).\]
To solve it,  we first figure our whether $z\equals 0$ or $z\notequals 0$. If $z\equals 0$, we ignore the antecedent and do in the consequent the same as what we did in the base case. Otherwise, if $z\notequals 0$, we wait till Environment chooses a constant $a$ for $y$ in the antecedent. Then we compute the value $b$ of $ [x]_{|z|\imult\mathfrak{K}}^{\mathfrak{K}}$ (remember Fact \ref{m6a}). If $b$ is $\code{\hat{c}}$ for some symbol $c$ (which can be established by using Fact \ref{comid} as many times as the number of symbols), then, using Lemma \ref{m6bb}, we compute $a\circ \code{\check{c}}$ and choose the computed value for $y$ in the consequent. Otherwise it does not matter what we choose for $y$, so choose $0$. The right inductive step is similar but simpler, as we do not need to give the case $z\equals 0$ a special consideration. Thus, (\ref{m19a}) is proven.  

Remember that $\mathfrak{K}\equals 2^{\mathfrak{k}}$. 
Let $a$ be the $\mathfrak{k}$th binary predecessor of $|x|$, that is, we have \(|x|\ \equals \ a \mathfrak{b}_1\ldots \mathfrak{b}_{\mathfrak{k}},\) where each $\mathfrak{b}_i$ is either $\zero$ or $\one$. 
Such an $a$ can be found by first computing the value of $|x|$ and then, starting from that value, repeatedly computing binary successor  (the constant) $\mathfrak{k}$ times.   Let  $b$ be the value of the binary predecessor of $2^a$. Note that $|b|\equals a$. Plugging $b$ for $z$ in  (\ref{m19a}), we make this resource compute a value  $c$ for which  $ \mathbb{N}([x]_{0}^{\mathfrak{K}\imult |b|},c)$, i.e. $ \mathbb{N}([x]_{0}^{\mathfrak{K}\imult a},c)$, is true. Now notice that, if $x$ is indeed the code of a sequence of symbols,   $a$ is the number of symbols in that  sequence  and, as the length of the code of each symbol is $\mathfrak{K}$, we have 
$\mathfrak{K}\mult a\equals |x|$; hence, we also have $[x]_{0}^{\mathfrak{K}\imult a}\equals x$; hence, as $ \mathbb{N}([x]_{0}^{\mathfrak{K}\imult a},c)$ is true, so is $ \mathbb{N}(x,c)$. This means that 
we win the target $ \ade y\mathbb{N}(x,y)$ by choosing $c$ for $y$.\vspace{-7pt} 
\end{proof}

\begin{lemma}\label{m6b} $\arfour$ proves each of the following:
%\marginpar{m6b}

1. $\mathbb{C}(x)\mli \ade y\mathbb{I}(x,y)$.

2. $\mathbb{C}(x)\mli \ade y\mathbb{J}(x,y)$.

3. $\mathbb{C}(x)\mli \ade y\mathbb{M}(x,y)$.
\end{lemma}

\begin{proof} Here we only prove clause 1 of the lemma. The remaining clauses are similar. 

Argue in $\arfour$. Let $\mathbb{I}'(x,z)$ be (a natural formalization of) the  predicate 
 ``$x$ is the code of a configuration  of the form (\ref{m66}), and $z\mless i$''.  By $\arfour$-Induction on $z$, we want to prove 
%\marginpar{m6c}
\begin{equation}\label{m6c}
\mathbb{C}(x)\mli \mathbb{I}'(x,|z|) \add \ade y(|y|\mleq |x|\mlc \mathbb{I}(x,y)\bigr).
\end{equation}

The basis is $\mathbb{C}(x)\mli \mathbb{I}'(x,|0|) \add \ade y(|y|\mleq |x|\mlc \mathbb{I}(x,y)\bigr)$. We find $[x]_{\mathfrak{K}} ^{\mathfrak{K}}$. If the latter is $\underline{\hat{a}}$ for some tape symbol $a$, we choose  the right $\add$-disjunct in the consequent and then choose $0$ for $y$. Otherwise we choose the left disjunct. 

The left inductive step is 
%\marginpar{m7b}
\begin{equation}\label{m7b}
\bigl(\mathbb{C}(x)\mli \mathbb{I}'(x,|z|) \add \ade y(|y|\mleq |x|\mlc \mathbb{I}(x,y)\bigr)\mli \bigl(\mathbb{C}(x)\mli \mathbb{I}'(x,|z\zero|) \add \ade y(|y|\mleq |x|\mlc \mathbb{I}(x,y)\bigr). 
\end{equation}
To solve it, we wait till Environment chooses a $\add$-disjunct in   the antecedent. If the left disjunct is chosen,  we find $[x]_{|z\izero|\imult \mathfrak{K}}^{\mathfrak{K}}$. If the latter is $\underline{\hat{a}}$ for some tape symbol $a$, we choose  the right $\add$-disjunct in the consequent of  (\ref{m7b}), and specify $y$ as the value of $|z\zero |$;  otherwise we choose the left $\add$-disjunct there. Suppose now the right disjunct is chosen by Environment in the  antecedent of (\ref{m7b}). We further wait till a constant $c$ is chosen for $y$ there. Then we choose the right $\add$-disjunct in the consequent of  (\ref{m7b}), and specify $y$ as $c$ in it. 
The right inductive step is virtually the same. It is not hard to see that our strategy is successful. 

Now, the target $\mathbb{C}(x)\mli \ade y\mathbb{I}(x,y)$ can be seen to be a logical consequence of (\ref{m6c}), the $\pa$-provable fact $\cla x\cla z\bigl(z\equals |x|\mli \gneg \mathbb{I}'(x,z)\bigr)$ and the $\arfour$-provable sentence $\ada x\ade z(z\equals |x|)$.\vspace{-7pt} 
\end{proof}

\begin{lemma}\label{m8c} $\arfour \vdash \ade y\mathbb{E}(x,y)$.
%\marginpar{m8c}
\end{lemma}

\begin{proof} Argue in $\arfour$. By $\arfour$-Induction on $x$, we want to show $\ade y\bigl(|y|\mleq \mathfrak{K}\mult |x|\mlc \mathbb{E}(x,y)\bigr)$, from which the target $\ade y\mathbb{E}(x,y)$ immediately follows by LC. 

The base $\ade y\bigl(|y|\mleq \mathfrak{K}\mult |0|\mlc \mathbb{E}(0,y)\bigr)$ is solved by choosing $0$ (the code of the empty sequence of symbols) for $y$. For the left inductive step 
\[\ade y\bigl(|y|\mleq \mathfrak{K}\mult |x|\mlc \mathbb{E}(x,y)\bigr)\mli \ade y\bigl(|y|\mleq \mathfrak{K}\mult |x\zero|\mlc \mathbb{E}(x\zero,y)\bigr), \]
we first figure out whether $x\equals 0$ or not. If yes, we ignore the antecedent and act in the consequent in the same way as in the basis case. Otherwise, we
wait till Environment chooses a constant $a$ for $y$ in the antecedent. Then, using Lemma \ref{m6bb}, we compute the value $b$ of $a\circ \code{\check{0}}$ and choose $b$ for $y$ in the consequent. The right inductive step is similar, with the difference that $b$ should be the value of $a\circ \code{\check{1}}$ there; also, the case of $x\equals 0$ does not require a special handling.\vspace{-7pt}  
\end{proof}

\begin{lemma}\label{m6p} $\arfour \vdash \gneg\cle y\mathbb{D}(x,y)\add \ade y\mathbb{D}(x,y)$.
%\marginpar{m6p}
\end{lemma}

\begin{proof} Argue in $\arfour$. By $\arfour$-Induction on $z$, we want to show that 
%\marginpar{m6q}
\begin{equation}\label{m6q}
 \mathfrak{K}\mult |z|\mleq |x|\mli \gneg \cle y\mathbb{D}([x]_{0}^{\mathfrak{K}\imult |z|},y)\add\ade y\mathbb{D}([x]_{0}^{\mathfrak{K}\imult |z|},y ).
\end{equation}
The basis $\mathfrak{K}\mult |0|\mleq |x|\mli \gneg \cle y\mathbb{D}([x]_{0}^{\mathfrak{K}\imult |0|},y)\add\ade y\mathbb{D}([x]_{0}^{\mathfrak{K}\imult |0|},y )$ is obviously solved by choosing the right $\add$-disjunct and specifying $y$ as $0$ in it.  

The left inductive step is 
\[\bigl(\mathfrak{K}\mult |z|\mleq |x|\mli \gneg \cle y\mathbb{D}([x]_{0}^{\mathfrak{K}\imult |z|},y)\add\ade y\mathbb{D}([x]_{0}^{\mathfrak{K}\imult |z|},y )\bigr)\mli 
\bigl( \mathfrak{K}\mult |z\zero|\mleq |x|\mli \gneg \cle y\mathbb{D}([x]_{0}^{\mathfrak{K}\imult |z\izero|},y)\add\ade y\mathbb{D}([x]_{0}^{\mathfrak{K}\imult |z\izero|},y )\bigr).\]
It is solved as follows. If $z\equals 0$, we do in the consequent the same as in the basis case. Suppose now $z\notequals 0$. We wait till Environment selects a $\add$-disjunct in the antecedent (if there is no such selection, we win). If the left disjunct is selected there, we also select the left disjunct in the consequent  and rest our case. Suppose now Environment selects the right disjunct. We wait till it further selects a constant $a$ for $y$ there. Then we compute the value $b$ of $[x]^{\mathfrak{K}}_{\mathfrak{K}\imult |z|}$. If $b\equals \code{\hat{0}}$, we select the right $\add$-disjunct in the consequent, compute  the value $c$ of  $a\zero$, and specify $y$ as $c$ there. If $b\equals \code{\hat{1}}$, we again select the right $\add$-disjunct in  the consequent, compute  the value $c$ of  $a\one$, and specify $y$ as $c$ there. Finally, if neither $b\equals \code{\hat{0}}$ nor $b\equals \code{\hat{1}}$, we select the left disjunct and retire in victory.  
The right inductive step is the same, with the only difference that the case  $z\equals 0$ does not require a special handling there. Obviously our strategy wins, and  (\ref{m6q}) is thus proven.

Now, to solve the target $ \gneg\cle y\mathbb{D}(x,y)\add \ade y\mathbb{D}(x,y)$, we use a trick similar to the one employed in the proof of Lemma \ref{m8d}. Namely, we compute the value $a$ of the $\mathfrak{k}$th binary predecessor of $|x|$, and then the value   $b$ of the binary predecessor of $2^a$. Note that $|b|\equals a$. Taking into account that $\mathfrak{K}\mult a\mleq |x|$, by plugging $b$ for $z$ in  (\ref{m6q}),   we essentially turn  this resource into 
\[\gneg \cle y\mathbb{D}([x]_{0}^{\mathfrak{K}\imult a},y)\add\ade y\mathbb{D}([x]_{0}^{\mathfrak{K}\imult a},y ). \]
This way we would either know that $\gneg \cle y\mathbb{D}([x]_{0}^{\mathfrak{K}\imult a},y)$ is true, or find a constant $c$ for which (we know that) 
$\mathbb{D}([x]_{0}^{\mathfrak{K}\imult a},c )$ is true. Note that if $x$ is the code of a sequence of symbols, then $\mathfrak{K}\imult a$ is nothing but $|x|$ and hence $[x]_{0}^{\mathfrak{K}\imult a}$ is nothing but $x$. With this observation in mind, if $\gneg \cle y\mathbb{D}([x]_{0}^{\mathfrak{K}\imult a},y)$ is true, then choosing the left disjunct of  $ \gneg\cle y\mathbb{D}(x,y)\add \ade y\mathbb{D}(x,y)$ wins this game; and, if $\mathbb{D}([x]_{0}^{\mathfrak{K}\imult a},c )$ is true, then choosing the right disjunct and specifying $y$ as $c$ in it wins the game.\vspace{-7pt}   
\end{proof}

\begin{lemma}\label{m4a}
%\marginpar{m4a}
$\arfour\vdash \mathbb{C}(x)\mli \ade y \mathbb{S}(x,y)$.
\end{lemma}

\begin{proof} Argue in $\arfour$ to justify $\mathbb{C}(x)\mli \ade y \mathbb{S}(x,y)$.  
Assume that $\mathbb{C}(x)$ is true, namely, that $x$ is (the code of) the configuration (\ref{m66}).

We find the state $a$ of   $x$, which is nothing but $[x]_{0}^{\mathfrak{K}}$.  Next, using clauses 1 and 2 of Lemma \ref{m6b}, we find  the locations $i$ and $j$ of the two scanning heads. Then, using these $i$ and $j$, we find the symbols $b$ and $c$ seen by the two heads.  This allows us to find (within $\pa$) the state $d$ of the deterministic successor $y$ of $x$, the symbol $e$ that will overwrite the old symbol $b$ on the work tape, and the directions in which the heads move. We now correspondingly  update $x$  in several steps. First of all, we change the state of $x$ to $d$. Technically, this is done by computing $\code{d}\circ [x]_{\mathfrak{K}}^{|x|\iminus\mathfrak{K}}$. In a similar fashion, details of which are left to the reader, we change the old underlined symbol of the work tape to $e$, and move the two underlines according to the directions in which the corresponding   heads move. In addition, if $a$ is a move state, we find the content of the work tape of the original configuration $x$ up to the location of the work-tape head, and append that content, $\pp$-prefixed and with each symbol $\hat{\mathfrak{s}}$ changed to $\check{\mathfrak{s}}$ using Lemma \ref{m8d}, to the content of the run tape of $x$. Finally, if the previously blank cell $\# m$ of the work tape is no longer blank, we insert a blank cell to the right of it in our representation of the configuration (again, technical details about how, exactly, all this can be done, are left as an easy exercise for the reader). The eventual value of $x$, after the above updates, will be exactly the sought value of the deterministic successor of the original $x$, that is, the value that we should choose for $y$ in the consequent of  $ \mathbb{C}(x)\mli \ade y \mathbb{S}(x,y)$.\vspace{-7pt} 
\end{proof}

\begin{lemma}\label{m4b}
%\marginpar{m4b}
$\arfour\vdash \mathbb{C}(z)\mli  \mathbb{A}'(z,|r|) \add \ade x \mathbb{B}(z,x)   $.
\end{lemma}

\begin{proof} Argue in $\arfour$. By $\arfour$-Induction on $r$, we want to prove 
\[\mathbb{C}(z)\mli \ade x\bigl(|x|\mleq |z|\plus |r| \mlc \mathbb{A}(z,x,|r|)\bigr) \add \ade x \bigl(|x|\mleq (|z|\plus |r|)\zero \mlc  \mathbb{B}(z,x)\bigr),\] 
from which the target $\mathbb{C}(z)\mli  \mathbb{A}'(z,|r|) \add \ade x \mathbb{B}(z,x) $ easily follows by LC.

To solve the base   \(\mathbb{C}(z)\mli \ade x\bigl(|x|\mleq |z|\plus |0| \mlc \mathbb{A}(z,x,|0|)\bigr) \add \ade x\bigl(|x|\mleq (|z|\plus |0|)\zero\mlc  \mathbb{B}(z,x)\bigr)
\), we figure out whether the state of $z$ is a move state or not. If yes, we choose the right $\add$-disjunct; if not, we choose the left $\add$-disjunct. In either case, we further choose the value of $z$ for the variable $x$ and win.  

The left inductive step is 
%\marginpar{m4c}
\begin{equation}\label{m4c}
\begin{array}{l}
\Bigl(\mathbb{C}(z)\mli \ade x\bigl(|x|\mleq |z|\plus |r| \mlc \mathbb{A}(z,x,|r|)\bigr) \add \ade x\bigl(|x|\mleq (|z|\plus |r|)\zero\mlc  \mathbb{B}(z,x)\bigr)\Bigr)\ \mli \\
\Bigl(\mathbb{C}(z)\mli \ade x\bigl(|x|\mleq |z|\plus |r\zero| \mlc \mathbb{A}(z,x,|r\zero|)\bigr) \add \ade x\bigl(|x|\mleq (|z|\plus |r\zero|)\zero\mlc  \mathbb{B}(z,x)\bigr)\Bigr).
\end{array}
\end{equation}
If $r\equals 0$, (\ref{m4c}) is won by solving its consequent in the same ways as the basis case was solved. Suppose now $r\notequals 0$. To solve (\ref{m4c}), we wait till Environment selects one of  the two $\add$-disjuncts in the  antecedent. 

If the right $\add$-disjunct is selected, we wait further till a constant $c$ for $x$ is selected there. Then  we  select the right $\add$-disjunct in the consequent, and choose the same $c$ for $x$ in it.  

Suppose now the left $\add$-disjunct is selected in the antecedent of (\ref{m4c}). Wait further till a constant $c$ for $x$ is selected there. We may assume that $\mathbb{A}(z,c,|r|)$  is true, or else we win the game. Using Lemma \ref{m4a}, we find  the deterministic successor $d$ of the configuration $c$. With a little thought, one can see that the size of $d$ cannot exceed the sum of the sizes of $z$ and $r\zero$ more than twice, so that $|d|\mleq (|z|\plus |r\zero|)\zero$ holds. We figure out whether the state of $d$ is a move state or not. If not, we select the left $\add$-disjunct in the consequent of  (\ref{m4c}), otherwise,  select the right disjunct. In either case, we further choose $d$ for $x$  and win.   

The right inductive step is virtually the same, with the only difference that the case $z\equals 0$ does not require a special handling.\vspace{-7pt} 
\end{proof}

\subsection{Proof of clause (a) of Lemma \ref{m2a}}\label{sA2}
%\marginpar{sA2}
Assume the conditions of clause (a) of Lemma \ref{m2a}. 

First, let us consider the case where $H_i(y,\vec{s})$ is the result of replacing in $E(\vec{s})$ a surface occurrence of a subformula $F_0\adc F_1$ by $F_j$ ($j\equals 0$ or $j\equals 1$). 
Let $\oo\alpha$ be the labmove that brings $E(\vec{s})$ down to $H_i(y,\vec{s})$.  For instance, if $E(\vec{s})$ is $G\mli F_0\adc F_1$ and $H_i(y,\vec{s})$ is $G\mli F_0$, then $\oo\alpha$ is ``$\oo 1.0$''.  

Argue in $\arfour$ to justify $ E^{\circ}_{\circ} (z,\vec{s})  \mli \ade u  H_{i}^{\circ}(u,y,\vec{s})$. 
Assume $E^{\circ}_{\circ} (z,\vec{s})$. This implies that $z$ is a legitimate configuration with yield $E(\vec{s})$, and that the same holds for the deterministic successor $b$ of $z$, which we compute using Lemma \ref{m4a}. Then the result of appending $\oo\alpha$ to the run tape contents of (the configuration encoded by) $b$ is a legitimate configuration with yield $H_i(y,\vec{s})$. The code of such a configuration is $b\circ \code{\check{\oo}\check{\alpha}}$. We compute the value $c$ of the latter using Lemma \ref{m6bb}, and win $ E^{\circ}_{\circ} (z,\vec{s})  \mli \ade u H_{i}^{\circ}(u,y,\vec{s})$ by choosing $c$  for $u$.

Next, consider the  case where $H_i(y,\vec{s})$ is the result of replacing in $E(\vec{s})$ a surface occurrence of a subformula $\ada x F(x)$ by $F(y)$. Let $\alpha$ be the string such that, for any constant $c$, the labmove $\oo\alpha c$ brings $E(\vec{s})$ down to $H_i(c,\vec{s})$. For instance, if $E(\vec{s})$ is $G\mli \ada xF(x)\mld J$ and $H_i(y,\vec{s})$ is $G\mli F(y)\mld J$, then $\alpha$ is ``$1.0.$''; and if $E(\vec{s})$ is just $\ada xF(x)$, then $\alpha$ is the empty string.  

Argue in $\arfour$ to justify $ E^{\circ}_{\circ} (z,\vec{s})  \mli \ade u H_{i}^{\circ}(u,y,\vec{s})$. Let $c$ be the value satisfying $\mathbb{E}(y,c)$. We compute the latter using Lemma \ref{m8c}.     
Assume $E^{\circ}_{\circ} (z,\vec{s})$. This implies that $z$ is a legitimate configuration with yield $E(\vec{s})$. Then the same holds for the deterministic successor $b$ of $z$, which we compute using Lemma \ref{m4a}. Then the result of appending $\oo\alpha y$ to the run tape contents of (the configuration encoded by) $b$ is a legitimate configuration with yield $H(y,\vec{s})$. The code of such a configuration is $b\circ \code{\check{\oo}\check{\alpha}}\circ c$.   We compute the value $d$ of the latter using Lemma \ref{m6bb}, and win $ E^{\circ}_{\circ} (z,\vec{s})  \mli \ade u H_{i}^{\circ}(u,y,\vec{s})$ by choosing $d$  for $u$.

\subsection{Proof of clause (b) of Lemma \ref{m2a}}\label{sA3}
%\marginpar{sA3}
Assume the conditions of clause (b) of Lemma \ref{m2a}. 
In $\arfour$, we can solve  (\ref{m2e}) as follows. Assume  $E^{\circ}  (z,\vec{s})$, which,  of course, implies $\mathbb{C}(z)$. We compute the value  of the binary predecessor of $2^{\chi(|z|)}$, and use that value to specify $r$ in the resource of  
Lemma \ref{m4b}. As a result, we get the resource
$\mathbb{A}'\bigl(z,\chi(|z|)\bigr) \add \ade x  \mathbb{B}(z,x)$. This means that we will either know that $\mathbb{A}'\bigl(z,\chi(|z|)\bigr)$ is true, or find a constant $a$ for which we will know that $\mathbb{B}(z,a)$ is true. 

If $\mathbb{A}'\bigl(z,\chi(|z|)\bigr)$ is true,  then so is $E^{\circ}_{\circ}  (z,\vec{s})$ and, by choosing the latter, we win (\ref{m2e}). 

Now, for the rest of this proof, suppose $\mathbb{B}(z,a)$ is true. Using clause 3 of Lemma \ref{m6b}, we find the number $i$ with $\mathbb{I}(a,i)$. Then we  find (the code of) the move $\alpha$ that $\cal X$ made in $a$. Namely, $\code{\hat{\alpha}}\equals  [a]_{\mathfrak{K}}^{\mathfrak{K}\imult i}$. Using Lemma \ref{m4a}, we also find the deterministic successor $b$ of $a$. Fix these $\alpha$ and $b$. 

Let $\beta_1,\ldots,\beta_m$ be all legal moves in position $E(\vec{s})$ that signify a choice of one of the two $\add$-disjuncts in some surface subformula $F_0\add F_1$ of $E(\vec{s})$. Let $H_1(y,\vec{s}), \ldots, H_m(y,\vec{s})$ be the corresponding $(\oo,y)$-developments of $E(\vec{s})$.

Further, let $\beta_{m+1},\ldots,\beta_{n}$ be all strings such that, any move signifying a choice of a constant $c$ for $x$ in some surface subformula $\ada xF(x)$ of $E(\vec{s})$ looks like $\beta_i c$ for some $i\in\{m\plus 1,\ldots,n\}$. Let $H_{m+1}(y,\vec{s}), \ldots, H_n(y,\vec{s})$ be the corresponding $(\oo,y)$-developments of $E(\vec{s})$. 

To solve (\ref{m2e}), we need to solve its consequent, which now  can  be rewritten as follows:   
%\marginpar{m2ee}
\begin{equation}\label{m2ee}
\begin{array}{r}   E^{\circ}_{\circ}  (z,\vec{s})\add \mathbb{L}\add \ade u\ade y H_{1}^{\circ}(u,y,\vec{s})\add\ldots\add\ade u \ade yH_{m}^{\circ}(u,y,\vec{s})  \add \\ \ade u\ade y H_{m+1}^{\circ}(u,y,\vec{s})\add\ldots\add\ade u \ade yH_{n}^{\circ}(u,y,\vec{s}). \end{array}\end{equation}

This is how we solve (\ref{m2ee}). First, for each $i\in\{1,\ldots,m\}$, we compare $\code{\hat{\alpha}}$ with $\code{\hat{\beta_i}}$. {\bf If} they turn out to be the same, {\bf then} we choose the disjunct $\ade u \ade yH_{i}^{\circ}(u,y,\vec{s})$ in (\ref{m2ee}), and specify $u$ and $y$ as $b$ and $0$ in it, respectively. Here our choice of $0$ for $y$ is arbitrary and has no effect on the game, as $H_{i}^{\circ}(u,y,\vec{s})$ does not contain the variable $y$, anyway.

{\bf Otherwise},  for each $i\in\{m\plus 1,\ldots,n\}$, we compare $[\code{\hat{\alpha}}]_{0}^{r_i}$ with $\code{\hat{\beta}_i}$, where $r_i$ is  the size of $\code{\hat{\beta}_i}$. {\bf If} they turn out to be the same, {\bf then},  we figure out the (code of the) ``rest'' $\gamma$ of the string $\alpha$. That is, $\gamma$ is the string such that $\alpha=\beta_i \gamma$. Employing Lemma \ref{m6p}, we either find a number $c$ with $\mathbb{D}(\code{\hat{\gamma}},c)$, or ($\add$) find out that such a $c$ does not exist ($\gneg \cle$). In the former case, we choose the disjunct $\ade u \ade y H_{i}^{\circ}(u,y,\vec{s})$ in (\ref{m2ee}) and specify $u$ and $y$ as $b$ and $c$ in it, respectively. In the latter case, 
$\alpha$ is an illegal move, so we choose $\mathbb{L}$, which is true because $\cal X$, having made an illegal move, loses. 

{\bf Otherwise}, $\alpha$ is simply an illegal move, so we (again) choose $\mathbb{L}$.

It is left to the reader to convince himself or herself that our strategy succeeds.

\end{document}